\documentclass[11pt]{article}
\usepackage[a4paper,margin=1in]{geometry}
\usepackage{amsmath,amssymb,amsthm}
\usepackage{microtype}
\usepackage{booktabs,tabularx,array,xurl,mathrsfs}
\usepackage{fancyhdr}
\usepackage[hidelinks]{hyperref}
\usepackage[T1]{fontenc}
\usepackage{lmodern}
\numberwithin{equation}{section}
\newtheorem{theorem}{Theorem}[section]
\newtheorem{lemma}[theorem]{Lemma}
\newtheorem{proposition}[theorem]{Proposition}
\newtheorem{corollary}[theorem]{Corollary}
\newtheorem{remark}[theorem]{Remark}
\title{Spectral gaps and slow modes of Pauli rotations and random quantum circuits}
\author{Ziyuan Dong$^{1}$ and Xiang Fan$^{2,*}$\\[0.5ex]
\small $^{1}$Institute of Quantum Computing and Software, School of Computer Science and Engineering,\\
\small Sun Yat-sen University, Guangzhou 510006, China\\
\small $^{2}$School of Mathematics, Sun Yat-sen University, Guangzhou 510275, China\\[0.5ex]
\small E-mail: \href{mailto:dongzy5@mail2.sysu.edu.cn}{dongzy5@mail2.sysu.edu.cn} (Z. Dong);\\
\small \href{mailto:fanx8@mail.sysu.edu.cn}{fanx8@mail.sysu.edu.cn} (X. Fan)\\
\small $^{*}$Corresponding author: Xiang Fan}
\date{}
\hypersetup{pdftitle={Spectral gaps and slow modes of Pauli rotations and random quantum circuits},pdfauthor={Ziyuan Dong and Xiang Fan},pdfsubject={Version 2.5}}

\newcommand{\SU}{\mathsf{SU}}
\newcommand{\PU}{\mathsf{PU}}
\newcommand{\SO}{\mathsf{SO}}
\newcommand{\OO}{\mathsf O}
\newcommand{\Sp}{\mathsf{Sp}}
\newcommand{\Cl}{\mathsf{Cl}}
\newcommand{\symgp}{\mathsf S}

\DeclareMathOperator{\Tr}{Tr}
\newcommand{\C}{\mathbb C}

\begin{document}
\maketitle
\begin{abstract}
We determine which representation sectors control the slowest convergence of
random Pauli rotations and several local random quantum circuits. For
$d=2^n$, the continuous Pauli walk has exact all-representation gaps
\[
 \Delta_{\SU(d)}=\frac{d-8}{8(d-1)}\quad(n\ge4),\qquad
 \Delta_{\PU(d)}=\frac{d(d-3)}{8(d^2-1)}\quad(n\ge3).
\]
The special-unitary gap is attained in $\bigwedge^8\C^d$ and is invisible to
balanced moments, while the projective gap is already attained in a balanced
fourth-moment representation. We obtain exact discrete and orthogonal Pauli analogues, exact finite-Clifford
chain/path gaps, and dimension-independent all-representation Haar brickwork
lower bounds; on three qubits every weighted Haar-$\SU(4)$ edge walk reduces
to one universal principal-angle constant. For real complete-graph models, the entire ordinary
fourth-tensor gap reduces to a quadratic-size occupation operator, is simple
for every $n\ge4$, and has an expansion beginning
\[
 \Delta_n^{(4)}=\frac{10}{9n}-\frac{88}{105n^2}+O(n^{-3}),
\]
with competing branches separated only beyond all algebraic orders. For the
two finite real Clifford complete-graph walks, this fourth-tensor mode is the
exact all-representation group gap for every $n\ge144$; in the same range it
simultaneously governs every ordinary even tensor degree $k\ge4$. Thus
different representation scopes can have genuinely different slow modes.
\end{abstract}


\section{Introduction and main results}\label{sec:introduction-candidate}\label{sec:introduction}

For the random-walk terminology used below, we follow Baer and Haah:
spectral gaps and moment operators are as in
\cite[Definition~1.1]{BH26}, balanced tensor representations and the
projective unitary target as in \cite[Remark~1.5]{BH26}, and the Pauli and
Clifford walk conventions as in \cite[Sections~1.1--2.2]{BH26}.

A representation is irreducible if it has no nonzero proper invariant
subspace. If $z\in Z(G)$ acts in an irreducible representation $\rho$ as
$\rho(z)=\chi_\rho(z)I$, we call $\chi_\rho$ its central character.
We write $\PU(d)=\SU(d)/Z(\SU(d))$.

Tensor moments are standard probes of random-circuit convergence, but the
compact-group gap is the infimum over all finite-dimensional unitary
representations.  A fixed moment is therefore sharp only after slower sectors
outside it are excluded.  The random Pauli walk makes this distinction exact:
for $d=2^n$ and $n\ge4$, its full $\SU(d)$ gap is $(d-8)/(8(d-1))$, attained in
$\bigwedge^8\C^d$, whereas the projective gap is $d(d-3)/(8(d^2-1))$ and is
already attained in the balanced fourth-moment representation
$U\mapsto U^{\otimes4}\otimes\overline U^{\otimes4}$.  Thus the balanced
moment is sharp projectively but misses the full special-unitary slow mode.  This
resolves the exact random-Pauli problem posed in the setting of Baer and
Haah~\cite{BH26}.

The same issue reappears in two forms: nested finite/Haar local fixed spaces need
not have the same global invariants, and the real complete-graph fourth tensor
still has competing symmetry branches. We therefore solve a restricted problem
first and then exclude slower sectors created by enlarging the target or
representation space.

We determine the Pauli gaps and their discrete/orthogonal analogues, then exact
finite Clifford chain/path gaps and dimension-independent Haar brickwork bounds,
including the three-qubit universal principal-angle reduction. Finally we identify
the exact real fourth-tensor branch, derive its all-orders expansion and
exponentially small splitting, and show that in the finite real Clifford
complete graph it becomes the all-representation bottleneck for every
$n\ge144$ and simultaneously attains every ordinary even tensor degree
$k\ge4$.

Throughout, $d=2^n$.  For a probability measure $\nu$ on a compact target $G$,
we use
\begin{equation}\label{eq:gap-convention}
 \Delta(\nu,G)
 =1-\sup_\rho\bigl\|M(\nu,\rho,G)-\Pi(\rho|G)\bigr\|,
 \qquad
 M(\nu,\rho,G)=\int_G\rho(g)\,d\nu(g),
\end{equation}
where $\rho$ ranges over finite-dimensional continuous unitary representations
and $\Pi(\rho|G)$ is Haar averaging.  For finite targets the same definition is
used over all irreducible unitary representations.  A statement labelled
``all-representation'' below refers to this infimum.  A fixed tensor-moment gap
is a different object and is identified explicitly whenever it is used.

\subsection*{Contributions and scope}

Table~\ref{tab:scope-candidate} records the target and representation scope of
each main result; ``exact'' never enlarges either scope implicitly.

\begin{table}[ht]
\centering
\small
\setlength{\tabcolsep}{4pt}
\begin{tabularx}{\textwidth}{@{}>{\raggedright\arraybackslash}p{0.17\textwidth}
>{\raggedright\arraybackslash}p{0.15\textwidth}
>{\raggedright\arraybackslash}p{0.21\textwidth}
>{\raggedright\arraybackslash}X@{}}
\toprule
Model & Target / sector & Scope & Result \tabularnewline
\midrule
Continuous Pauli rotations & $\SU(d)$, $\PU(d)$ & all representations & exact gaps; distinct center-sensitive gap-attaining representations \tabularnewline
Discrete and orthogonal Pauli rotations & generated or projective Clifford; $\SO(d)$ & all representations; all ordinary tensor degrees & exact gaps; first orthogonal degree attaining the group gap is $8$ \tabularnewline
Unitary Clifford brickwork & projective and actual generated targets & all finite-group irreducibles & exact chain gap for the projective and actually generated targets \tabularnewline
Uniform nearest-neighbor sampling & same finite unitary targets & all finite-group irreducibles & exact random-edge path gap; projective and actual corrections distinguished \tabularnewline
Haar brickwork & local $\SU(4)$, $\OO(4)$, $\SO(4)$ & all compact-group irreducibles & explicit dimension-independent lower bounds; on three unitary sites, an exact one-parameter reduction for every edge weighting \tabularnewline
Real complete graph, ordinary fourth tensor & $\mathsf O(4)$-local/$\mathsf{SO}(4)$-local finite Clifford and Haar & entire specified fourth tensor & exact occupation reduction and simple gap eigenvalue; fixed-order asymptotics; internal branch ordering and exponential upper bound \tabularnewline
Finite real complete graph & actual $\mathsf O(4)$-local/$\mathsf{SO}(4)$-local finite Clifford targets & all finite-group irreducibles for $n\ge144$; all ordinary even tensor degrees & exact saturation by $\Delta_n^{(4)}$; threshold uniform in the tensor degree \tabularnewline
\bottomrule
\end{tabularx}
\caption{Scope of the main statements.  ``Exact'' is always relative to the target and representation family shown in the same row.}
\label{tab:scope-candidate}
\end{table}

\subsection{A. Exact Pauli gaps and center-sensitive gap-attaining representations}

For tensor moments we write $\tau_{t,t}(U)=U^{\otimes t}\otimes\overline U^{\otimes t}$.
A random Pauli rotation chooses a nonidentity Hermitian Pauli word $P$ uniformly
and then applies $e^{\mathrm i\theta P}$ with $\theta$ uniform modulo $2\pi$.
Its closed generated group is $\SU(d)$.  Pushing the walk to the projective
quotient removes the central characters that are invisible to balanced moments.

\begin{theorem}[Continuous Pauli rotations]\label{thm:A-candidate}
For $d=2^n$,
\begin{equation}\label{eq:A-candidate}
 \Delta(\nu_{\mathrm{RPR}},\SU(d))
 =\frac{d-8}{8(d-1)}\quad(n\ge4),
 \qquad
 \Delta(\nu_{\mathrm{RPR}},\PU(d))
 =\frac{d(d-3)}{8(d^2-1)}\quad(n\ge3).
\end{equation}
At $n=3$ both values are $5/63$.  For $n\ge4$ the special-unitary value is
attained in $\bigwedge^8\C^d$ by the vector constructed below from affine three-flats.  The projective
value is attained in the irreducible representation of highest weight
$(1^4,0^{d-8},-1^4)$, with ``highest weight'' in the sense of
\cite[Definitions~8.1 and~8.3]{Etingof23}; we record $\mathsf{SU}(d)$ weights
as zero-sum tuples, as specified in the representation conventions below.
It is also attained in every balanced tensor order $t\ge4$.
\end{theorem}

Thus the fourth-balanced-moment value conjectured by Baer and Haah is the exact
\emph{projective} gap, while the full $\SU(d)$ target has a slower central
character for $n\ge4$; at $d=16$, $1/15<26/255$.  The proof pairs a
representation-uniform $\SU(d)$ lower bound with an attaining $\bigwedge^8$
vector.  The projective lower bound needs additional arithmetic from maximal
commuting Pauli sets.

The same framework gives exact discrete and real-orthogonal analogues.  The
$\pi/4$ walk distinguishes its actually generated subgroup from the larger
special Clifford normalizer, whose inaccessible central phases force zero gap
for $n\ge4$.  For odd-$Y$ rotations on $\SO(2^n)$, the exact group gap first
appears in ordinary tensor degree eight for every $n\ge2$.

\subsection{B. Exact Clifford chains and comparison with Haar circuits}

Write $U,O,S$ for the unitary, $\mathsf O(4)$-local, and $\mathsf{SO}(4)$-local models;
$\mathrm C,\mathrm H$ distinguish finite Clifford and Haar gates, and a bar marks
the projective unitary target.  Brickwork multiplies the odd and even layer
projections, whereas random-edge sampling averages adjacent-edge projections.
Accordingly, the two-layer gap is governed by a principal-angle norm rather
than a single self-adjoint edge spectrum.

\begin{theorem}[Unitary finite Clifford chains]\label{thm:B-finite-candidate}\label{thm:clifford-main}
For the projective unitary Clifford brickwork walk on an open chain,
\begin{equation}\label{eq:B-brick-projective}
 \overline\Delta_{\mathrm C,U}(n)
 =1-\frac45\cos\frac\pi n,
 \qquad n\ge2.
\end{equation}
For the actually generated target,
\begin{equation}\label{eq:B-brick-actual}
 \Delta_{\mathrm C,U}(n)=
 \begin{cases}
 1,&n=2,\\[1mm]
 1/2,&n=3,\\[1mm]
 1-(4/5)\cos(\pi/n),&n\ge4.
 \end{cases}
\end{equation}
For uniformly sampled nearest-neighbor edges, the projective all-representation
gap is
\begin{equation}\label{eq:B-path-projective}
 \overline\delta_{\mathrm C,U}^{\rm path}(n)
 =\frac{1-(4/5)\cos(\pi/n)}{n-1},
 \qquad n\ge2.
\end{equation}
For the actual generated target the same formula holds for $n\ge4$, while the
values at $n=2,3$ are $1$ and $1/4$, respectively.
\end{theorem}

Positive edge rates on the projective path reduce exactly to a tridiagonal
Jacobi matrix, from which \eqref{eq:B-path-projective} follows; optimizing
nonuniform rates is not needed here.

The same trigonometric factor occurs in the exact second-moment Haar spectrum
\cite{Deneris25}, but Theorem~\ref{thm:B-finite-candidate} optimizes over all
projective Clifford representations and tracks the additional central-character sector.  The
real chains require separate finite-group estimates because the $\mathsf O(4)$-local and $\mathsf{SO}(4)$-local
ensembles remain distinct.

Because nested finite and Haar local ranges can have different global fixed
spaces, our comparison estimates retain the target projector.  Applied to the
exact Pauli and Clifford estimates, they give:

\begin{theorem}[Haar brickwork bounds]\label{thm:B-haar-candidate}\label{thm:haar-main}
For every $n\ge3$, the all-representation brickwork gaps for local Haar gates
satisfy
\begin{equation}\label{eq:B-haar-bounds}
 \Delta_{\mathrm H,U}(n)>\frac1{40},\qquad
 \Delta_{\mathrm H,O}(n)>\frac1{135},\qquad
 \Delta_{\mathrm H,S}(n)>\frac1{3300}.
\end{equation}
Here $U,O,S$ denote the unitary, $\mathsf O(4)$-local, and
$\mathsf{SO}(4)$-local models, respectively.  The proof is representation-uniform; it does not identify the
exact compact-group gaps.
\end{theorem}

\begin{theorem}[Three-site unitary Haar reduction]
\label{thm:B-three-site-haar}
Let $p,q,r\ge0$, $p+q+r=1$, and let $\nu_p$ choose the three edges of a
three-qubit system with probabilities $p,q,r$ and apply an independent Haar
$\SU(4)$ gate on the chosen edge.  The target is kept fixed as $\SU(8)$, or
as its projective quotient.  There is a weight-independent constant
\begin{equation}\label{eq:B-cstar-range}
 \frac12\le c_*\le\frac{20}{21}
\end{equation}
such that
\begin{equation}\label{eq:B-weighted-haar}
 \Delta(\nu_p,\SU(8))=\Delta(\nu_p,\PU(8))
 =1-\phi_{c_*}(p,q,r),
\end{equation}
where $\phi_c$ is the largest eigenvalue of the $3\times3$ matrix with
diagonal entries $p,q,r$ and off-diagonal entries
$c\sqrt{p_ip_j}$.  Equivalently, it is the largest real root of
\begin{equation}\label{eq:B-weighted-haar-cubic}
 x^3-x^2+(1-c^2)(pq+pr+qr)x
 -(1-c)^2(1+2c)pqr=0.
\end{equation}
The same $c_*$ is the all-representation nontrivial norm of the original
three-site brickwork product.  Consequently
\begin{equation}\label{eq:B-cstar-consequences}
 \Delta_{\mathrm H,U}(3)=1-c_*,\qquad
 \delta_{\mathrm H,U}^{\rm path}(3)=\frac{1-c_*}{2},\qquad
 \delta_{\mathrm H,U}(3)=\frac{2(1-c_*)}{3},
\end{equation}
for brickwork, equal sampling of two supported edges, and uniform sampling of
all three edges, respectively.  Uniform three-edge sampling is the unique
maximizer among all probability weights.
\end{theorem}

The theorem reduces the problem exactly to $c_*$; it does not evaluate that
constant.  In particular,
\begin{equation}\label{eq:B-three-site-bracket}
 \frac1{21}\le\Delta_{\mathrm H,U}(3)\le\frac12,
 \qquad
 1-\phi_{20/21}(p,q,r)
 \le\Delta(\nu_p,\SU(8))
 \le1-\phi_{1/2}(p,q,r).
\end{equation}
Whether $c_*=1/2$ remains open; no fixed-moment assumption enters the reduction.

Small continuous representations explain the qualification.  In
$\operatorname{End}(\bigwedge^4\C^8)$, a principal cosine $1/2$ witnesses
$c_*\ge1/2$, while
higher projective representations could make it larger; in the orthogonal
model, triality similarly beats the finite $\mathsf O(4)$-local fourth-moment value.

\subsection{C. Exact real fourth-moment branch attaining the gap and beyond-all-orders splitting}

We now fix the ordinary fourth tensor of the uniform complete-graph real walk.
For $n\ge4$, the $\mathsf O(4)$-local and $\mathsf{SO}(4)$-local finite Clifford walks and their Haar analogues
have the same entire fourth-tensor gap.  After sectors with a nontrivial one-site component and nonsymmetric sectors
are excluded, it is realized by a symmetric occupation operator on triples
$(a,b,c)$ with $a+b+c=n$, of dimension $\binom{n+2}{2}$.

Concretely, give $V=\mathbb C^3$ the metric $\operatorname{diag}(1,2,1)$,
put $u_0=(1,0,0)$ and $u_\pm=(1,1,\pm1)/2$, and let $P$ be the
orthogonal projection onto $\operatorname{span}\{u_j^{\otimes2}:j=0,+,-\}$.
Let $\Delta_n^{(4)}$ denote this common fourth-tensor gap and put
\begin{equation}\label{eq:C-En}
 L_n=\frac1{n-1}\sum_{i<j}(I-P)_{ij},\qquad
 E_n=\operatorname{gap}(L_n),\qquad
 \Delta_n^{(4)}=\frac{2E_n}{n}.
\end{equation}
The symmetric occupation sector is $\operatorname{Sym}^n\C^3$ with its natural
weighted inner product; we first prove its exclusion of all complementary sectors and select the
true finite-$n$ branch attaining the gap.

\begin{theorem}[Exact real fourth-tensor gap and symmetry of the gap eigenvector]
\label{thm:C-selection-candidate}
For every $n\ge4$, the $\mathsf O(4)$-local and $\mathsf{SO}(4)$-local finite real Clifford walks and their
corresponding Haar walks have the same entire ordinary-fourth-tensor gap
$\Delta_n^{(4)}=2E_n/n$.  The associated largest nonfixed eigenvalue is
algebraically simple on the entire fourth tensor for each of the four
ensembles.  On $\mathscr H_n=\operatorname{Sym}^nV$, let $E_n^+$ and $E_n^-$
be the lowest positive energies in the two eigenspaces of
$\operatorname{diag}(1,1,-1)^{\otimes n}$, and let $E_n^{\rm std}$ be the
lowest positive energy in the standard isotypic component of the internal
$\symgp_3$ permuting $u_0,u_+,u_-$.  Then
\begin{equation}\label{eq:C-ordering}
 E_n=E_n^+<E_n^-\le E_n^{\rm std},
\end{equation}
and the $E_n$-eigenvector in $\mathscr H_n$ is $\symgp_3$-invariant.
\end{theorem}

After this finite-$n$ selection, we expand the true branch and bound its
separation from the competing positive branches.

\begin{theorem}[All fixed orders and beyond-all-orders splitting]
\label{thm:C-asymptotic-candidate}
There is a uniquely determined sequence of rational numbers $(d_r)_{r\ge1}$
such that, for every fixed $R\ge1$,
\begin{equation}\label{eq:C-all-orders}
 \Delta_n^{(4)}
 =\sum_{r=1}^{R}\frac{d_r}{n^r}+O_R(n^{-R-1}).
\end{equation}
The coefficients are obtained from finitely many fixed-excitation-number subspaces $N=k$.  In
particular,
\begin{equation}\label{eq:C-first-terms}
 \Delta_n^{(4)}
 =\frac{10}{9n}-\frac{88}{105n^2}
 -\frac{3463976}{1157625n^3}
 -\frac{18725777197256}{1374129815625n^4}+O(n^{-5}).
\end{equation}
Moreover, with $N_0=3289600$ and $\gamma=(\log4)/N_0$, for every
$n\ge N_0$,
\begin{equation}\label{eq:C-exp-splitting}
 0<E_n^- -E_n
 \le E_n^{\rm std}-E_n
 \le126e^{-\gamma n}.
\end{equation}
No convergence of the formal power series, uniformity in $R$, matching
exponential lower bound, leading exponential coefficient, or optimal decay
exponent is asserted.
\end{theorem}

\begin{corollary}[Finite-group and even-tensor saturation above the effective threshold]
\label{cor:D-complete-saturation}
For the two finite real Clifford uniform complete-graph walks on their actual
generated groups,
\[
 g_O(n)=g_S(n)=\Delta_n^{(4)},\qquad n\ge144.
\]
For every such $n$, every ordinary even tensor degree $k\ge4$ has the same
gap $\Delta_n^{(4)}$; the threshold is independent of $k$.  The value $144$
is sufficient and is not claimed optimal, and no uniform assertion is made
here for $6\le n\le143$.
\end{corollary}

Thus finite-order asymptotics do not select the true branch: the even, odd, and
standard positive branches share every algebraic coefficient, while both
competing branches lie strictly above $E_n$ at finite $n$, with
$E_n^-\le E_n^{\rm std}$, and their residual splitting from $E_n$ is beyond
all algebraic orders.

\subsection*{Related work and priority}

Random circuits and unitary designs are linked to local Hamiltonians for circuit
moment operators~\cite{HarrowLow09,BHH16,HarrowMehraban23}; all-to-all moments
lead to multilevel Lipkin--Meshkov--Glick Hamiltonians~\cite{BV10}, while
overlapping Haar projections enter direct design constructions
\cite{LaRacuenteLeditzky26}. Related work includes $\mathsf{SU}(d)$-symmetric
designs~\cite{LiEtAl24}, moment-independent gaps~\cite{ChenEtAl25}, structured
second moments~\cite{SuzukiEtAl26,Yada26}, and Aldous-type unitary-group
walks~\cite{AlonPuder26}. For the mean-field analysis we use standard
spin-coherent localization and Schur-complement methods of the type developed in
\cite{ManaiWarzel23}. Our model-specific contributions are the
representation-uniform Pauli/Clifford results, the all-orders fourth-moment
recursion and branch selection, and the finite-real complete-graph
saturation by that fourth-tensor slow mode for every $n\ge144$.

Haah, Liu, and Tan noted that balanced tensors do not exhaust $\SU(d)$
representations~\cite{haah2025efficient}; here that warning becomes an exact gap
distinction.  A contemporaneous manuscript of Zhang establishes the four-qubit
$\SU(16)$ value $1/15$ and the same exterior-power slow mode~\cite{Zhang26}.
The present results determine the full $n$-dependence, the projective gap, and
the finite-real complete-graph bottleneck above the effective threshold.

\subsection*{Organization of the paper}

Section~2 fixes targets and comparison tools; Sections~3--5 treat Pauli, finite
Clifford, and Haar brickwork gaps. Sections~6--7 analyze the real fourth tensor
and its asymptotics; Section~8 proves exact finite complete-graph saturation for $n\ge144$.
Appendices~A--C contain the comparison, ordered-pair localization, and effective
fourth-tensor enclosure proofs.

\paragraph{Limitations and open directions.}
We leave open $c_*$, the exact all-representation Haar brickwork and Haar
complete-graph gaps, the finite saturation range below the stated sufficient
threshold, and the leading nonzero asymptotic of the branch splitting.

\section{Models, targets, and comparison tools}\label{sec:models}

\subsection{Representations and Haar averages}\label{subsec:representations}

All representation spaces are finite-dimensional complex Hilbert spaces,
and representations of compact groups are continuous and unitary.
Representations of real matrix groups are also taken on complex Hilbert
spaces; they are not restricted to complexifications of real representations.
For a closed subgroup $H\le G$, set
\[
 \Pi(\rho|H)=\int_H\rho(h)\,d\mu_H(h),\qquad
 V^H=\{v:\rho(h)v=v\text{ for all }h\in H\}.
\]
Haar invariance and inversion show that $\Pi(\rho|H)$ is the orthogonal
projection onto $V^H$. If $H\le H'$, then
$\Pi(\rho|H')\preceq\Pi(\rho|H)$. We suppress $\rho$ when the
representation is fixed. The order $A\preceq B$ for self-adjoint operators
means $\langle v,Av\rangle\le\langle v,Bv\rangle$ for every $v$.
When an inequality is specified only on a subspace, it is a quadratic-form
inequality there; invariance of that subspace is not implicit.

The representation gap is
$\Delta(\nu,\rho,G)=1-\|M(\nu,\rho,G)-\Pi(\rho|G)\|$.
Equation~\eqref{eq:gap-convention} is its infimum over $\rho$.
A bound for one representation is therefore an upper bound for the group
gap, whereas a group lower bound must apply uniformly to every
representation. For a finite group, the regular representation on $\mathbb C[G]$ by left
translation contains all irreducibles; for a compact group the same comparison is applied
one finite-dimensional irreducible at a time. Multiplicity does not change
an operator norm on an orthogonal direct sum.

For $\mathsf{SU}(d)$, highest weights are recorded as zero-sum real tuples.
More precisely, $\lambda$ is a highest weight if a nonzero vector $v$
satisfies
\[
 D\rho(\operatorname{diag}(h_1,\ldots,h_d))v
   =\Bigl(\sum_i\lambda_i h_i\Bigr)v,
 \qquad D\rho(E_{ij})v=0\quad(i<j),
\]
for all complex $h_i$ of sum zero. Here $D\rho$ is the complex-linear
extension of the differential and $E_{ij}$ is a matrix unit
\cite[Propositions~5.4--5.5]{Hall00}\cite[Definitions~8.1 and~8.3]{Etingof23}.
This convention specifies the attaining representation in
Theorem~\ref{thm:pauli-continuous}; no highest-weight classification is used
in defining the gap.

\subsection{Pauli groups, circles, and discrete rotations}\label{subsec:pauli-models}

Write $X,Y,Z$ for the Hermitian Pauli matrices and
$\mathcal P_n=\{I,X,Y,Z\}^{\otimes n}\setminus\{I_d\}$.
The phase-forgotten labels form $\mathbb F_2^{2n}$.
Writing a label as $(x,z)$, the commutation sign is determined by
\[
 \lambda((x,z),(x',z'))=x\cdot z'+z\cdot x'.
\]
The two representatives of opposite sign define the same rotation circle
and the same cyclic subgroup. We use only one representative per
nonidentity label.

For $\Theta_\infty=\mathbb R/(2\pi\mathbb Z)$ and
$\Theta_8=\{m\pi/4+2\pi\mathbb Z:m=0,\ldots,7\}$, let
\begin{equation}\label{eq:pauli-projections}
 \Pi(\rho|\Theta,P)=
 \int_\Theta\rho(e^{\mathrm i\theta P})\,d\mu_\Theta(\theta),\qquad
 \mathcal H_\Theta^\rho=\sum_{P\in\mathcal P_n}\bigl(I-\Pi(\rho|\Theta,P)\bigr).
\end{equation}
These are subgroup averages and sums of orthogonal projection defects.
In particular,
\[
 M(\nu_{\mathsf{RPR}},\rho)=I-\frac{\mathcal H_{\Theta_\infty}^\rho}{d^2-1},
 \qquad
 M(\nu_{\mathsf{RPCR}},\rho)=I-\frac{\mathcal H_{\Theta_8}^\rho}{d^2-1}.
\]
Both moment operators are positive semidefinite. Their gaps, relative to
the actual generated group, equal the smallest nonzero eigenvalue of
the corresponding defect sum divided by $d^2-1$, with the representation
containing no nonfixed vectors treated by the norm definition.
For the real walk replace $\mathcal P_n$ by
$\mathcal Y_n=\{P\in\mathcal P_n:P\text{ has an odd number of }Y\text{ factors}\}$
and the denominator by $|\mathcal Y_n|=d(d-1)/2$.
The Lie algebra directions $\mathrm iP$, for the respective two label
sets, span $\mathfrak{su}(d)$ and $\mathfrak{so}(d)$.
Consequently the continuous targets are $\mathsf{SU}(d)$ and
$\mathsf{SO}(d)$, respectively.

Let $\mathcal P_n^{\mathrm{sp}}=\langle\mathrm iI_d,X_j,Z_j:1\le j\le n\rangle
\cap\mathsf{SU}(d)$. For $H\le G$ we use
$N_G(H)=\{g\in G:gHg^{-1}=H\}$. Define
\begin{equation}\label{eq:unitary-targets}
 \Cl(n)=N_{\mathsf{SU}(d)}(\mathcal P_n^{\mathrm{sp}}),\quad
 K_n=\langle g_P:P\in\mathcal P_n\rangle,\quad
 g_P=e^{\mathrm i\pi P/4},\quad
 \overline{\Cl}(n)=\Cl(n)/Z(\mathsf{SU}(d)).
\end{equation}
This is the special-normalizer convention of \cite[Definition~1.2]{BH26}.
The normalizer is finite: its conjugation action has finite image on the
Pauli group, and its kernel is the finite scalar center of
$\mathsf{SU}(d)$. The projective Pauli subgroup
$\overline{\mathcal P}_n\simeq\mathbb F_2^{2n}$ is normal in
$\overline{\Cl}(n)$, with
$\overline{\Cl}(n)/\overline{\mathcal P}_n\simeq\Sp(2n;\mathbb F_2)$
\cite{BH26,NRS00}. Thus $\overline{\Cl}(n)$ is not itself the symplectic quotient.

\begin{proposition}[Generated groups and central phases]\label{prop:unitary-groups}
Let $\mu_r=\{\zeta I:\zeta^r=1\}$. Then
\begin{equation}\label{eq:central-structure}
 \Cl(n)=Z(\mathsf{SU}(d))K_n,\qquad
 K_n\cap Z(\mathsf{SU}(d))=\mu_{2^{\min(n,3)}}.
\end{equation}
In particular, $K_n=\Cl(n)$ for $n\le3$ and
$[\Cl(n):K_n]=2^{n-3}$ for $n\ge4$.
For $n\ge2$, the group generated by adjacent embedded copies of the
standalone two-qubit group $\Cl(2)$ is exactly $K_n$.
\end{proposition}

\begin{proof}
The square of $g_P$ is $\mathrm iP$, so its projective image contains the
projective Pauli element with label $P$.
Conjugation by $g_P$ induces the symplectic transvection
$x\mapsto x+\lambda(x,p)p$ on labels. Transvections generate the symplectic
quotient. Together with the Pauli subgroup this proves that the projective
image of $K_n$ is all of $\overline{\Cl}(n)$, giving the first equality in
\eqref{eq:central-structure}.

We give a field argument for the upper bound on the scalar subgroup; it
does not use an attaining representation or any gap formula. Every entry of
$g_P=(I+\mathrm iP)/\sqrt2$ lies in $F=\mathbb Q(\zeta_8)$, where
$\zeta_8=e^{\mathrm i\pi/4}$. The same holds for every word in these
generators and their inverses. If a scalar $\zeta I$ lies in $K_n$, then
$\zeta\in F$ and $\zeta^{2^n}=1$. A primitive root of order $2^r$ has
degree $2^{r-1}$ over $\mathbb Q$ for $r\ge1$, whereas $[F:\mathbb Q]=4$.
Thus $r\le3$. The degree formula follows, for example, by applying
Eisenstein's criterion at $2$ to $(x+1)^{2^{r-1}}+1$ for $r\ge2$.
Together with the determinant condition this bounds the scalar subgroup
by $\mu_{2^{\min(n,3)}}$.

For the reverse inclusion, $g_P^4=-I$ gives $\mu_2$.
When $n\ge2$, choose $P=Z_1$, $Q=Z_2$; then
\[
 g_P^2g_Q^2g_{PQ}^2=(\mathrm iP)(\mathrm iQ)(\mathrm iPQ)=-\mathrm iI.
\]
For $n\ge3$ the diagonal Paulis commute and
\[
 \prod_{p\in\mathbb F_2^n\setminus\{0\}}
       e^{\mathrm i\pi Z(p)/4}=e^{-\mathrm i\pi/4}I.
\]
Indeed, the sum of their diagonal eigenvalues is $d-1$ at the zero
computational label and $-1$ at every other label, and $8$ divides $d$.
These identities give the required lower inclusions. The assertions about
$\Cl(n)$ and its index now follow from $|Z(\mathsf{SU}(d))|=d$.

Let $K_n^{\rm adj}$ be generated by adjacent copies of $\Cl(2)$.
We have already proved $\Cl(2)=K_2$. Each embedded $K_2$ generator is an
$n$-qubit $g_P$, hence $K_n^{\rm adj}\le K_n$.
Conversely, adjacent Clifford gates have conjugation actions realizing
single-site Hadamard and phase gates, controlled-NOT gates $\lvert a,b\rangle\mapsto\lvert a,a\oplus b\rangle$, and adjacent
swaps. Determinant-one representatives differ only by scalars and so have
the same conjugation actions. These operations carry $Z_1$ to any nonzero
Pauli label, up to sign: local gates first change nonzero factors to $Z$,
and controlled-NOT gates and swaps remove or move their support.
As $g_{Z_1}\in K_n^{\rm adj}$, conjugation and inversion put every $g_P$ in $K_n^{\rm adj}$.
Thus $K_n\le K_n^{\rm adj}$.
\end{proof}

The usual transvection and Clifford quotient facts in the first paragraph
are group-structural facts, not consequences of any spectral formula.
The field argument makes the central calculation independent of the
later exterior-power construction.

For $n\ge4$, if either the discrete Pauli walk or the unitary Clifford brickwork walk is instead compared with the larger target $\Cl(n)$, its gap is zero: $K_n\lneq\Cl(n)$ is normal and the nontrivial quotient has a character fixed by every step but not by Haar averaging on $\Cl(n)$.

\subsection{The two real local-gate ensembles}\label{subsec:real-models}

The real Pauli group is
\[
 \mathcal P(n;\mathbb R)
 =\{\varepsilon X^{x_1}Z^{z_1}\otimes\cdots\otimes X^{x_n}Z^{z_n}:
       \varepsilon\in\{1,-1\},\ x,z\in\mathbb F_2^n\}.
\]
Unlike the Hermitian representatives, these matrices need not be symmetric.
Their squares are $(-1)^{\phi(x,z)}I$, where $\phi(x,z)=x\cdot z$.
Its polar form is the symplectic form $\lambda$.
Define
\[
 C_2^O:=N_{\mathsf O(4)}(\mathcal P(2;\mathbb R)),
 \qquad
 \Cl(2;\mathbb R)
 :=N_{\mathsf{SO}(4)}(\mathcal P(2;\mathbb R))
 =C_2^O\cap\mathsf{SO}(4).
\]
For $n\ge3$, the groups generated by adjacent embedded copies of either
local ensemble will be identified below with the common global normalizer
$\Cl(n;\mathbb R)$. Each finite gap is relative to the group generated by
its stated local ensemble; no equality of local averaging operators is
inferred from equality of the generated global groups.

The $\mathsf O(4)$-local quotient is $\OO^+(4;\mathbb F_2)$, of order $72$.
The $\mathsf{SO}(4)$-local quotient is its order-$36$ subgroup preserving the two
anisotropic triples, and is isomorphic to $\symgp_3\times \symgp_3$.
This character is defined by the physical determinant, not by the binary
label action. To see the distinction, an embedded single-qubit Hadamard
$H\otimes I_2$ has physical determinant $+1$, but its binary label action
has $\operatorname{rank}(I+H)=1$.
The six anisotropic labels in the split four-space form two triples:
within each triple the polar pairing between distinct vectors is one,
and between triples it is zero. The corresponding skew Pauli matrices
span the two simple ideals of $\mathfrak{so}(4)$.
A determinant-positive orthogonal transformation preserves these ideals,
while a determinant-negative one interchanges them.
The local generators and this character are also checked directly using
their four-dimensional matrix actions in the finite verification material.

For the continuous ensembles, an edge gate belongs to
$\mathsf{SU}(4)$, $\mathsf O(4)$, or $\mathsf{SO}(4)$, with Haar probability
measure in the indicated local group. We use $n\ge3$ when comparing all
three ensembles on the common form of their global targets.
For a real edge gate $U$,
\[
 \det(U\otimes I_{2^{n-2}})=(\det U)^{2^{n-2}}=1\quad(n\ge3).
\]
At $n=2$ the $\mathsf O(4)$-local target is $\mathsf O(4)$, not
$\mathsf{SO}(4)$.

\begin{proposition}[Real generation and continuous targets]\label{prop:real-generation}
For $n\ge3$, both real finite local-gate ensembles generate
\[
 \Cl(n;\mathbb R)
 =N_{\mathsf O(d)}(\mathcal P(n;\mathbb R))
 =N_{\mathsf{SO}(d)}(\mathcal P(n;\mathbb R)).
\]
The closed groups generated by the two continuous real ensembles are
$\mathsf{SO}(d)$. The continuous unitary ensemble generates $\mathsf{SU}(d)$.
\end{proposition}
\begin{proof}
The real normalizer in $\mathsf O(d)$ is generated by the real Pauli group,
single-site Hadamards, and controlled-NOT gates; its Pauli quotient is
$\OO^+(2n;\mathbb F_2)$ \cite{NRS00}. For $n\ge3$ these generators all have determinant
one, so the two displayed normalizers agree.
Adjacent $\mathsf O(4)$-local finite Clifford gates contain their usual adjacent generators, and
adjacent swaps move their supports. Thus they generate the full normalizer.

The $\mathsf{SO}(4)$-local finite Clifford gates also generate each adjacent controlled-NOT gate. Write $C_{ij}$
for the controlled-NOT with control $i$ and target $j$, and $H_i$ for a
normalized Hadamard. On three sites the following is an exact physical
matrix identity, with the rightmost factor acting first:
\begin{equation}\label{eq:strict-word}
\begin{split}
 C_{12}={}&(C_{23}H_2C_{32})(H_1C_{12}C_{21})
          (C_{32}H_2C_{23})\\
         &\hspace{12mm}\cdot(C_{12}C_{21})(C_{32}H_2C_{23}).
\end{split}
\end{equation}
Each parenthesized factor is a two-site real Clifford gate with local
determinant $+1$. The identity is checked by multiplying the eight-by-eight
matrices: after clearing the four Hadamard denominators, the integer
product is $4C_{12}$. Reversing the sites gives the other edge.
Every adjacent edge lies in a three-site window; the $\mathsf{SO}(4)$-local finite Clifford groups also
contain the required Hadamards and Paulis. This proves the finite assertion.
The exact integer check of \eqref{eq:strict-word} is included with the
finite verification material.

For the continuous real case, the generated group contains this finite
normalizer and the circle generated by $\mathrm iY_1$.
The orthogonal label group is transitive on the anisotropic labels, so
finite conjugation supplies every odd-$Y$ Pauli circle.
Their infinitesimal generators span $\mathfrak{so}(d)$.
The closed generated group therefore contains the connected group
$\mathsf{SO}(d)$ and is contained in it by the determinant calculation above.
For unitary gates, the same argument uses $K_n$, its transitive action on
nonzero Pauli labels, and the span $\mathfrak{su}(d)$.
\end{proof}

In the source convention of \cite{BH26}, the real local normalizer is
defined inside $\mathsf{SO}(4)$, whereas the $\OO^+(4;\mathbb F_2)$ quotient used in the
real local calculation has order $72$. Keeping $\Cl(2;\mathbb R)$ and $C_2^O$ separate
resolves this distinction without changing the prescribed distributions.
Our bounds are stated independently for both choices. In particular, a
calculation using the $72$-element quotient is not a proof of the same
local bound for the $36$-element $\mathsf{SO}(4)$-local quotient.

\subsection{One step, two projections, and powers}\label{subsec:layer-conventions}

For an edge $i,i+1$, let $P_i$ denote its local finite-group average in a
fixed representation. The layer projections and step operator are
\begin{equation}\label{eq:layers}
 O=\prod_{i\ \mathrm{odd}}P_i,\qquad
 E=\prod_{i\ \mathrm{even}}P_i,\qquad B=EO.
\end{equation}
Products within a layer are orthogonal projections, since disjoint edge
actions commute. Empty products are the identity. Their intersection is
the fixed space of the generated group, with projection $\Pi$.
Continuous averages are defined in the same way and may instead be
denoted by $Q_i,Q_o,Q_e$ when comparing the two ensembles.
Then $Q_i\le P_i$ and $Q_\ell\le P_\ell$ for each layer, but the finite
and continuous global projections are generally different.

\begin{lemma}[Two-projection norm convention]\label{lem:two-projection-powers}
Let $P,Q$ be orthogonal projections, let $R$ project onto their range
intersection, and put $c=\|QP-R\|$. For every integer $k\ge1$,
\begin{equation}\label{eq:projection-powers}
 \|(QP)^k-R\|=c^{2k-1}.
\end{equation}
If $c>0$, the largest nontrivial eigenvalue of $QP$ is $c^2$.
Moreover,
\begin{equation}\label{eq:layer-H-bound}
 (I-P)+(I-Q)\succeq(1-c)(I-R).
\end{equation}
\end{lemma}
\begin{proof}
The finite-dimensional two-projection decomposition consists of their
common range, commuting one-dimensional blocks, and two-dimensional blocks
on which, in an orthonormal basis,
\[
 P=\begin{pmatrix}1&0\\0&0\end{pmatrix},\qquad
 Q=\begin{pmatrix}s^2&s\sqrt{1-s^2}\\s\sqrt{1-s^2}&1-s^2\end{pmatrix},
 \quad 0<s<1.
\]
On such a block $\|QP\|=s$ and $(QP)^k=s^{2k-2}QP$; its nonzero eigenvalue
is $s^2$. The sum of defects has eigenvalues $1-s$ and $1+s$.
On the commuting blocks outside the common range, $QP=0$ and the defect
sum is either $1$ or $2$. Taking maxima or minima as appropriate proves
both claims. If no nontrivial two-dimensional block occurs, then $c=0$
and the same formulas hold.
\end{proof}

The power formula remains valid after taking the supremum over all
representations, since $s\mapsto s^{2k-1}$ is continuous and increasing on
$[0,1]$. In particular, the gap in this paper is $1-\sqrt\lambda$ when a
comparison source denotes the largest nontrivial two-layer eigenvalue by
$\lambda$. This conversion applies to products of two orthogonal
projections, not to arbitrary non-Haar local averages.

\begin{lemma}[Spectator endpoint averages preserve the layer product]\label{lem:endpoint-completion}
Suppose subgroup averages $S,T$ satisfy $ES=SE=E$, $ST=TS$,
$SO=OS$, $OT=TO=O$, and $TE=ET$.
Then $(TE)(OS)=EO$.
\end{lemma}
\begin{proof}
We have $TEOS=ETOS=EOS=ESO=EO$.
\end{proof}
In an odd-site open chain, $S$ averages a single-site subgroup on the
endpoint untouched by the odd layer and contained in the even-layer group;
$T$ averages the other endpoint. The hypotheses follow from containment
and disjoint supports. Completing both layers in this way preserves the
original product in every representation, although it changes the
individual layer ranks. It does not authorize an arbitrary enlargement
of a layer subgroup.

\subsection{Weighted graph walks and single-site fixed spaces}\label{subsec:edge-models}

Let $\mathcal E$ be a graph on the $n$ sites, with edge probabilities
$p_e\ge0$ satisfying $\sum_{e\in\mathcal E}p_e=1$.
Delete the zero-probability edges. Unless stated otherwise, the resulting
support graph has no isolated vertex, and its target $G_{\mathcal E,p}$
is the closed group generated by the supported local gate groups. For a
disconnected graph this definition does not silently replace the target
by the full $n$-qubit group. In a representation of this target, set
\begin{equation}\label{eq:weighted-mean-normalization}
 T_p=\sum_e p_eP_e,\qquad H_p=\sum_e p_e(I-P_e)=I-T_p.
\end{equation}
The mean $T_p$ is positive and self-adjoint. On a nonzero complement of the
generated-group fixed space, its norm gap is the smallest eigenvalue of
$H_p$; on a zero-dimensional complement the norm gap is one.
For arbitrary positive rates $w_e$, the unnormalized Hamiltonian
$\sum_e w_e(I-P_e)$ must instead be divided by $\sum_e w_e$ to obtain the
one-step probability normalization. In particular, a uniform open path
has $n-1$ edges and a uniform complete graph has $\binom n2$ edges.

On the open path, within-layer commutation gives
\begin{equation}\label{eq:edge-layer-comparison}
 H_{\mathrm{path}}:=\sum_{i=1}^{n-1}(I-P_{i,i+1})
 \succeq(I-O)+(I-E)
 \succeq\Delta_{\mathrm{brick}}(I-\Pi),
\end{equation}
whenever $\Delta_{\mathrm{brick}}$ is a representation-uniform lower bound
for the two-layer gap. Dividing by $n-1$ gives a random-edge lower bound,
not an exact formula. If all permuted paths generate the same target,
averaging over site permutations gives the same normalized lower bound
on the complete graph: a fixed unordered pair occurs in a random path
with probability $2/n$. This comparison keeps the same invariant projector
and does not have a converse.

Choose a single-site subgroup $\mathsf S_i$ at each site which is
contained in every incident local group, and let $R_i$ be its Haar
projection. For finite unitary gates one may use the embedded single-site
special Clifford group, and for Haar unitary gates one may use
$\mathsf{SU}(2)$. In the real models one may use the full single-site
real Clifford group, or $\mathsf O(2)$ for Haar gates. These real choices
are legal also for $\mathsf{SO}(4)$-local gates, because
$\det(U\otimes I_2)=(\det U)^2=1$.
Containment and disjoint supports imply
\[
 R_iR_j=R_jR_i,\qquad R_iP_e=P_eR_i,\qquad P_{ij}\le R_iR_j.
\]
Write
\[
 c_p(S)=\sum_{e:e\cap S\ne\varnothing}p_e,\qquad
 d_i=\sum_{e\ni i}p_e,\qquad d_{\min}=\min_i d_i.
\]
Here $c_p(S)$ counts the probability of selecting an edge meeting $S$;
each edge is counted once even if both endpoints lie in $S$.

\begin{proposition}[Weighted single-site reduction]\label{prop:single-site-reduction}
Let $V_S$ be the joint sector with $R_i=0$ for $i\in S$ and $R_i=1$
otherwise. Then
\begin{equation}\label{eq:sector-edge-bound}
 0\preceq T_p|_{V_S}\preceq(1-c_p(S))I.
\end{equation}
Let $\Delta_{\varnothing}(p)$ be one minus the supremum, over all
representations of $G_{\mathcal E,p}$, of the norm of $T_p-\Pi$ on
$V_{\varnothing}$. The norm on a zero-dimensional nonfixed restriction
is zero. For each of the finite and continuous local ensembles in this
paper, the whole-group gap satisfies
\begin{equation}\label{eq:weighted-graph-reduction}
 \Delta(p,G_{\mathcal E,p})
       =\min\{\Delta_{\varnothing}(p),d_{\min}\}.
\end{equation}
Moreover, in the conjugation representation on
$\operatorname{End}(\mathbb C^{2^n})$, every Pauli basis element $R_v$ has
\begin{equation}\label{eq:pauli-adjoint-graph}
 T_pR_v=(1-c_p(\operatorname{supp}v))R_v.
\end{equation}
For a fixed support $S$, this gives $3^{|S|}$ basis vectors; multiplicities
are added when different supports give the same eigenvalue.
\end{proposition}
\begin{proof}
If $e$ meets $S$, then $P_e$ vanishes on $V_S$ because
$P_{ij}\le R_iR_j$. The other summands are positive contractions, proving
\eqref{eq:sector-edge-bound}. For nonempty $S$ and any $i\in S$ we have
$c_p(S)\ge d_i\ge d_{\min}$. Every global invariant lies in
$V_{\varnothing}$, since the support graph has no isolated vertex and its
edge groups contain all $\mathsf S_i$. Thus the orthogonal joint-sector
decomposition gives the lower bound by the minimum in
\eqref{eq:weighted-graph-reduction}.

For the reverse bound, use the full conjugation representation
$X\mapsto gXg^{-1}$; no irreducibility of this representation is needed.
If $e$ is disjoint from the support of $R_v$, then $P_eR_v=R_v$.
If $e$ meets that support, a Pauli operator supported on that edge
anticommutes with $R_v$. This operator belongs to each of the relevant
finite or continuous local groups. Right invariance of the local Haar
average therefore gives $P_eR_v=-P_eR_v=0$. This proves
\eqref{eq:pauli-adjoint-graph}.
A traceless single-site Pauli at a vertex of minimum degree has eigenvalue
$1-d_{\min}$, belongs to a nonempty single-site sector, and is not globally
invariant because $d_{\min}>0$. Hence the whole-group gap is at most
$d_{\min}$; it is also at most $\Delta_{\varnothing}(p)$ by restriction.
Both inequalities together prove the equality. Counting the three
nonidentity one-site Pauli choices proves the multiplicity statement.
\end{proof}

The exact reduction uses a witness present in the whole-group problem.
An individual representation lacking that witness need only satisfy the
corresponding lower estimate. The Pauli basis in
\eqref{eq:pauli-adjoint-graph} describes the full matrix-conjugation
representation, not just the Lie-algebra adjoint representation of an
orthogonal group. Neither assertion covers every group representation by
itself. Any central obstruction surviving the single-site averages stays
in $V_{\varnothing}$.
If the support has isolated vertices, remove them and adjust the target
and fixed projector before using \eqref{eq:weighted-graph-reduction}; an
untouched-site Pauli would otherwise be a fixed vector rather than a gap
witness. Reweighting edges or retaining a larger target is a different
spectral-gap problem and is stated explicitly whenever used.

\subsection{Projection-comparison estimates used below}\label{subsec:comparison-principles}

The circuit arguments later in the paper use several elementary but reusable
comparison estimates for families of orthogonal projections.  We state them
here before their first application, with the representation-uniform
hypotheses and target projectors explicit.  Their proofs,
including the orthogonal decomposition that underlies the angle estimates, are
collected in Appendix~\ref{sec:comparison}.  For two subspaces $M,N$ with
$S=M\cap N$, write
\[
 c(M,N)=\|P_MP_N-P_S\|
\]
for their Friedrichs cosine.  We also set
\[
 \Psi(x)=1-\frac{2}{1+\sqrt{1+2x}}.
\]

The following bound passes from local principal angles to angles between sums of
commuting local ranges.

\begin{theorem}[Bipartite pairwise comparison]\label{thm:bipartite-angle}
Let $(M_i)_{i=1}^p$ and $(N_j)_{j=1}^q$ be two families whose projections
commute within each family, and put $M=\sum_iM_i$ and $N=\sum_jN_j$.
Let $C=(c_{ij})$ be a nonnegative real matrix such that
$c(M_i,N_j)\le c_{ij}$ for every $i,j$. Then
\begin{equation}\label{eq:bipartite-angle}
 c(M,N)\le\min\{1,\|C\|\}.
\end{equation}
\end{theorem}

The next result compares nested finite and continuous fixed-space projections while
allowing their global invariant spaces to differ.

\begin{theorem}[Two nested layers]\label{thm:nested-product}
Let $Q_o\le P_o$ and $Q_e\le P_e$ be orthogonal projections. Let $A$
project onto $\operatorname{ran}P_o\cap\operatorname{ran}P_e$, and let
$R$ project onto $\operatorname{ran}Q_o\cap\operatorname{ran}Q_e$.
Suppose that on $R^\perp$,
\[
 \|P_eP_o-A\|\le c<1,\qquad
 AQ_\ell A\preceq(1-a)A\quad(\ell=o,e),\qquad 0<a\le1.
\]
Then
\begin{equation}\label{eq:nested-product}
 \|Q_eQ_o-R\|\le1-a(1-c).
\end{equation}
The coefficient is sharp under exactly these projection hypotheses.
\end{theorem}

For two-color chains we use the following energy-to-angle conversion and a weighted
path comparison.

\begin{theorem}[Two-color energy-to-angle bound]\label{thm:two-colour-energy}
Let $h_i$ and $k_j$ be orthogonal projections, commuting within each
family, and set
\[
 O=\prod_i(I-h_i),\qquad E=\prod_j(I-k_j),\qquad
 H=\sum_i h_i+\sum_j k_j.
\]
Let $A$ project onto $\ker H$ and suppose
$H\succeq\gamma(I-A)$ with $\gamma>0$.
For each $k_j$ suppose at most $d_E$ of the $h_i$ fail to commute with
it, and for each $h_i$ suppose at most $d_O$ of the $k_j$ fail to commute
with it. Put $D=(d_E+d_O)/2$. If $D>0$, then
\begin{equation}\label{eq:two-colour-energy}
 1-\|EO-A\|\ge
 \Psi_D(\gamma):=1-\frac2{1+\sqrt{1+4\gamma/D}}.
\end{equation}
If every cross pair commutes, then $EO=A$.
For the two colors of an open path one may take $D=2$, giving
$\Psi_2=\Psi$.
\end{theorem}

\begin{lemma}[Path comparison from adjacent-edge bounds]\label{lem:weighted-path-cover}
Let $n\ge3$, and let $(P_i^S)$ and $(P_i^O)$, $1\le i\le n-1$, be two edge-projection
families with the same common-range projector $A$, and with commutation
within each parity class. Write $h_i^X=I-P_i^X$. Suppose
\[
 h_i^S+h_{i+1}^S\succeq\beta(h_i^O+h_{i+1}^O)
 \quad(1\le i<n-1),\qquad \beta>0.
\]
Let $O_O,E_O$ be the two parity products of the $P_i^O$, and suppose
$\|E_OO_O-A\|\le c<1$. With $t_n=1$ for odd $n$ and
$t_n=(n-2)/n$ for even $n$, one has
\begin{equation}\label{eq:weighted-path-cover}
 \sum_i h_i^S\succeq
 \frac{\beta}{2}\left(1+t_n-
          \sqrt{(1-t_n)^2+4t_nc^2}\right)(I-A).
\end{equation}
\end{lemma}

\section{Exact Pauli gaps and tensor-degree attainment}\label{sec:pauli-proofs}

Here the torus weight lattice means the lattice of characters occurring in
the restriction of a representation to the relevant maximal torus.

We prove the continuous and discrete Pauli assertions of
Theorems~\ref{thm:pauli-continuous} and~\ref{thm:pauli-discrete}.
The complex and real walks share their local three-direction inequalities,
but their torus weight lattices are different. We keep that distinction
explicit. The final subsection determines every ordinary tensor gap of
the orthogonal walk, including the first tensor degree attaining the
group gap.

\begin{theorem}[Continuous Pauli rotations]\label{thm:pauli-continuous}
The complex gaps are
\begin{equation}\label{eq:pauli-complex-main}
\begin{aligned}
 \Delta(\nu_{\mathsf{RPR}},\mathsf{SU}(d))
 &=\frac{d-8}{8(d-1)} &&(n\ge4),\\
 \Delta(\nu_{\mathsf{RPR}},\mathsf{PU}(d))
 &=\frac{d(d-3)}{8(d^2-1)} &&(n\ge3).
\end{aligned}
\end{equation}
At $n=3$ both values are $5/63$. For $n\ge4$, the special-unitary gap is
attained in $\bigwedge^8\mathbb C^d$. For $n\ge3$, the projective gap is
attained in the highest-weight representation $(1^4,0^{d-8},-1^4)$ and in
every balanced tensor $\tau_{t,t}$ with $t\ge4$. The orthogonal gaps are
\begin{equation}\label{eq:pauli-real-main}
 \Delta(\nu_{\mathsf{RPOR}},\mathsf{SO}(2^n))=
 \begin{cases}
 1,&n=1,\\5/24,&n=2,\\5/56,&n=3,\\
 (2^n-8)/(8(2^n-1)),&n\ge4.
 \end{cases}
\end{equation}
\end{theorem}

For the discrete law below, each Pauli direction is averaged uniformly over
$e^{\mathrm im\pi P/4}$, $m=0,\ldots,7$.
\begin{theorem}[Discrete Pauli rotations]\label{thm:pauli-discrete}
\begin{equation}\label{eq:discrete-main}
\begin{aligned}
 \Delta(\nu_{\mathsf{RPCR}},K_n)
 &=\begin{cases}1/2,&n=1,\\1/5,&n=2,\\(d-4)/(4(d-1)),&n\ge3,\end{cases}\\
 \Delta(\nu_{\mathsf{RPCR}},\overline{\Cl}(n))
 &=\begin{cases}1/2,&n=1,\\d/(4(d+1)),&n\ge2.\end{cases}
\end{aligned}
\end{equation}
For $n\ge3$ the generated-group value is attained in
$\bigwedge^4\mathbb C^d$; for $n\ge2$ the projective value is attained in
$\tau_{2,2}$. On $\Cl(n)$ the gap agrees with the generated-group gap for
$n\le3$ and is zero for $n\ge4$.
\end{theorem}

\subsection{Local inequalities, Pauli labels, and binary support}\label{subsec:pauli-common}

In a fixed representation, we use the Baer--Haah projector
$\Pi(\rho|\Theta,P)$ from \eqref{eq:pauli-projections}, and also put
$\Theta_{16}=\{m\pi/8+2\pi\mathbb Z:m\in\mathbb Z\}$.
Write $\rho(e^{\mathrm i\theta P})=e^{\mathrm i\theta J_P}$.
We also write $\sigma^x=X$, $\sigma^y=Y$, and $\sigma^z=Z$.
The self-adjoint generator $J_P$ has integral eigenvalues, and
$\Pi(\rho|\Theta_\infty,P)$ is its zero-eigenspace projection. Set
\[
\begin{aligned}
 F_P&=I-\Pi(\rho|\Theta_\infty,P),&
 F_{8,P}&=I-\Pi(\rho|\Theta_8,P),\\
 \mathcal H&=\mathcal H_{\Theta_\infty}^\rho=\sum_{P\in\mathcal P_n}F_P,&
 \mathcal H_8&=\mathcal H_{\Theta_8}^\rho=\sum_{P\in\mathcal P_n}F_{8,P}.
\end{aligned}
\]
For a representation defined only on a finite subgroup, the cyclic average
and its defect are defined directly, without introducing $J_P$.
The continuous and discrete averages are $I-\mathcal H/(d^2-1)$ and
$I-\mathcal H_8/(d^2-1)$, respectively. Both are positive semidefinite.
Their fixed spaces are those of $\mathsf{SU}(d)$ and $K_n$; all pair
sums below are ordered. When we pass to the real walk, the sums instead
run over $\mathcal Y_n$ and the normalizing factor is $d(d-1)/2$.

As in the proof of \cite[Lemma~3.4]{BH26}, let $\mathcal G$ be the graph on
$\mathcal P_n$ with $P\sim Q$ exactly when $PQ=-QP$. Identifying the labels
with $\mathbb F_2^{2n}\setminus\{0\}$, write a label as $(a,b)$,
where $(0,0),(1,0),(0,1),(1,1)$ at each qubit represent
$\mathbf{1}_2,\sigma^x,\sigma^z,\sigma^y$, respectively, up to scalar phase.
The symplectic form is
$\lambda((a,b),(a',b'))=a\cdot b'+b\cdot a'$.
Anticommutation is equivalent to this pairing being one.
If $P,Q$ anticommute, $\mathrm{i}P,\mathrm{i}Q$ generate a
copy of $\mathfrak{su}(2)$. Interchanging $P,Q$ if necessary, put
$R=-\mathrm{i}PQ\in\mathcal P_n$. Following Baer--Haah, call
$T=\{P,Q,R\}$ an \emph{$\mathfrak{su}(2)$-triangle}. The choice of sign is
immaterial because $\Pi(\rho|\Theta,-P)=\Pi(\rho|\Theta,P)$. Write
\[
 \mathcal H_T=\sum_{P\in T}F_P,\qquad
 \mathcal H_{8,T}=\sum_{P\in T}\bigl(\mathbf{1}-\Pi(\rho|\Theta_8,P)\bigr).
\]
Every ordered anticommuting pair lies in exactly one such triangle. Each
Pauli operator has $d^2/2$ anticommuting neighbors and belongs to $d^2/4$
triangles.

Applied to the sums of $\mathbf{1}-\Pi(\rho|\Theta,P)$ above, the local
estimates of \cite[Lemmas~3.13 and~3.10, and the proof of Lemma~3.4]{BH26} give
\begin{equation}\label{eq:exactgap-local-inputs}
 \mathcal H_T^2\succeq\frac54\mathcal H_T,\qquad
 \mathcal H_{8,T}^2\succeq\frac32\mathcal H_{8,T}
\end{equation}
in every representation of the corresponding copy of $\mathsf{SU}(2)$.

Set $\mathcal R_{\mathrm{ac}}=\sum_{PQ=-QP}F_PF_Q$. Triangle counting yields
\begin{equation}\label{eq:exactgap-triangle-count}
 \mathcal R_{\mathrm{ac}}
 =\sum_T(\mathcal H_T^2-\mathcal H_T)
 \succeq\frac14\sum_T\mathcal H_T
 =\frac{d^2}{16}\mathcal H.
\end{equation}

The adjacency matrix of $\mathcal G$ has $(P,Q)$ entry $1$ if
$PQ=-QP$ and $0$ otherwise.

\begin{lemma}[Spectrum of $\mathcal G$]\label{lem:exactgap-G-spectrum}
Let $A_{\mathcal G}$ be the adjacency matrix of $\mathcal G$ on the $d^2-1$ nonzero symplectic
labels. Then
\begin{equation*}
 A_{\mathcal G}\mathbf 1=\frac{d^2}{2}\mathbf 1,
 \qquad
 A_{\mathcal G}^2=\frac{d^2}{4}(\mathbf{1}+\mathbf 1\mathbf 1^{\mathsf T}),
 \qquad
 A_{\mathcal G}\preceq\frac d2\mathbf{1}+\frac{d}{2(d+1)}\mathbf 1\mathbf 1^{\mathsf T}.
\end{equation*}
\end{lemma}

\begin{proof}
The first identity follows by counting solutions to one nonzero linear
equation over $\mathbb F_2$. For two distinct nonzero labels
$p,q$, the equations $\lambda(p,r)=\lambda(q,r)=1$ are independent and have $d^2/4$
solutions. For $p=q$ there are $d^2/2$ solutions. This proves the second identity. On
$\mathbf 1^\perp$, the eigenvalues of $A_{\mathcal G}$ are among $\{-d/2,d/2\}$. The operator on the right-hand side of the last inequality acts by
$d^2/2$ on $\mathbf 1$ and by $d/2$ on $\mathbf 1^\perp$, proving the inequality.
\end{proof}

For a function $f$ on a finite set, write
$\operatorname{wt}(f)=|\{x:f(x)\ne0\}|$.

\begin{lemma}[A weight bound for polynomials over $\mathbb F_2$]\label{lem:exactgap-boolean-weight}
A nonzero function $q:\mathbb F_2^n\to\mathbb F_2$ represented by a polynomial of degree at most one in each variable
and total degree at most $r\le n$ is nonzero on at least $2^{n-r}$ inputs.
\end{lemma}

\begin{proof}
Use induction and write $q(x,t)=g(x)+th(x)$. If $h\ne0$, then $\deg h\le r-1$ and
$\operatorname{wt}(g)+\operatorname{wt}(g+h)\ge\operatorname{wt}(h)\ge2^{n-r}$. If $h=0$, the weight is $2\operatorname{wt}(g)$ and the induction
hypothesis gives the claim (using the trivial lower bound $1$ when the degree bound reaches the number
of variables). This also covers constant nonzero functions.
\end{proof}

\begin{lemma}[Dyadic Walsh divisibility]\label{lem:exactgap-integral-cubic}
Let $m,s$ be nonnegative integers, let
$w=(w_x)_{x\in\mathbb F_2^m}$ have integral entries, and set
\[
 a(p)=\sum_xw_x(-1)^{p\cdot x}.
\]
If every $a(p)$ is divisible by $2^s$, then $a(p)/2^s$ modulo two
is represented by a multilinear polynomial of degree at most $s$.
If $\sum_xw_x=0$, this polynomial vanishes at zero.
\end{lemma}
\begin{proof}
The integer multilinear polynomial representing $a$ is
\begin{equation}\label{eq:exactgap-walsh-polynomial}
 a(p)=\sum_xw_x\prod_{i:x_i=1}(1-2p_i).
\end{equation}
Boolean-cube interpolation expresses each coefficient as an integral
linear combination of values of $a$, so every coefficient is divisible
by $2^s$. The displayed product also makes every coefficient of degree
at least $s+1$ divisible by $2^{s+1}$. Divide by $2^s$ and reduce modulo
two. The last assertion follows from $a(0)=\sum_xw_x$.
\end{proof}

The cases $s=3$ and $s=2$ will control the continuous and discrete
projective walks, respectively. For the orthogonal torus we use $s=3$
on a cube of dimension $n-1$, without imposing a zero-sum condition.

\subsection{The special-unitary gap}\label{subsec:su-gap}

Apply Lemma~\ref{lem:exactgap-G-spectrum} to
$\mathcal T:V\to\mathbb C^{d^2-1}\otimes V$, $\mathcal T v=(F_Pv)_{P\in\mathcal P_n}$. Since
$\mathcal T^*\mathcal T=\mathcal H$,
$\mathcal T^*((\mathbf 1\mathbf 1^{\mathsf T})\otimes \mathbf{1})\mathcal T=\mathcal H^2$, and
$\mathcal T^*(A_{\mathcal G}\otimes \mathbf{1})\mathcal T=\mathcal R_{\mathrm{ac}}$, it yields
\begin{equation*}
 \mathcal R_{\mathrm{ac}}\preceq\frac d2\mathcal H+\frac{d}{2(d+1)}\mathcal H^2.
\end{equation*}
Combining this with \eqref{eq:exactgap-triangle-count} proves, in every representation,
\begin{equation}\label{eq:exactgap-su-lower}
 \mathcal H^2\succeq\frac{(d-8)(d+1)}8\mathcal H.
\end{equation}
For $d>8$, every nonzero eigenvalue of $\mathcal H$ is at least $(d-8)(d+1)/8$.
This proves the special-unitary lower bound in Theorem~\ref{thm:pauli-continuous}. No restriction on the action of the center or on the number of tensor
factors has been used.

At $d=8$ we instead retain the nonnegative commuting terms in the same calculation, obtaining
\begin{equation*}
 \mathcal H^2\succeq(1+d^2/16)\mathcal H.
\end{equation*}
Indeed, all $F_PF_Q$ with $P\ne Q$ and $PQ=QP$ are positive semidefinite. Thus every nonzero eigenvalue of $\mathcal H$ is at least $5$ when $d=8$.
The $\tau_{4,4}$ eigenvector of \cite[Proposition~3.19]{BH26} gives equality,
proving the $n=3$ special-unitary assertion of Theorem~\ref{thm:pauli-continuous}.

Here an \emph{affine three-flat} in $\mathbb F_2^n$ means a translate $b+U$
of a three-dimensional linear subspace $U\le\mathbb F_2^n$.

\subsubsection{A vector from affine three-dimensional subspaces}\label{subsec:exactgap-affine}

Index eight tensor positions by $x\in\mathbb F_2^3$, in lexicographic binary order. For $c\in\mathbb F_2^8$, write
$\lvert c\rangle=\lvert c_1\rangle\otimes\cdots\otimes\lvert c_8\rangle$,
where $\lvert0\rangle,\lvert1\rangle$ are the standard basis of $\mathbb C^2$.
Let $\operatorname{RM}(1,3)$ denote the first-order binary Reed--Muller code, and set
\begin{equation*}
 \mathcal C=\operatorname{RM}(1,3)
 =\bigl\{(a+b\cdot x)_{x\in\mathbb F_2^3}:a\in\mathbb F_2,\ b\in\mathbb F_2^3\bigr\}
 \subseteq\mathbb F_2^8,
 \qquad
 \lvert \mathcal C\rangle=\sum_{c\in\mathcal C}\lvert c\rangle.
\end{equation*}
A binary linear code of length eight is a linear subspace of $\mathbb F_2^8$;
the weight of a word is the number of its coordinates equal to one.
Code sums and invariants of the Clifford groups considered in
\cite[Sections~3 and~6]{NRS00} provide classical background. Here all
properties needed for the exterior product are proved directly.
Put $\mathcal C^\perp=\{y\in\mathbb F_2^8:y\cdot c=0\text{ for all }c\in\mathcal C\}$,
where the dot product is over $\mathbb F_2$.
The code $\mathcal C$ has dimension four and weights $0,4,8$.
It is self-orthogonal, meaning $\mathcal C\subseteq\mathcal C^\perp$:
the product of two affine functions has degree at most two, and its sum
over $\mathbb F_2^3$ is zero modulo two. Its dimension then implies
$\mathcal C=\mathcal C^\perp$. Thus it is self-dual (equality with its
dual) and doubly even (every word has weight divisible by four).

Regroup the tensor factors of $\lvert \mathcal C\rangle^{\otimes n}$ into eight groups of $n$ qubit factors. The resulting
vector is
\begin{equation*}
 \widetilde{\varpi}_n=\sum_{b\in\mathbb F_2^n}\ \sum_{L:\mathbb F_2^3\to\mathbb F_2^n\ \mathrm{linear}}
 \bigotimes_{x\in\mathbb F_2^3}\lvert b+Lx\rangle.
\end{equation*}
Let $\rho^{\wedge8}$ be the exterior-power representation on $\bigwedge^8\mathbb C^d$,
and define $\varpi_n$ to be the image of $\widetilde{\varpi}_n$ under the natural map
$v_1\otimes\cdots\otimes v_8\mapsto v_1\wedge\cdots\wedge v_8$.
The orientation of each wedge uses the stated order of $\mathbb F_2^3$;
consistency means independence of the affine parametrization.

\begin{lemma}\label{lem:exactgap-affine-three-flats}
For $n\ge3$, $\varpi_n\ne0$. After division by a common factor $1344$, it is the sum of one
consistently oriented wedge for each affine three-flat in $\mathbb F_2^n$.
\end{lemma}

\begin{proof}
For noninjective $L$, the corresponding summands have repeated tensor factors
and vanish under the map to the exterior power. Each injective
$L$ gives the eight vertices of an affine three-flat. Two parametrizations of the same flat differ by
an element of $\operatorname{AGL}(3,2)$, the group of maps $x\mapsto Ax+b$
with $A$ an invertible $3\times3$ matrix over $\mathbb F_2$ and $b\in\mathbb F_2^3$. Every such permutation of eight vertices is even: a nonzero translation
is a product of four transpositions, and a coordinate map $x_i\mapsto x_i+x_j$ ($i\ne j$), leaving the
other coordinates fixed, is a product of two transpositions. These maps generate the affine group, so the summands from different
parametrizations do not cancel. Each
flat has $|\operatorname{AGL}(3,2)|=8\cdot168=1344$ parametrizations.
\end{proof}

Let
\[
 h=2^{-1/2}\begin{pmatrix}1&1\\1&-1\end{pmatrix},
 \qquad
 s=\operatorname{diag}(1,\mathrm{i})
\]
be the Hadamard and phase matrices. For $y\in\mathbb F_2^8$, the coefficient
of $\lvert y\rangle$ in $h^{\otimes8}\lvert\mathcal C\rangle$ is
\[
 2^{-4}\sum_{c\in\mathcal C}(-1)^{c\cdot y}.
\]
It equals $|\mathcal C|/2^4=1$ for $y\in\mathcal C^\perp=\mathcal C$ and
zero otherwise, so $h^{\otimes8}\lvert\mathcal C\rangle=\lvert\mathcal C\rangle$. Doubly even weights give $s^{\otimes8}\lvert \mathcal C\rangle=\lvert \mathcal C\rangle$. Applying a
controlled-NOT (CNOT) gate to each pair of corresponding coordinates in two
eight-qubit codeword factors maps $(c_1,c_2)$ to $(c_1,c_1+c_2)$.
This is a bijection of $\mathcal C\times\mathcal C$ because $\mathcal C$ is linear,
so it also preserves the sum. It follows that $\widetilde{\varpi}_n$, and hence $\varpi_n$, is fixed by the eightfold action of
every composition of these operators.

The matrix $e^{\mathrm{i}\pi\sigma^z/4}=e^{\mathrm{i}\pi/4}s^\dagger$ also fixes this eightfold tensor: its extra scalar has
eighth power one. Standard Hadamard, phase and CNOT conjugations carry $\sigma^z$ on one qubit to any
Pauli operator in $\mathcal P_n$, up to sign. For example, CNOTs build any tensor product of $\mathbf{1}_2$ and $\sigma^z$
with at least one nonidentity factor, and single-qubit
gates convert its nonidentity factors to $\sigma^x,\sigma^y$ or $\sigma^z$. Thus
\begin{equation}\label{eq:exactgap-pi4-fixed}
 \rho^{\wedge8}(e^{\mathrm{i}\pi P/4})\varpi_n=\varpi_n
 \qquad(P\in\mathcal P_n).
\end{equation}

\subsubsection{Computing the exterior-power eigenvalue}

The quadratic Casimir operator used below is the sum
$\sum_{P\in\mathcal P_n}J_P^2$.
Its normalization is fixed by
$\Tr(PQ)=d\,\delta_{P,Q}$ for $P,Q\in\mathcal P_n$,
where $\delta_{P,Q}$ is the Kronecker delta: the matrices $P/\sqrt d$
are orthonormal for the trace inner product $(A,B)\mapsto\Tr(AB)$ on
traceless Hermitian matrices.
Thus this sum is $d$ times the Casimir operator for a trace-orthonormal basis.

On the exterior power, $J_P=\sum_{r=1}^8P^{(r)}$, with $P^{(r)}$ acting on
tensor factor $r$ and the sum restricted to the antisymmetric tensor space.
Its possible eigenvalues are $0,\pm2,\pm4,\pm6,\pm8$.
Equation~\eqref{eq:exactgap-pi4-fixed} restricts those occurring in
$\varpi_n$ to $0,\pm8$. Since $F_P$ projects onto the nonzero eigenspaces of $J_P$,
\begin{equation*}
 F_P\varpi_n=\frac{J_P^2}{64}\varpi_n.
\end{equation*}
The one-qubit identity
$\mathbf{1}_2\otimes\mathbf{1}_2+\sigma^x\otimes\sigma^x+\sigma^y\otimes\sigma^y+\sigma^z\otimes\sigma^z=2T$, tensored over $n$ qubits, implies
\begin{equation}\label{eq:exactgap-pauli-swap}
 \sum_{P\in\mathcal P_n}P\otimes P=dT-\mathbf 1.
\end{equation}
Here $T(v\otimes w)=w\otimes v$, and $T_{\{r,s\}}$ denotes the corresponding
transposition of tensor factors $r$ and $s$. On the exterior-eight subspace every transposition has
eigenvalue $-1$. Therefore
\begin{equation*}
 \sum_{P\in\mathcal P_n}J_P^2
 =8(d^2-1)\mathbf{1}+2\sum_{1\le r<s\le8}(dT_{\{r,s\}}-\mathbf 1)
 =8(d-8)(d+1)\mathbf{1}.
\end{equation*}
More generally, the same identity on $\bigwedge^k\mathbb C^d$ gives
\begin{equation}\label{eq:su-exterior-casimir}
 \sum_{P\in\mathcal P_n}J_P^2=k(d-k)(d+1)I.
\end{equation}
Together with the preceding formula, this gives the exact eigenvector equation
\begin{equation*}
 \mathcal H\varpi_n=\frac{(d-8)(d+1)}8\varpi_n.
\end{equation*}
For $d>8$, $\bigwedge^8\mathbb C^d$ is a nontrivial irreducible $\mathsf{SU}(d)$ representation. Dividing by $|\mathcal P_n|=4^n-1=d^2-1$ and combining with \eqref{eq:exactgap-su-lower} proves the special-unitary assertion of Theorem~\ref{thm:pauli-continuous}.

\subsection{The projective-unitary gap}\label{subsec:pu-gap}

\subsubsection{Reduction to the subgroup-fixed space}\label{subsec:exactgap-subgroup}

Use the generated group $K_n$ from \eqref{eq:unitary-targets}.
Each generator conjugates the chosen Hermitian Pauli representatives by a signed permutation:
for $k\in K_n$ and $P\in\mathcal P_n$, one has $kPk^{-1}=\pm P'$ for some $P'\in\mathcal P_n$. Conjugation transports the averaging projectors according to
$\rho(k)\Pi(\rho|\Theta,P)\rho(k)^{-1}=\Pi(\rho|\Theta,kPk^{-1})$.
Moreover,
\[
 \Pi(\rho|\Theta,-P)=\Pi(\rho|\Theta,P)
 \qquad\bigl(\Theta\in\{\Theta_8,\Theta_{16},\Theta_\infty\}\bigr).
\]
These identities show that summing the projectors, their complements, or their differences over $\mathcal P_n$ gives operators that commute with $K_n$. We never identify $K_n$ with the
special Clifford group $\Cl(n)$; the distinction in central phases is essential below.

For $n\ge4$, this is a proper subgroup of $\Cl(n)$:
by \eqref{eq:exactgap-pi4-fixed} it fixes $\varpi_n$, whereas
$e^{2\pi\mathrm{i}/d}\mathbf{1}\in \Cl(n)$ does not. No assertion
that these rotations generate $\Cl(n)$ is used.

Define
\begin{equation*}
 \mathcal H_8=\sum_{P\in\mathcal P_n}\bigl(\mathbf{1}-\Pi(\rho|\Theta_8,P)\bigr).
\end{equation*}
The positive projectors satisfy
\[
 \Pi(\rho|\Theta_\infty,P)
 \preceq \Pi(\rho|\Theta_{16},P)
 \preceq \Pi(\rho|\Theta_8,P).
\]
Also $\ker \mathcal H_8=V^{K_n}$ and $\mathcal H\succeq \mathcal H_8$.

By the second inequality in \eqref{eq:exactgap-local-inputs}, the
same triangle count gives
$\sum_{PQ=-QP}(\mathbf{1}-\Pi(\rho|\Theta_8,P))
(\mathbf{1}-\Pi(\rho|\Theta_8,Q))\succeq(d^2/8)\mathcal H_8$.
The terms indexed by distinct commuting Pauli operators are nonnegative, so
\begin{equation*}
 \mathcal H_8^2\succeq(1+d^2/8)\mathcal H_8.
\end{equation*}
Both $\mathcal H$ and $\mathcal H_8$ preserve $V^{K_n}$ and its orthogonal complement. It follows that on
$(V^{K_n})^\perp$,
\begin{equation}\label{eq:exactgap-complement}
 \mathcal H\succeq \mathcal H_8\succeq(1+d^2/8)\mathbf{1}\succ\frac{d(d-3)}8\mathbf{1}.
\end{equation}
Only the $K_n$-fixed space remains to be treated.

\subsubsection{Odd and even multiples of eight}\label{subsec:exactgap-residue}

Define the following refinements of $F_P$ and $\mathcal H$:
\begin{equation*}
 \begin{aligned}
 F_P^{\mathrm{odd}}&:=\Pi(\rho|\Theta_8,P)-\Pi(\rho|\Theta_{16},P),&
 \mathcal H^{\mathrm{odd}}&:=\sum_{P\in\mathcal P_n}F_P^{\mathrm{odd}},\\
 F_P^{\mathrm{even}}&:=\Pi(\rho|\Theta_{16},P)-\Pi(\rho|\Theta_\infty,P),&
 \mathcal H^{\mathrm{even}}&:=\sum_{P\in\mathcal P_n}F_P^{\mathrm{even}}.
 \end{aligned}
\end{equation*}
If $\rho(e^{\mathrm{i}\theta P})=e^{\mathrm{i}\theta J_P}$, then $F_P^{\mathrm{odd}}$ projects onto the $J_P$-eigenspaces whose eigenvalues are congruent to $8$ modulo
$16$, whereas $F_P^{\mathrm{even}}$ projects onto the nonzero $J_P$-eigenspaces whose eigenvalues are divisible by $16$. Every vector in $V^{K_n}$ has nonzero components only
in $J_P$-eigenspaces with eigenvalues divisible by eight. On that fixed
space,
\begin{equation*}
 \mathcal H=\mathcal H^{\mathrm{odd}}+\mathcal H^{\mathrm{even}}.
\end{equation*}
Let
\begin{equation*}
 \mathcal R_{\mathrm{com}}=\sum_{\substack{P\ne Q\\PQ=QP}}F_PF_Q,
 \qquad
 \mathcal D_{\triangle}=\sum_T\left(\mathcal H_T^2-\frac54\mathcal H_T\right),
\end{equation*}
where the sum defining $\mathcal D_{\triangle}$ runs over all
$\mathfrak{su}(2)$-triangles. Both operators are positive semidefinite and satisfy
\begin{equation}\label{eq:exactgap-H-square}
 \mathcal H^2=(1+d^2/16)\mathcal H+\mathcal R_{\mathrm{com}}+\mathcal D_{\triangle}.
\end{equation}
If $Z(\mathsf{SU}(d))$ acts trivially on $V$, we will establish on
$V^{K_n}$ the inequalities
\begin{equation}\label{eq:exactgap-projective-estimates}
 \mathcal R_{\mathrm{com}}\succeq\frac{(d-8)(d+2)}{16}\mathcal H^{\mathrm{odd}},
 \qquad
 \mathcal D_{\triangle}\succeq\frac{d^2}{16}\mathcal H^{\mathrm{even}}.
\end{equation}
These estimates treat the odd and even multiples of eight separately.
Both estimates concern the full sums over Pauli operators.

\subsubsection{Commuting terms and binary polynomials}\label{subsec:exactgap-commuting}

\begin{proposition}[Commuting-pair bound]\label{prop:exactgap-commuting}
In every representation of $\mathsf{PU}(d)$, for $n\ge3$, the first inequality in \eqref{eq:exactgap-projective-estimates} holds on $V^{K_n}$.
\end{proposition}

\begin{proof}
For $p=(p_1,\ldots,p_n)\in\mathbb F_2^n$, put
$\sigma^z(p)=\bigotimes_{j=1}^n(\sigma^z)^{p_j}$.
Choose $L\subseteq\mathcal P_n$ maximal under inclusion among
pairwise commuting subsets. First take the diagonal one,
$L=\{\sigma^z(p):p\in\mathbb F_2^n\setminus\{0\}\}$. Simultaneously diagonalize the action of the determinant-one diagonal
matrices. On each common eigenspace, $\operatorname{diag}(t_x)$ acts by
$\prod_x t_x^{m_x}$, with $m\in\mathbb Z^d$ determined modulo adding a
constant integer vector. The scalar $e^{2\pi\mathrm{i}/d}\mathbf{1}$ acts
there by $e^{2\pi\mathrm{i}\sum_xm_x/d}$. Because the representation
factors through $\mathsf{PU}(d)$, trivial central action implies
$\sum_xm_x\in d\mathbb Z$. Thus
$w=m-(\sum_xm_x/d)\mathbf{1}$ is an integral representative of sum zero.

A vector fixed by all $e^{\mathrm{i}\pi\sigma^z(p)/4}$ has nonzero components only in common eigenspaces for which
$a_p=\sum_xw_x(-1)^{p\cdot x}\in8\mathbb Z$ for every $p\ne0$. The identity $a_0=0$ supplies the condition at zero. On one such common eigenspace define
\[
 k=|\{p\ne0:a_p\ne0\}|,
 \qquad
 k_o=|\{p\ne0:a_p/8\text{ is odd}\}|.
\]
If $k_o>0$, Lemmas~\ref{lem:exactgap-boolean-weight} and~\ref{lem:exactgap-integral-cubic}, including $q(0)=0$, imply $k\ge k_o\ge d/8$. Therefore,
also when $k_o=0$,
\begin{equation}\label{eq:exactgap-commuting-weight}
 k(k-1)\ge(d/8-1)k_o.
\end{equation}
For $\mathcal H_L=\sum_{P\in L}F_P$, the left side is the eigenvalue of $\mathcal H_L^2-\mathcal H_L$, and $k_o$ is the
eigenvalue of $\sum_{P\in L}F_P^{\mathrm{odd}}$. This proves the corresponding quadratic-form inequality on
$V^{K_n}$.

The same calculation applies to every maximal commuting set.
Choose $n$ members with linearly independent Pauli labels. Their joint
eigenspaces are one-dimensional: each joint spectral projection has trace
one, since every nonidentity product of these generators is traceless.
In a joint eigenbasis their products are the operators $\sigma^z(p)$, up to
signs. The change of basis can be chosen in $\mathsf{SU}(d)$, so the
integral zero-sum description of the diagonal action still applies.
The calculation only requires that the vector be fixed by the
$\pi/4$ rotations from the chosen commuting set.
Changing the sign of a Pauli representative does not affect $F_P$ or
$F_P^{\mathrm{odd}}$.

In terms of labels, maximal commuting sets are the nonzero vectors of
$n$-dimensional subspaces on which the symplectic form vanishes.
The group $\Sp(2n;\mathbb F_2)$ of linear maps preserving
this form permutes these subspaces. It is transitive on nonzero labels
and on ordered pairs of distinct nonzero labels with pairing zero.
Indeed, such a pair is linearly independent over $\mathbb F_2$, and either
a single nonzero label or such a pair can be extended to a symplectic
basis $e_1,\ldots,e_n,f_1,\ldots,f_n$ with
$\lambda(e_i,e_j)=\lambda(f_i,f_j)=0$ and
$\lambda(e_i,f_j)=\delta_{ij}$; mapping one such basis to another
gives the required transformation.
Let $r_1$ count the maximal commuting sets containing one given Pauli operator,
and $r_c$ count those containing a given ordered distinct commuting pair.
The transitivity just established makes these counts independent of the chosen operator
or pair. Double counting with one Pauli operator fixed gives
\[
 r_1(d-2)=r_c(d^2/2-2),
 \qquad
 \frac{r_c}{r_1}=\frac2{d+2}.
\]
Summing \eqref{eq:exactgap-commuting-weight} over maximal commuting sets yields
\[
 r_c\mathcal R_{\mathrm{com}}\succeq r_1(d/8-1)\mathcal H^{\mathrm{odd}}.
\]
Division gives the claimed coefficient $(d-8)(d+2)/16$.
\end{proof}

\subsubsection{A local inequality on \texorpdfstring{$\mathfrak{su}(2)$}{su(2)}-triangles}\label{subsec:exactgap-cl1}

\begin{lemma}[$\mathsf{Cl}(1)$-fixed-space inequality]\label{lem:exactgap-cl1-local}
In any $\mathsf{SU}(2)$ representation, write $F_x=F_{\sigma^x}$, $F_y=F_{\sigma^y}$,
$F_z=F_{\sigma^z}$ and $F_x^{\mathrm{even}}=F_{\sigma^x}^{\mathrm{even}}$,
$F_y^{\mathrm{even}}=F_{\sigma^y}^{\mathrm{even}}$, $F_z^{\mathrm{even}}=F_{\sigma^z}^{\mathrm{even}}$. Let
$\mathcal H_T=F_x+F_y+F_z$ and $\mathcal H_T^{\mathrm{even}}=F_x^{\mathrm{even}}+F_y^{\mathrm{even}}+F_z^{\mathrm{even}}$. Set
\[
 \mathsf G_{\Theta_8}:=\langle e^{\mathrm{i}\pi\sigma^x/4},e^{\mathrm{i}\pi\sigma^y/4},e^{\mathrm{i}\pi\sigma^z/4}\rangle.
\]
Baer--Haah identify $\mathsf G_{\Theta_8}$ with their special Clifford group
$\mathsf{Cl}(1)$ in \cite[Lemma~3.10]{BH26}. On the $\mathsf{Cl}(1)$-fixed subspace,
\begin{equation}\label{eq:exactgap-cl1-local}
 \mathcal H_T^2-\frac54\mathcal H_T\succeq\frac14\mathcal H_T^{\mathrm{even}}.
\end{equation}
\end{lemma}

\begin{proof}
Decompose the representation into spin-$j$ summands, $j\in\tfrac12\mathbb Z_{\ge0}$.
Each summand acts on homogeneous polynomials $f$ of degree $2j$ in $z\in\mathbb C^2$
by $(\rho_j(U)f)(z)=f(U^{-1}z)$. Write $J_a=2L_a$ for $a\in\{x,y,z\}$; in particular, the eigenvalues of $J_z$ are
$2m$, $m=-j,-j+1,\ldots,j$. Baer--Haah further show in \cite[Lemma~3.10]{BH26} that
$\mathsf{Cl}(1)$ has order $48$ and center $\{\pm\mathbf{1}\}$, and that
$\mathsf{Cl}(1)/\{\pm\mathbf{1}\}\cong\symgp_4$ is the rotational symmetry group of the octahedron.
Both $\mathcal H_T$ and $\mathcal H_T^{\mathrm{even}}$ preserve the fixed subspace.

Half-integral spins have no fixed vectors because the central element $-\mathbf{1}$ acts as $-1$. For odd
integral spin, a vector in the zero-eigenspace of any $J_a$ changes sign under a $\pi$ rotation about a
perpendicular coordinate axis. Thus it is orthogonal to every $\mathsf{Cl}(1)$-fixed vector. On the fixed space
$\mathcal H_T=3\mathbf{1}$. Since $\mathcal H_T^{\mathrm{even}}\preceq \mathcal H_T$, \eqref{eq:exactgap-cl1-local} follows.

Suppose next that $j$ is even. Realize the spin-$j$ representation as the
complex space of degree-$j$ spherical harmonics: restrictions to $S^2$
of homogeneous complex polynomials $h$ of degree $j$ satisfying
$\partial_1^2h+\partial_2^2h+\partial_3^2h=0$.
The inner product is $\int_{S^2}\overline f g\,dA$, with the usual
surface area $\int_{S^2}dA=4\pi$. For $a\in\{x,y,z\}$ choose
\[
 e_a(t)=\sqrt{\frac{2j+1}{4\pi}}\,P_j(t_a),\qquad t\in S^2,
\]
where $t_a$ is the $a$-coordinate and $P_j$ is the degree-$j$ Legendre
polynomial defined by
$P_j(s)=(2^j j!)^{-1}(\frac{d}{ds})^j(s^2-1)^j$, so $P_j(1)=1$.
These functions are unit zero-eigenvalue vectors of $J_a$, since they
are invariant under rotations about the corresponding axis.
The spherical-harmonic addition formula
\cite[Eq.~(14.30.9)]{DLMF} gives their inner product as $P_j$ of the dot
product of the two axes. For distinct coordinate axes this is
\begin{equation*}
 u_j=P_j(0)=(-1)^{j/2}\binom{j}{j/2}2^{-j}.
\end{equation*}
Since $P_j(-s)=P_j(s)$ for even $j$, reversing an axis does not change
the chosen vector. The octahedral rotation group $\symgp_4$ therefore
permutes these three vectors without signs. Consequently, on the $\mathsf{Cl}(1)$-fixed space,
the operator $\sum_{a\in\{x,y,z\}}\Pi(\rho|\Theta_\infty,\sigma^a)$ has range contained in the span of $u=e_x+e_y+e_z$. If this vector is
nonzero, the eigenvalue on that line is $1+2u_j$; on its orthogonal complement the sum is zero.
Thus the possible $\mathcal H_T$ eigenvalues on the fixed space are
\begin{equation}\label{eq:exactgap-HT-spectrum}
 h_j=2(1-u_j)\quad\text{on }\mathbb C u,
 \qquad
 3\quad\text{on }u^\perp.
\end{equation}
When $u=0$, the first case is absent.

For $j=0$, $\mathcal H_T=\mathcal H_T^{\mathrm{even}}=0$. For $j=4$, every $J_a$-eigenvalue has absolute value at most eight, so $\mathcal H_T^{\mathrm{even}}=0$
and the ordinary $5/4$ bound suffices. For every other even $j$ except $j=8$, the eigenvalues in
\eqref{eq:exactgap-HT-spectrum} are at least $3/2$. Indeed, $u_j<0$ for $j\equiv2\pmod4$, and for
$j\equiv0\pmod4$, $j\ge12$, the positive sequence decreases from
$u_{12}=231/1024<1/4$. Therefore
\[
 \mathcal H_T(\mathcal H_T-5/4)\succeq \mathcal H_T/4\succeq \mathcal H_T^{\mathrm{even}}/4
\]
on all those fixed spaces. It remains to check spin eight.

The $\mathsf{Cl}(1)$-fixed space in spin eight is one-dimensional. Indeed, it equals the $\symgp_4$-fixed space, whose dimension is the trace of the average of its 24 action matrices:
\[
 \frac1{24}\bigl(17+8(-1)+6(1)+9(1)\bigr)=1,
\]
using rotations of angles $0,2\pi/3,\pi/2,\pi$, respectively. On this line,
\begin{equation*}
 u_8=\frac{35}{128},
 \qquad
 \lVert u\rVert^2=3+6u_8=\frac{297}{64},
 \qquad
 \mathcal H_T=\frac{93}{64}.
\end{equation*}
For the $z$ direction, $F_z^{\mathrm{even}}$ selects $m=8,-8$. The overlap of $e_x$ with a unit $L_z$-eigenvector of eigenvalue $8$ or $-8$
has squared modulus
\[
 D=\frac{\binom{16}{8}}{4^8}.
\]
In the degree-sixteen polynomial model, choose
$\sqrt{\binom{16}{k}}\,z_1^{16-k}z_2^k$ ($0\le k\le16$) as an
orthonormal basis. Expanding the rotated unit vector with zero $L_x$
eigenvalue gives the displayed squared coefficient. The corresponding component of $e_y$ has the same phase, since $e^{\pm\mathrm{i} 8\pi/2}=1$; $e_z$ has no such
component. The two extremal components of $u$ therefore contribute $8D$. Summing over $a\in\{x,y,z\}$
and normalizing gives
\[
 \mathcal H_T^{\mathrm{even}}=\frac{24D}{3+6u_8}=\frac{65}{64}.
\]
Finally,
\begin{equation*}
 \frac{93}{64}\left(\frac{93}{64}-\frac54\right)
 -\frac14\frac{65}{64}
 =\frac{169}{4096}>0.
\end{equation*}
This proves the last case and the lemma, including any number of copies of each spin representation.
\end{proof}

Every $K_n$-fixed vector is fixed by the local copy of $\mathsf G_{\Theta_8}=\mathsf{Cl}(1)$ on every $\mathfrak{su}(2)$-triangle. Sum
\eqref{eq:exactgap-cl1-local} over these triangles, recalling that each Pauli operator belongs to $d^2/4$ such triangles. It follows that
\begin{equation}\label{eq:exactgap-even-defect}
 \mathcal D_{\triangle}\succeq\frac14\sum_T\sum_{P\in T}F_P^{\mathrm{even}}
 =\frac{d^2}{16}\mathcal H^{\mathrm{even}}
 \qquad\text{on }V^{K_n}.
\end{equation}
This is the second inequality in \eqref{eq:exactgap-projective-estimates}.

\subsubsection{Proof of the projective gap and its attaining representation}

For $d\ge8$, $0\le (d-8)(d+2)/16\le d^2/16$. Combining
\eqref{eq:exactgap-H-square}, Proposition~\ref{prop:exactgap-commuting}, \eqref{eq:exactgap-even-defect}, and $\mathcal H=\mathcal H^{\mathrm{odd}}+\mathcal H^{\mathrm{even}}$ on $V^{K_n}$ gives
\begin{equation*}
 \begin{aligned}
 \mathcal H^2&\succeq(1+d^2/16)\mathcal H+\frac{(d-8)(d+2)}{16}\mathcal H^{\mathrm{odd}}+(d^2/16)\mathcal H^{\mathrm{even}}\\
 &\succeq\left(1+d^2/16+\frac{(d-8)(d+2)}{16}\right)\mathcal H
 =\frac{d(d-3)}8\mathcal H.
 \end{aligned}
\end{equation*}
Together with the stronger bound \eqref{eq:exactgap-complement} on the complement, this shows that every positive eigenvalue of $\mathcal H$ in
every representation of $\mathsf{PU}(d)$ is at least $d(d-3)/8$. The $\tau_{4,4}$ eigenvector from
\cite[Proposition~3.19]{BH26} supplies equality after division by $|\mathcal P_n|=4^n-1=d^2-1$. This proves the projective gap formula in Theorem~\ref{thm:pauli-continuous}.

It remains to identify the highest weight of an attaining representation. We use the objects of
\cite[Proposition~3.19]{BH26}: $\xi=(1^4)$, the projection $\Pi_\xi$ onto
$\bigwedge^4\mathbb C^d\subset(\mathbb C^d)^{\otimes4}$, and the averaging
linear map $\mathbf\Gamma=\Pi(\tau_{4,4}|\mathsf{SU}(d))$ on
$\operatorname{End}((\mathbb C^d)^{\otimes4})$, explicitly
$\mathbf\Gamma(X)=\int_{\mathsf{SU}(d)}U^{\otimes4}X(U^\dagger)^{\otimes4}\,d\mu(U)$.
Put
\[
 \Omega=\mathbf{1}_d^{\otimes4}+\sum_{P\in\mathcal P_n}P^{\otimes4}.
\]
The Baer--Haah upper-bound eigenvector is exactly
\[
 E_\xi=(\mathbf{1}-\mathbf\Gamma)(\Pi_\xi\Omega\Pi_\xi)
       =(\mathbf{1}-\mathbf\Gamma)(\Pi_\xi\Omega).
\]
It is nonzero by the same proposition and is
$K_n$-fixed because conjugation by $K_n$ permutes the Pauli summands and the fourth tensor
power removes their signs. For every $P$, the $J_P$-eigenvalues occurring in $E_\xi$ lie in $[-8,8]$. Invariance under $e^{\mathrm{i}\pi P/4}$ restricts these eigenvalues to
$0,\pm8$. Hence $F_PE_\xi=J_P^2E_\xi/64$, and
$(\sum_{P\in\mathcal P_n}J_P^2)E_\xi=8d(d-3)E_\xi$, using its known
$\mathcal H$ eigenvalue from \cite[Proposition~3.19]{BH26}.

Writing $V(\lambda)$ for the irreducible representation of highest weight
$\lambda$, the Pieri rule for tensoring with an exterior power
\cite[Section~1.1]{HMS19} gives, for $d\ge8$,
\begin{equation*}
 \operatorname{End}(\bigwedge^4\mathbb C^d)
 =\bigoplus_{r=0}^4V(1^r,0^{d-2r},-1^r).
\end{equation*}
As $\mathsf{GL}(d,\mathbb C)$ representations,
$(\bigwedge^4\mathbb C^d)^*\cong
\bigwedge^{d-4}\mathbb C^d\otimes\det^{-1}$.
The Pieri rule adds four boxes to the column $(1^{d-4})$, with no two
added boxes in the same row. The resulting highest weights, padded to
$d$ coordinates, are $(2^r,1^{d-2r},0^r)$ for $0\le r\le4$, each occurring
once. Tensoring by $\det^{-1}$ subtracts one from every
coordinate, giving the displayed decomposition after restriction to
$\mathsf{SU}(d)$.
With the normalization fixed by \eqref{eq:exactgap-pauli-swap}, the operator
$\sum_{P\in\mathcal P_n}J_P^2$ acts on the summand indexed by $r$ as the scalar
\begin{equation*}
 2rd(d-r+1).
\end{equation*}
For a zero-sum highest weight $\lambda$, this scalar is
$c_\lambda=d(\lambda,\lambda+2\delta)$, where
$\delta_i=(d+1-2i)/2$ for $1\le i\le d$ and the inner product is Euclidean.
Thus $\delta=\tfrac12\sum_{i<j}(\varepsilon_i-\varepsilon_j)$,
where $\varepsilon_i$ are the standard coordinate vectors; each difference
in the sum has squared length two. To evaluate the scalar, choose a highest-weight vector and a trace-orthonormal
basis. Write $E_{ij}$ for the matrix with a single entry $1$ in position
$(i,j)$. The diagonal terms contribute
$(\lambda,\lambda)$ and the matrix-unit pairs $E_{ij},E_{ji}$ ($i<j$) contribute
$\sum_{i<j}(\lambda_i-\lambda_j)=(\lambda,2\delta)$;
the factor $d$ comes from $\Tr(PQ)=d\,\delta_{P,Q}$.
For $\lambda=(1^r,0^{d-2r},-1^r)$, the two contributions are $2r$ and
$2r(d-r)$, respectively. The same normalization gives $d^2-1$ for the action
$U\mapsto U$ on $\mathbb C^d$. These five values are distinct
for $0\le r\le4$, and only $r=4$ gives $8d(d-3)$. Thus $E_\xi$ lies in the stated highest-weight
representation. Let $e_1,\ldots,e_d$ be the standard basis of $\mathbb C^d$. The nonzero
vector $\sum_i e_i\otimes\bar e_i$ is fixed by $\tau_{1,1}$. Tensoring
with $t-4$ copies of this vector, and regrouping the factors, gives an
injective map from $\tau_{4,4}$ to $\tau_{t,t}$ commuting with the group
actions, proving the remaining assertions.

\subsection{Discrete Pauli rotations and finite-representation occurrence}\label{subsec:discrete-proof}

The cyclic averages in this subsection are those of the actual eight-element
subgroups, not averages of a single chosen quarter-turn. We first make
explicit why the torus arguments below apply to every finite-group
representation.

\begin{lemma}[Occurrence under restriction]\label{lem:finite-occurrence}
Let $G$ be a compact matrix group and $F\le G$ a finite subgroup.
Every irreducible representation of $F$ occurs in the restriction of a
finite-dimensional continuous unitary representation of $G$.
\end{lemma}
\begin{proof}
Choose a faithful matrix representation of $G$. Its matrix entries separate
the points of $F$. For each $a\in F$ and $b\ne a$, choose an entry function
$f_{a,b}$ with $f_{a,b}(a)\ne f_{a,b}(b)$. The polynomial function
\[
 p_a(g)=\prod_{b\in F\setminus\{a\}}
 \frac{f_{a,b}(g)-f_{a,b}(b)}{f_{a,b}(a)-f_{a,b}(b)}
\]
restricts to the indicator of $a$ on $F$. Polynomials in the matrix entries
of degree at most $|F|-1$, restricted to $G$, form a finite-dimensional
space stable under left translations by $G$. Evaluation on $F$ maps
this space onto $\mathbb C[F]$ and intertwines left translations by $F$.
Choose a $G$-invariant inner product by averaging; the orthogonal complement
of the kernel of evaluation is an $F$-invariant copy of $\mathbb C[F]$.
Every irreducible of $F$ occurs in this regular representation.
\end{proof}

For $G=\mathsf{PU}(d)$ one may use the faithful conjugation representation
on $\operatorname{End}(\mathbb C^d)$. The lemma asserts occurrence under
restriction; it does not assert that an individual finite representation
extends to the same vector space as a representation of $G$.

\subsubsection{The gap on the generated group}

Let $V$ be any $K_n$ representation. For the defects $F_{8,P}$, the discrete
local inequality in \eqref{eq:exactgap-local-inputs} gives
\[
 \mathcal R_{\mathrm{ac},8}:=
 \sum_{PQ=-QP}F_{8,P}F_{8,Q}\succeq\frac{d^2}{8}\mathcal H_8.
\]
This application is legitimate in a finite representation: each triangle
is a homomorphic image of the binary octahedral group generated by its
three quarter-turns. Pulling its representation back gives exactly the
three cyclic averages in the local inequality. Applying the graph bound
of Lemma~\ref{lem:exactgap-G-spectrum} to $v\mapsto(F_{8,P}v)_P$ also gives
\[
 \mathcal R_{\mathrm{ac},8}\preceq
 \frac d2\mathcal H_8+\frac{d}{2(d+1)}\mathcal H_8^2.
\]
Consequently,
\begin{equation}\label{eq:discrete-actual-lower}
 \mathcal H_8^2\succeq\frac{(d-4)(d+1)}4\mathcal H_8.
\end{equation}
Its kernel is $V^{K_n}$, so this proves the generated-group lower bound
in Theorem~\ref{thm:pauli-discrete} for $n\ge3$.

To attain it, use $\bigwedge^4\mathbb C^d$. For $n\ge2$ the fourth exterior
actions of the Hermitian Paulis form a representation of the elementary
abelian projective Pauli group: all multiplication phases have fourth
power one. Its fixed-space dimension is
\begin{equation}\label{eq:pauli-exterior-four-dimension}
 \frac1{d^2}\left[\binom d4+(d^2-1)\binom{d/2}{2}\right]
 =\frac{(d-1)(d-2)}6>0.
\end{equation}
Indeed, a nonidentity Pauli has $d/2$ eigenvalues of each sign, and the
coefficient of $z^4$ in $(1+z)^{d/2}(1-z)^{d/2}$ is $\binom{d/2}{2}$.
Take a nonzero fixed vector $w$. Since $e^{\mathrm i\pi P/2}=\mathrm iP$
fixes $w$, only the $J_P$ weights $0,\pm4$ occur in $w$. Hence
\[
 F_{8,P}w=\frac{J_P^2}{16}w,
 \qquad
 \mathcal H_8w=\frac{(d-4)(d+1)}4w,
\]
by \eqref{eq:su-exterior-casimir}. For $d\ge8$ the last eigenvalue is
positive, so $w$ is orthogonal to all $K_n$ invariants and proves equality.
No irreducibility of the restricted exterior-power representation is needed.

At $n=2$, positivity of the distinct commuting terms instead gives
$\mathcal H_8^2\succeq(1+d^2/8)\mathcal H_8=3\mathcal H_8$.
After division by $15$ this is $1/5$. At $n=1$, there is just one triangle,
so its local gap gives $\Delta\ge(3/2)/3=1/2$. The projective witnesses
constructed below attain both of these bounds on $K_n$ as well.

\subsubsection{The projective lower bound}

The normal projective Pauli subgroup $\overline{\mathcal P}_n\simeq\mathbb F_2^{2n}$ is normal in $\overline{\Cl}(n)$. The full sum $\mathcal H_8$
commutes with $\overline{\mathcal P}_n$. Since $[g_P]^2=[P]$, its cyclic projector is bounded by
$(I+\rho([P]))/2$. On a nontrivial $\overline{\mathcal P}_n$ character exactly $d^2/2$
nonzero labels act by $-1$. Thus
\begin{equation}\label{eq:discrete-pauli-complement}
 \mathcal H_8\succeq\frac{d^2}{2}I
 \quad\hbox{on }(V^{\overline{\mathcal P}_n})^\perp.
\end{equation}
On $V^{\overline{\mathcal P}_n}$ every $[g_P]$ acts as an involution and its image in
$\overline{\Cl}(n)/\overline{\mathcal P}_n$ is the transvection $t_p$. Accordingly,
$F_{8,P}=(I-\rho(t_p))/2$ on that subspace.

By Lemma~\ref{lem:finite-occurrence} we can embed the finite representation
in the restriction of a $\mathsf{PU}(d)$ representation and use its torus
weights. For a maximal commuting Pauli class $L$, the argument of
Proposition~\ref{prop:exactgap-commuting} supplies integral zero-sum
weights $w=(w_x)$. A vector fixed by $\overline{\mathcal P}_n$ has only components for which
\[
 a_p=\sum_xw_x(-1)^{p\cdot x}\in4\mathbb Z,
 \qquad a_0=0.
\]
Here the divisibility follows from invariance under $[\mathrm iP]$;
it does not require the separate torus components to be fixed by the
whole finite group. On such a component the cyclic defect has eigenvalue
one exactly when $a_p/4$ is odd. By
Lemma~\ref{lem:exactgap-integral-cubic} with $s=2$ and
Lemma~\ref{lem:exactgap-boolean-weight}, a nonzero parity function has
at least $d/4$ nonzero values, none at $p=0$.
Writing $\mathcal H_{8,L}=\sum_{P\in L}F_{8,P}$ therefore gives
\[
 \mathcal H_{8,L}^2-\mathcal H_{8,L}
 \succeq(d/4-1)\mathcal H_{8,L}
 \quad\hbox{as a quadratic form on }V^{\overline{\mathcal P}_n}.
\]
The maximal-isotropic double count from
Proposition~\ref{prop:exactgap-commuting} has $r_1/r_c=(d+2)/2$.
Summing the last inequality over $L$ yields
\[
 \mathcal R_{\mathrm{com},8}\succeq
 \frac{(d-4)(d+2)}8\mathcal H_8
 \quad\hbox{on }V^{\overline{\mathcal P}_n}.
\]
Together with the discrete triangle bound this gives
\begin{equation}\label{eq:discrete-projective-lower}
 \mathcal H_8^2
 =\mathcal H_8+\mathcal R_{\mathrm{ac},8}
                +\mathcal R_{\mathrm{com},8}
 \succeq\frac{d(d-1)}4\mathcal H_8.
\end{equation}
Since $V^{\overline{\mathcal P}_n}$ reduces $\mathcal H_8$ and
\eqref{eq:discrete-pauli-complement} is stronger, the desired lower bound
$d/[4(d+1)]$ follows for $n\ge2$ in every projective Clifford representation.

\subsubsection{Projective attainment and balanced moments}

In the conjugation realization of $\tau_{2,2}$, the vectors
$P\otimes P$, $P\in\mathcal P_n$, span a permutation module for the
nonzero Pauli labels. Signs disappear under the tensor square. A fixed
quarter-turn acts here as $t_p$, and its cyclic average is $(I+t_p)/2$.
The average of all transvections is
\[
 T=\frac{(d^2/2-1)I+A_{\mathcal G}}{d^2-1}.
\]
To check the off-diagonal entries, a label $q$ with $\lambda(p,q)=1$ is
sent to $q+p$; for any anticommuting pair $q,r$ there is exactly one
such $p=q+r$. The remaining $d^2/2-1$ nonzero labels fix $q$.
For $n\ge2$, $A_{\mathcal G}$ has eigenvalue $d/2$ on the nonconstant
space. Its multiplicity is $(d-2)(d+1)/2>0$, as follows from its trace
and the two nonconstant eigenvalues in
Lemma~\ref{lem:exactgap-G-spectrum}. Thus $(I+T)/2$ has an attained
nontrivial eigenvalue
\[
 \frac{3d+4}{4(d+1)},
\]
which proves the matching gap. For $n=1$ the graph is the complete graph
on three labels; its nonconstant eigenvalue is $-1$, giving eigenvalue
$1/2$ of $(I+T)/2$. Transitivity makes the constant line the full invariant
space in this label module. These are genuine representations of the
projective group and hence, by pullback, of $K_n$.
The larger-target assertion in Theorem~\ref{thm:pauli-discrete} follows from the quotient-character observation after Proposition~\ref{prop:unitary-groups}.

\begin{corollary}[Discrete balanced moments]\label{cor:discrete-moments}
For $n\ge2$, relative to the Haar averages of $K_n$ or $\overline{\Cl}(n)$,
\[
 \Delta(\nu_{\mathsf{RPCR}},\tau_{1,1})=
 \frac{d^2}{2(d^2-1)},\qquad
 \Delta(\nu_{\mathsf{RPCR}},\tau_{t,t})=
 \frac{d}{4(d+1)}\quad(t\ge2).
\]
Thus balanced order two is the first balanced order attaining the
projective group gap.
\end{corollary}
\begin{proof}
On a nonidentity Pauli in the adjoint representation, a cyclic average
is the identity for commuting directions and zero for anticommuting
directions. Its average eigenvalue is $(d^2/2-1)/(d^2-1)$.
The invariant part of $\operatorname{End}(\mathbb C^d)$ is the scalar
line, because already the Pauli group has scalar commutant.
For $t\ge2$, tensor the preceding $\tau_{2,2}$ witness with $t-2$
copies of the invariant tensor $\sum_i e_i\otimes\bar e_i$.
This preserves the attained eigenvalue and its orthogonality to the
finite-group invariant space. The group lower bound proves equality.
The order-one value is strictly larger, and the projective image of
$K_n$ is $\overline{\Cl}(n)$, so both finite targets give the same moments.
\end{proof}

\subsection{The orthogonal Pauli rotation gap}\label{subsec:rpor-proof}

Put $N=|\mathcal Y_n|=d(d-1)/2$. In this subsection $\mathcal H$,
$\mathcal H_8$, $\mathcal H^{\mathrm{odd}}$, and
$\mathcal H^{\mathrm{even}}$ denote the same sums as above with indices
restricted to $\mathcal Y_n$. The moment is $I-\mathcal H/N$ and its
fixed space is $V^{\mathsf{SO}(d)}$. Define only the subgroup needed for
the proof,
\[
 K_n^Y=\langle e^{\mathrm i\pi P/4}:P\in\mathcal Y_n\rangle,
 \qquad S=V^{K_n^Y}.
\]
No identification of this subgroup with a particular full normalizer is
required. If $p,q$ have quadratic value one and $\lambda(p,q)=1$, then
$\phi(p+q)=1$ also. Signed conjugation by each generator therefore permutes
$\mathcal Y_n$; all four sums reduce $S$ and $S^\perp$.

\subsubsection{A direct torus argument}

For a fixed label $p$ with $\phi(p)=1$, the cosets of $\langle p\rangle$
in $p^\perp$ each have precisely one element of quadratic value one:
$\phi(x+p)=\phi(x)+1$. There are $d^2/4-1$ distinct commuting labels of
quadratic value one other than $p$. The number of anticommuting neighbors
is consequently $d(d-2)/4$, and each vertex belongs to
$t=d(d-2)/8$ triangles. The discrete local inequality gives
\begin{equation}\label{eq:rpor-complement}
 \mathcal H_8^2\succeq(1+t/2)\mathcal H_8,
 \qquad
 \mathcal H\succeq\left(1+\frac{d(d-2)}{16}\right)I
 \quad\hbox{on }S^\perp.
\end{equation}
Here, as before, the commuting terms are positive and
$\mathcal H\succeq\mathcal H_8$. This lower bound is strictly larger
than $\varepsilon_*:=d(d-8)/16$.

Here a symplectic Lagrangian means an $n$-dimensional subspace of the
$2n$-dimensional label space on which the symplectic form vanishes identically.
To treat $S$, let $W$ be a symplectic Lagrangian in the label space with
$\phi|_W\ne0$, and let $U=\ker(\phi|_W)$, of dimension $n-1$.
Choose a basis $u_1,\ldots,u_{n-1}$ of $U$ and a label $p_0\in W$
of quadratic value one. The corresponding $S_i$ are commuting real
symmetric Pauli involutions, and the Hermitian representative $P_0$
is imaginary and commutes with them. Their joint real eigenspaces
\[
 E_x=\{v:S_i v=(-1)^{x_i}v\ \text{for all }i\},
 \qquad x\in\mathbb F_2^{n-1},
\]
have dimension two. In fact, the trace of each joint spectral projection
$2^{-(n-1)}\prod_i(I+(-1)^{x_i}S_i)$ is $d/2^{n-1}=2$,
since every nonidentity product is a traceless Pauli.
On each $E_x$, $\mathrm iP_0$ is an orthogonal complex structure.
Choose an orthonormal basis in which it is the same two-dimensional
rotation generator $J_2$ on each plane. The signed representatives
\[
 P_p=P_0\prod_i S_i^{p_i},\qquad p\in\mathbb F_2^{n-1},
\]
enumerate the $d/2$ quadratic-value-one labels of $W$, and satisfy
\[
 \mathrm iP_p|_{E_x}=(-1)^{p\cdot x}J_2.
\]
They span the torus of independent rotations of these real two-planes.
A representation of $\mathsf{SO}(d)$ has integral weights
$m=(m_x)\in\mathbb Z^{d/2}$ on this torus, with eigenvalues
\begin{equation}\label{eq:orthogonal-walsh}
 a_p=\sum_xm_x(-1)^{p\cdot x}.
\end{equation}
This is an actual torus in the original orthogonal group. It is not a
restriction of an unspecified $\mathsf{SU}(d)$ representation, and no
spinorial or half-integral torus weights are admitted.

Every torus component of a vector in $S$ has $a_p\in8\mathbb Z$.
Apply Lemma~\ref{lem:exactgap-integral-cubic} with cube dimension $n-1$
and $s=3$. For $n\ge4$, if $a_p/8$ is odd somewhere, it is odd on at
least $2^{n-4}=d/16$ points. With
$k=|\{p:a_p\ne0\}|$ and $k_o=|\{p:a_p/8\text{ is odd}\}|$ this gives
\[
 k(k-1)\ge(d/16-1)k_o.
\]
There is no requirement that $a_0=0$ in this orthogonal calculation:
the index $p=0$ corresponds to $P_0$, not the identity matrix.
The inequality is a quadratic-form inequality on $S$ after summing its
torus components; an individual component need not itself be $K_n^Y$-fixed.

Let $r_1,r_c$ count Lagrangians through one prescribed quadratic-value-one
label or two distinct commuting such labels. The counts depend only on
the dimensions of these isotropic spans: the relevant symplectic
quotients are isomorphic. Every such Lagrangian has nonzero restricted
$\phi$. Counting other commuting labels through the first one gives
\[
 r_1(d/2-1)=r_c(d^2/4-1),\qquad r_1/r_c=(d+2)/2.
\]
Thus, for the ordered commuting sum,
\begin{equation}\label{eq:rpor-commuting}
 \mathcal R_{\mathrm{com}}\succeq
 \frac{(d-16)(d+2)}{32}\mathcal H^{\mathrm{odd}}
 \quad\hbox{on }S.
\end{equation}
Each orthogonal anticommuting triangle has a homomorphism from
$\mathsf{SU}(2)$ generated by its three Pauli directions. Pulling back
the given orthogonal representation is sufficient to apply
Lemma~\ref{lem:exactgap-cl1-local}; an extension to $\mathsf{SU}(d)$
is neither required nor assumed. Summing that lemma gives
\begin{equation}\label{eq:rpor-even}
 \mathcal D_\triangle\succeq\frac{t}{4}\mathcal H^{\mathrm{even}}
 =\frac{d(d-2)}{32}\mathcal H^{\mathrm{even}}
 \quad\hbox{on }S.
\end{equation}
The exact square expansion is
\[
 \mathcal H^2=(1+t/4)\mathcal H+
 \mathcal R_{\mathrm{com}}+\mathcal D_\triangle.
\]
For $d\ge16$, the odd coefficient in \eqref{eq:rpor-commuting} is
nonnegative and no larger than the even coefficient in
\eqref{eq:rpor-even}. Since $\mathcal H=\mathcal H^{\mathrm{odd}}+
\mathcal H^{\mathrm{even}}$ on $S$, we obtain
\[
 \mathcal H^2\succeq
 \left(1+\frac{d(d-2)}{32}+\frac{(d-16)(d+2)}{32}\right)\mathcal H
 =\frac{d(d-8)}{16}\mathcal H.
\]
The complementary estimate \eqref{eq:rpor-complement} completes the
all-representation lower bound. Only sums known to reduce $S$ were used
in passing from compressed forms to eigenvalue bounds.

\subsubsection{Attainment for \texorpdfstring{$n\ge4$}{n at least four}}

Use the same vector $\varpi_n\in\bigwedge^8\mathbb C^d$ from the affine-three-flat
construction in the special-unitary proof. It is fixed by every Pauli quarter-turn,
so $F_P\varpi_n=J_P^2\varpi_n/64$ also for $P\in\mathcal Y_n$.
For $B_{ab}=-\mathrm i(E_{ab}-E_{ba})$, $a<b$, one has
$\Tr(B_{ab}^2)=2$. On an exterior basis vector, the square
of its represented action is one if exactly one of $a,b$ occurs, and
zero otherwise. Hence the Pauli-normalized orthogonal Casimir is
\begin{equation}\label{eq:so-exterior-casimir}
 \sum_{P\in\mathcal Y_n}J_P^2
 =\frac d2\sum_{a<b}J_{B_{ab}}^2
 =\frac d2 k(d-k)I
 \quad\hbox{on }\bigwedge^k\mathbb C^d.
\end{equation}
At $k=8$ this gives $\mathcal H\varpi_n=d(d-8)\varpi_n/16$.
Its positive eigenvalue puts $\varpi_n$ orthogonal to the orthogonal-group
invariants, even when the exterior representation is reducible.
Dividing by $N$ proves \eqref{eq:pauli-real-main} for $n\ge4$.

\subsubsection{The low-dimensional boundaries and triality}\label{subsec:triality}

For $n=1$ the sole circle is the whole of $\mathsf{SO}(2)$, so one step
is already Haar. For $n=2,3$, the continuous triangle bound and positivity
of the commuting terms give
\[
 \mathcal H^2\succeq(1+t/4)\mathcal H,
 \qquad
 \Delta\ge\frac{1+d(d-2)/32}{d(d-1)/2}.
\]
These are respectively $5/24$ and $5/56$, also the values of the
Baer--Haah orthogonal lower bound \cite[Theorem~3.53]{BH26}.
We supply attaining representations rather than continuing the exterior
construction through a trivial determinant representation.

For an auxiliary dimension $D\ge3$, put
\[
 s_4(x)=\sum_{a=1}^D x_a^4,\qquad r^2=\sum_a x_a^2,\qquad
 f_D=s_4-\frac3{D+2}r^4.
\]
This is a nonzero harmonic quartic: its Laplacian is
$12r^2-3\cdot4(D+2)r^2/(D+2)=0$.
If $Q_{ab}$ averages coordinate-plane rotations, then
$Q_{ab}(x_a^4+x_b^4)=\tfrac34(x_a^2+x_b^2)^2$ and $Q_{ab}r^4=r^4$.
Therefore
\begin{equation}\label{eq:kac-quartic-witness}
 \sum_{a<b}(I-Q_{ab})f_D
 =\frac{D-1}{4}s_4-\frac32\sum_{a<b}x_a^2x_b^2
 =\frac{D+2}{4}f_D.
\end{equation}
For $n=2$ the six directions split into the commuting triples
\[
 \{YI,XY,ZY\},\qquad\{IY,YX,YZ\}.
\]
The homomorphism $\mathsf{SO}(4)\to\mathsf{SO}(3)\times\mathsf{SO}(3)$
from the action on the $\pm1$ eigenspaces of the Hodge star on
$\bigwedge^2\mathbb R^4$ (the self-dual and anti-self-dual two-forms) identifies them
with the two coordinate-axis triples. Use harmonic quartics on the first
factor and the trivial representation on the second.
Equation~\eqref{eq:kac-quartic-witness} with $D=3$ gives unnormalized
eigenvalue $5/4$; division by the six sampled axes gives $5/24$.

For $n=3$, choose the seven Hermitian words
\[
 YII,\ XYI,\ XXY,\ XZY,\ ZYX,\ ZIY,\ ZYZ
\]
and write their real skew-symmetric matrices as $\Gamma_a=\mathrm iP_a$
($1\le a\le7$). They satisfy $\Gamma_a^2=-I$ and
$\Gamma_a\Gamma_b=-\Gamma_b\Gamma_a$ for $a\ne b$.
This can be read directly from the words: within the first-$X$ and
first-$Z$ groups the remaining pairs anticommute; between the groups the
remaining pairs commute and the first factors anticommute.
The seven $\Gamma_a$ and twenty-one $\Gamma_a\Gamma_b$ have distinct
Pauli labels and exhaust the twenty-eight skew Pauli directions, up to
sign. For completeness, no nonempty sum of at most four of these labels
vanishes: if its length is even, pair it with one of its own labels;
if its length is odd, pair it with a label outside the sum. In either
case the pairing is one. This proves the required distinctness.

Write $L_{ab}=E_{ab}-E_{ba}$, allowing $L_{ba}=-L_{ab}$.
The linear map
\begin{equation}\label{eq:triality-map}
 \phi(\Gamma_a)=2L_{a8},\qquad
 \phi(\Gamma_a\Gamma_b)=-2L_{ab}\quad(a<b)
\end{equation}
is an isomorphism of Lie algebras. Both sides give bases. Bracket
preservation follows from the Clifford relations: for example,
$[\Gamma_a,\Gamma_b]=2\Gamma_a\Gamma_b$ maps to
$[2L_{a8},2L_{b8}]=-4L_{ab}$, while
$[\Gamma_a\Gamma_b,\Gamma_a]=2\Gamma_b$ and
$[\Gamma_a\Gamma_b,\Gamma_b]=-2\Gamma_a$ give the mixed brackets.
Disjoint pairs commute; pairs sharing one index give the remaining
brackets by the same anticommutation relation. Thus the twenty-eight
Pauli circle directions are carried bijectively to the coordinate-plane
directions. The factors $\pm2$ do not alter their Haar averages.

Let $p:\operatorname{Spin}(8)\to\mathsf{SO}(8)$ be the standard cover.
Integrate $\phi$ to an automorphism $\Phi$ of the simply connected group.
Pull the harmonic-quartic representation back along $p$ and precompose
with $\Phi$. This representation kills the entire center of
$\operatorname{Spin}(8)$: its image under $p$ lies in $\{\pm I\}$,
and degree four is even. Since an automorphism preserves the center,
the new representation still kills $\ker p$ and descends to the
\emph{original} $\mathsf{SO}(8)$. We do not assert that $\phi$ integrates
to an automorphism of that fixed quotient group.
In this $294$-dimensional representation,
\eqref{eq:kac-quartic-witness} gives the Pauli defect eigenvalue
$(8+2)/4=5/2$. It is positive and hence orthogonal to invariants.
Its normalized gap is $(5/2)/28=5/56$, completing
Theorem~\ref{thm:pauli-continuous}.

\subsection{All ordinary tensor gaps of the orthogonal walk}\label{subsec:rpor-hierarchy}

The preceding group theorem identifies the eventual tensor gap. We now
locate the first attaining degree and determine every lower degree as
well. Here tensor degree always means the ordinary representation
$\rho_k(U)=U^{\otimes k}$, not the balanced order.

\begin{lemma}[Tensor occurrence of the boundary witnesses]\label{lem:rpor-degree-eight}
For every $n\ge2$, an orthogonal Pauli vector attaining the group gap
occurs in $(\mathbb C^{2^n})^{\otimes8}$.
\end{lemma}
\begin{proof}
For $n\ge4$ use the exterior-eight vector. For $n=2$, the first
three-dimensional factor above is $\bigwedge^2_+\mathbb C^4$.
Its harmonic quartics lie in $\operatorname{Sym}^4(\bigwedge^2_+\mathbb C^4)$
and hence in the eighth tensor power.

For $n=3$, the automorphisms of type $D_4$ permute its three
 eight-dimensional fundamental representations, namely the vector and
two half-spin representations \cite[Section~2.4]{Baez02}.
The twist in \eqref{eq:triality-map} is not the vector representation:
$\Gamma_a$ has eigenvalues $\pm\mathrm i$, each four times, while
$2L_{a8}$ has eigenvalues $2\mathrm i,-2\mathrm i$ and six zeros.
Thus it sends the vector highest weight to one of the half-spin weights,
say $\omega_s$, and sends harmonic quartics to highest weight $4\omega_s$.
In the usual type-$D_4$ coordinates the two half-spin weights are
$(\varepsilon_1+\varepsilon_2+\varepsilon_3\pm\varepsilon_4)/2$.
The two $\pm1$ Hodge-star eigenspaces of $\bigwedge^4\mathbb C^8$ have highest weights
$\varepsilon_1+\varepsilon_2+\varepsilon_3\pm\varepsilon_4=2\omega_s$:
wedges of the four corresponding weight vectors are highest-weight
vectors. Their squares are nonzero highest-weight vectors of weight
$4\omega_s$ in $\operatorname{Sym}^2(\bigwedge^4_\pm\mathbb C^8)$.
Complete reducibility gives an equivariant copy of the required
irreducible representation there, and this space is contained in the
eighth tensor power. This proves tensor occurrence as well as descent
to the specified orthogonal group.
\end{proof}

\begin{theorem}[Orthogonal gaps in all ordinary tensor degrees]\label{thm:rpor-hierarchy}
Let $d=2^n$, $n\ge1$, $k\ge1$, and let $a_n$ be the orthogonal group
gap in \eqref{eq:pauli-real-main}. Then
\begin{equation}\label{eq:rpor-hierarchy}
 \Delta(\nu_{\mathsf{RPOR}},\rho_k,\mathsf{SO}(d))=
 \begin{cases}
  1,&n=1\text{ or }k\text{ odd},\\
  (d-2)/(2(d-1)),&n\ge2,\ k=2,\\
  1/4,&n=2,\ k=4,6,\\
  (d-4)/(4(d-1)),&n\ge3,\ k=4,6,\\
  a_n,&n\ge2,\ k\ge8\text{ even}.
 \end{cases}
\end{equation}
For $n\ge2$, the first ordinary degree attaining the group gap is
exactly eight. Equivalently, the first balanced order attaining it is four.
\end{theorem}

\begin{proof}
If $n=1$, each step is Haar. In any dimension under consideration every
sampled circle contains $-I$, which belongs to $\mathsf{SO}(d)$.
For odd $k$, both its circle projection and its group invariant projection
vanish, so the representation gap is one. Henceforth take $n\ge2$ and
$k$ even, and set
\[
 \mathcal H_k=\sum_{P\in\mathcal Y_n}(I-Q_P),\qquad
 \mathcal W_k=((\mathbb C^d)^{\otimes k})^{\mathcal P(n;\mathbb R)}.
\]
The real Pauli group lies in $\mathsf{SO}(d)$ and, on even tensors, acts
through its elementary abelian quotient by $\{\pm I\}$.
Its conjugations change only the signs of the sampled axes, so
$\mathcal W_k$ reduces $\mathcal H_k$.
For a nontrivial character $p\mapsto(-1)^{\lambda(u,p)}$, the number of
quadratic-value-one labels on which it is negative is
\begin{equation}\label{eq:quadratic-character-count}
 \#\{p:\phi(p)=1,\ \lambda(u,p)=1\}=
 \begin{cases}d^2/4,&\phi(u)=0,\\ d(d-2)/4,&\phi(u)=1.\end{cases}
\end{equation}
Indeed, $\sum_p(-1)^{\phi(p)}=d$ by summing first over the second half
of the coordinates, and completing the square gives
$\sum_p(-1)^{\phi(p)+\lambda(u,p)}=(-1)^{\phi(u)}d$.
Expanding the two indicator functions proves
\eqref{eq:quadratic-character-count}.
Since $Q_P$ is bounded by the fixed-space projection of $\mathrm iP$,
we have
\begin{equation}\label{eq:rpor-tensor-charge}
 \mathcal H_k\succeq\frac{d(d-2)}4 I
 \quad\hbox{on }\mathcal W_k^\perp.
\end{equation}

Let $T_{rs}$ exchange tensor slots $r,s$ and let $K_{rs}$ contract them
with the standard symmetric bilinear form and reinsert its metric tensor.
The trace-normalized orthogonal Casimir on the full tensor space is
\begin{equation}\label{eq:so-tensor-casimir}
 C_k:=\sum_{P\in\mathcal Y_n}J_P^2
 =\frac d2\left(k(d-1)I+2\sum_{r<s}(T_{rs}-K_{rs})\right).
\end{equation}
For completeness, $\sum_{a<b}B_{ab}^2=(d-1)I$ and
$\sum_{a<b}B_{ab}\otimes B_{ab}=T-K$ for the matrices $B_{ab}$
used in \eqref{eq:so-exterior-casimir}; expanding
$(\sum_r B_{ab}^{(r)})^2$ and multiplying by $d/2$ proves the formula.
For $k=2$, all $J_P$ weights are $0,\pm2$, so
$\mathcal H_2=C_2/4$. The Casimir eigenvalues on
$\bigwedge^2\mathbb C^d$, the traceless symmetric square, and the metric
line are $d(d-2),d^2,0$, respectively. Division by $4N$ gives the
second-degree formula, with attainment in the exterior square.

For $k=4,6$ and $v\in\mathcal W_k$, invariance under
$\mathrm iP=e^{\mathrm i\pi P/2}$ forces the $J_P$ weight components of
$v$ to have weights in $4\mathbb Z$. The only possibilities in $[-k,k]$
are $0,\pm4$. Therefore
\begin{equation}\label{eq:rpor-low-casimir}
 (I-Q_P)v=J_P^2v/16,\qquad
 \mathcal H_k v=C_kv/16\quad(v\in\mathcal W_k).
\end{equation}
The individual $J_P$ need not preserve $\mathcal W_k$; the displayed
vector identity and the Casimir sum are all that is needed.

Let $\mathcal T_k^{\rm tl}$ be the intersection of the kernels of every pairwise
contraction. On this completely traceless tensor space $K_{rs}=0$,
and $T_{rs}\succeq-I$ gives
\begin{equation}\label{eq:traceless-casimir-bound}
 C_k\succeq\frac{dk(d-k)}2 I\quad\hbox{on }\mathcal T_k^{\rm tl}.
\end{equation}
The orthogonal complement of $\mathcal T_k^{\rm tl}$ is the sum of the images
of the metric insertion maps from degree $k-2$.
These maps intertwine the group actions and hence both the Casimirs
and the Pauli averages. If a vector in their sum is Pauli-fixed, average
any tuple of preimages over the finite Pauli group; its averaged
preimages are in $\mathcal W_{k-2}$ and still map to that vector.
The spectrum on this fixed trace part is consequently contained in that
of $\mathcal H_{k-2}|_{\mathcal W_{k-2}}$: the insertion map intertwines
the two self-adjoint operators and is surjective onto the trace part.
This is a spectral consequence of an equivariant surjection, not an
unjustified norm comparison between overlapping insertion images.

Identify $\mathbb C^d\otimes\mathbb C^d$ with
$\operatorname{End}(\mathbb C^d)$ by vectorization. Its Pauli-fixed
space $\mathcal W_2$ is the scalar commutant, namely the metric line,
since the $X_j,Z_j$ generate the full matrix algebra.
At $k=4$, the Pauli-fixed trace part is therefore entirely
$\mathsf{SO}(d)$-invariant. On the traceless part,
\eqref{eq:traceless-casimir-bound} gives $C_4\succeq2d(d-4)I$.
For $d\ge8$ this is attained by a Pauli-fixed vector in
$\bigwedge^4\mathbb C^d$: its fixed dimension is
\eqref{eq:pauli-exterior-four-dimension}, because replacing a Hermitian
Pauli by a real representative changes its fourth exterior action by no
phase. Equation~\eqref{eq:so-exterior-casimir} gives its Casimir value
$2d(d-4)$. Equation~\eqref{eq:rpor-tensor-charge} is stronger on the
other Pauli sectors. Thus the fourth-degree gap is
$[2d(d-4)/16]/N=(d-4)/(4(d-1))$ for $d\ge8$.

For $k=6$ and $d\ge16$, the traceless estimate is $3d(d-6)$,
which exceeds $2d(d-4)$ by $d(d-10)>0$.
The trace part is controlled by degree four, and tensoring its attaining
vector with the metric realizes the fourth-degree eigenvalue in degree
six. This proves the stated sixth-degree value for $n\ge4$.

The remaining boundary $d=8,k=6$ requires a better traceless estimate.
In any unitary representation of the permutation group on $k$ slots, put
\[
 L=\sum_{r<s}\frac{I+T_{rs}}2.
\]
For a three-slot triangle, the corresponding sum $L_\triangle$ has
spectrum $0,3/2,3$, as seen in the sign, standard, and trivial
representations of $\symgp_3$. Thus $L_\triangle^2\succeq(3/2)L_\triangle$.
Each single term occurs in $k-2$ triangles, each pair of terms sharing
a slot occurs in one, and products of disjoint terms are positive.
Expanding the squares therefore gives
\[
 L^2\succeq\sum_\triangle L_\triangle^2-(k-3)L
 \succeq\frac k2L.
\]
Its kernel is exactly the completely antisymmetric space, so
$\sum_{r<s}(I+T_{rs})\succeq kI$ on its orthogonal complement.
Hodge duality identifies $\bigwedge^6\mathbb C^8$ with
$\bigwedge^2\mathbb C^8$ as an $\mathsf{SO}(8)$ representation.
The latter has no Pauli-fixed vector, since $\mathcal W_2$ is the
symmetric metric line. Pauli averaging commutes with the antisymmetrizer,
so $\mathcal W_6$ is orthogonal to $\bigwedge^6\mathbb C^8$.
Consequently $\sum_{r<s}T_{rs}\succeq-9I$ there. On the traceless part,
\[
 C_6\succeq4(6\cdot7-18)I=96I>64I=2d(d-4)I.
\]
The degree-four trace argument applies unchanged, proving the sixth-degree
gap $1/7$ at $n=3$.

For $n=2$ the vector representation of
$\operatorname{Spin}(4)=\mathsf{SU}(2)\times\mathsf{SU}(2)$ is
spin $(1/2,1/2)$. In even tensor degrees at most six, both spins are
integers at most three. The two triples of Pauli directions belong to
the two factors. For one triple, the least defect eigenvalues in spins
$0,1,2,3$ are respectively $0,2,3/2,2$.
Indeed, the three zero-weight unit vectors have overlap $P_j(0)$
for distinct axes; this is zero for $j=1,3$ and $-1/2$ for $j=2$.
The nonzero spectrum of the sum of their rank-one projections is the
spectrum of this three-by-three Gram matrix, giving the asserted
values. The positive minimum across nontrivial pairs of these spins
is $3/2$, and spin $(2,0)$ occurs in
$\operatorname{Sym}^2(\bigwedge^2_+\mathbb C^4)$ inside degree four.
Metric padding realizes it also in degree six. Division by six sampled
axes gives $1/4$ in both degrees.

Finally, Lemma~\ref{lem:rpor-degree-eight} supplies an attaining vector
in degree eight for every $n\ge2$. Metric padding embeds it in every
higher even degree, and the full-group lower bound proves equality.
The degree-two values exceed the degree-four values just computed.
For $n\ge4$,
\[
 \frac{d-4}{4(d-1)}-a_n=\frac d{8(d-1)}>0,
\]
and the two small comparisons are $1/4>5/24$ and $1/7>5/56$.
Odd degrees have gap one. Hence no degree below eight attains the group
gap. Real matrices satisfy $\overline U=U$, so balanced order $t$ is
ordinary degree $2t$, proving the final assertion.
\end{proof}

Because the Pauli moment operators are positive semidefinite, their
power norms in every representation with the gaps above are
\begin{equation}\label{eq:pauli-time-contraction}
 \|M(\nu,\rho,G)^r-\Pi(\rho|G)\|=
 \bigl(1-\Delta(\nu,\rho,G)\bigr)^r\qquad(r\ge1).
\end{equation}
This exponent is $r$, unlike the $2r-1$ exponent for alternating
products of two subgroup projections. Both statements use the same
operator-norm convention, but different one-step operators.

This completes the analysis of how the gap depends on the target and on the
representation class.  The next two sections use the exact Pauli gaps in the
comparison estimates for finite Clifford chains and their Haar analogues.

\section{Exact Clifford chains}\label{sec:clifford-brickwork}

For unitary finite gates, the two-layer step law is the Brickwork Random
Clifford Circuit $\nu_{\mathsf{BRCC}}$ of Baer and Haah. We keep its target
explicit: the actually generated target is $K_n$, while the ambient special
Clifford normalizer is $\Cl(n)$. For the $\mathsf{SO}(4)$-local real finite
law we use $\nu_{\mathsf{BROCC}}$; the $\mathsf O(4)$-local extension is
denoted $\nu_{\mathsf{BROCC}}^{O}$.

For an open chain on $n$ qubits, let $P_i^X$ be the local averaging
projection on edge $(i,i+1)$ for $X\in\{U,O,S\}$. Products over odd and even
edges are denoted $O_X$ and $E_X$, and $\Pi_X$ is the fixed-space projection
for the actual group generated by all supported local gates. In the unitary
projective problem we place a bar over the corresponding target and fixed
projection. All statements below are representation-uniform unless a tensor
power is named explicitly.

\subsection{Exact three-site principal cosines}\label{subsec:clifford-local}

\begin{lemma}[Exact adjacent-edge cosines]\label{lem:three-site-cosines}
After the common fixed space is removed, the maximum principal-angle cosines
for two adjacent local ranges are
\begin{equation}\label{eq:three-site-cosines}
 c_U=\frac25,\qquad c_O=\frac49,\qquad
 c_S=\frac{1+\sqrt{17}}6.
\end{equation}
The first value is for the projective unitary Clifford problem. In the
$-1$ central-character sector of the actual three-qubit unitary generated group the two
local ranges have no common fixed vector and their product norm is exactly
$1/2$.
\end{lemma}
\begin{proof}
For the projective unitary problem, quotient first by the Pauli subgroup.
The three-site quotient is $\Sp(6;\mathbb F_2)$ and the two local
subgroups are copies of $\Sp(4;\mathbb F_2)$. By Frobenius reciprocity \cite[Section~2.1, Step~vii]{BH26}, an irreducible
representation occurs in the induced coset representation exactly when it has
nonzero vectors fixed by the local subgroup. Hence the normalized coset-overlap
calculation captures every nonzero principal cosine. If $D_+$ denotes the integral unsigned squared-overlap
matrix, an exact annihilating polynomial is
\[
 x(x-144)(x-576)(x-2304)(x-14400).
\]
The unsigned incidence graph is connected, so $14400$ is the unique fixed
root, while an explicit polynomial transform of one coset indicator is a
nonzero $2304$-eigenvector. Hence the largest nontrivial squared cosine is
$2304/14400=4/25$.

For the genuine $-1$ central-character sector, the same quotient carries the sign
character on each local $\Sp(4;\mathbb F_2)\cong \symgp_6$. The signed
induced sign-character matrix has squared-overlap annihilating polynomial
\[
 x(x-144)(x-2304)(x-3600),
\]
and an explicit integer vector attains the top root $3600$. Thus the norm is
$\sqrt{3600/14400}=1/2$. There is no common signed fixed space. Other unitary
central characters annihilate every local average.

For the real models, the nontrivial Pauli sectors and joint sectors $V_S$ with $S\ne\varnothing$ are
removed first. On the remaining sector the quotient is
$\OO^+(6;\mathbb F_2)$; spectator completion gives local quotient orders $72$
for the $\mathsf O(4)$-local model and $36$ for the $\mathsf{SO}(4)$-local
orientation kernel. Exact induced
module compressions, together with the same right-translation argument in the regular representation,
give top nontrivial squared cosines
\[
 \frac{16}{81},\qquad \frac{9+\sqrt{17}}{18},
\]
respectively. The corresponding square roots are the last two values in
\eqref{eq:three-site-cosines}. Both are attained in genuine representations;
in particular the ordinary fourth tensor contains matching real witnesses.
The finite matrices and exact arithmetic verification files are included with the
computational material, while the induced-module argument is what upgrades
them from selected moment calculations to all representations.
\end{proof}

\subsection{The unitary Clifford chain}\label{subsec:unitary-brickwork-proof}

Write $\Cl_{i,j}\cong\Cl(2)$ for the local special Clifford group on sites
$i,j$. In this subsection
\[
 P_i^U=\Pi(\rho|\Cl_{i,i+1}),\qquad \Pi_i=I-P_i^U,
\]
so $\Pi_i$ is exactly the edge-defect projector used in the proof of
\cite[Theorem~2.7]{BH26}.

\begin{proposition}[Projective unitary brickwork]\label{prop:unitary-projective-brickwork}
For every $n\ge2$,
\[
 \|\bar E_U\bar O_U-\bar\Pi_U\|=\frac45\cos\frac\pi n,
 \qquad
 \overline\Delta_{\mathrm C,U}(n)=1-\frac45\cos\frac\pi n.
\]
\end{proposition}
\begin{proof}
Put $M_i=\ker P_i^U$ for odd edges and $N_j=\ker P_j^U$ for even edges.
Within each color the local projections commute. A nonadjacent pair has
cosine zero, while Lemma~\ref{lem:three-site-cosines} gives cosine at most
$2/5$ for adjacent pairs. The matrix in
Theorem~\ref{thm:bipartite-angle} is therefore $(2/5)$ times the bipartite
incidence matrix of the path on the $n-1$ edge vertices. Its norm is
$\frac45\cos(\pi/n)$. Taking orthogonal complements converts the angle
between the defect sums to the angle between the two layer ranges, so the
same number bounds the product norm.

It remains to attain the bound. In the balanced second moment, let $F_j^{\rm dw}$ be
the operator that swaps the two replicas on the first $j$ sites and acts
trivially on the remaining sites, $0\le j\le n$. Local Clifford and Haar
second moments agree, and direct two-site projection gives
\[
 \Pi_iF_j^{\rm dw}=0\quad(j\ne i),\qquad
 \Pi_iF_i^{\rm dw}=F_i^{\rm dw}-\frac25(F_{i-1}^{\rm dw}+F_{i+1}^{\rm dw}).
\]
Thus on the domain-wall space the edge Hamiltonian has the positive
eigenvalue
$1-(4/5)\cos(\pi/n)$ with the sine-wave eigenvector corrected by endpoint
components. Same-parity defects have orthogonal ranges on this space, so the
two-projection canonical form turns that eigenvalue into the singular value
$\frac45\cos(\pi/n)$ of $\bar E_U\bar O_U$. The eigenvector has positive eigenvalue and is
therefore orthogonal to the complete fixed space. This proves attainment.
For $n=2$ the same formula gives norm zero, as the single local average is
already the global one.
\end{proof}

\begin{proposition}[Unitary brickwork for the actually generated target]\label{prop:unitary-actual-brickwork}
For the actual group $K_n$ generated by the adjacent special two-qubit
Clifford groups,
\[
 \Delta_{\mathrm C,U}(n)=
 \begin{cases}
 1,&n=2,\\[1mm]
 1/2,&n=3,\\[1mm]
 1-(4/5)\cos(\pi/n),&n\ge4.
 \end{cases}
\]
If instead the target is the larger special normalizer, its gap is zero for
$n\ge4$.
\end{proposition}
\begin{proof}
Decompose by the character of the global scalar $\zeta_8I$ contained in the
actual generated group. The trivial character descends to the projective
problem. If the character is neither $+1$ nor $-1$, a local scalar acts
nontrivially and every local average vanishes. In the $-1$ sector there is no
three-site common fixed vector and Lemma~\ref{lem:three-site-cosines} gives
$\|P_i^UP_{i+1}^U\|\le1/2$. Commuting factors inside each layer allow one
adjacent odd-even pair to be moved next to each other, while all remaining
factors are contractions; hence the whole brickwork norm is at most $1/2$ in
this sector. The signed three-site calculation attains $1/2$ when $n=3$.
For $n\ge4$, the projective value
$(4/5)\cos(\pi/n)\ge2\sqrt2/5>1/2$ is slower and is attained by the same
balanced-second-moment witness, whose center action is trivial.

The larger normalizer contains scalar phases not generated by the local
walk when $n\ge4$. A nontrivial character of the quotient by $K_n$ is then
fixed by every step but not by Haar measure on the larger target, so the gap
relative to that target is zero. This is why the target cannot be silently
enlarged.
\end{proof}

\subsection{The exact random-edge path and positive rates}
\label{subsec:unitary-path-proof}

\begin{theorem}[The positive-rate unitary path]\label{thm:unitary-positive-path}
Let $w_i>0$, $1\le i\le n-1$, and let
$H(w)=\sum_iw_i\Pi_i$ on an open chain.  Define
\[
 J(w)_{ii}=w_i,\qquad
 J(w)_{i,i+1}=J(w)_{i+1,i}=-\frac25\sqrt{w_iw_{i+1}},
\]
with all other entries zero.  The projective all-representation Hamiltonian
gap is $\lambda_{
\min}(J(w))$.  For the actual generated group it has the same value if
$n\ge4$; at $n=2$ it is $w_1$, and at $n=3$ it is
\[
 \frac{w_1+w_2-\sqrt{(w_1-w_2)^2+w_1w_2}}2.
\]
For a probability-normalized random-edge step, divide these Hamiltonian
gaps by $\sum_iw_i$.
\end{theorem}
\begin{proof}
Put $a=2/5$ and $h_i=\Pi_i$.  Lemma~\ref{lem:three-site-cosines}
and the two-projection canonical form imply, for every $\alpha>0$,
\[
 h_ih_{i+1}+h_{i+1}h_i
 \succeq-a(\alpha h_i+\alpha^{-1}h_{i+1})
\]
in every projective representation.  Common-range and commuting blocks
satisfy the same inequality.  Nonadjacent defects commute and their
products are positive.  Since the matrix with diagonal one and adjacent
entries $-a$ is positive definite, $J(w)$ is positive definite.  Its least
eigenvalue $\lambda$ has a strictly positive eigenvector $v$; this follows
from the Perron--Frobenius theorem applied to a sufficiently large scalar shift;
see, for example, \cite[Theorem~2.2]{BoylePF}.
In the inequality for adjacent edges choose
\[
 \alpha_i=\sqrt{w_i/w_{i+1}}\,v_{i+1}/v_i.
\]
Expanding $H(w)^2$ and discarding nonadjacent positive terms yields
$H(w)^2\succeq\lambda H(w)$, because the coefficient of $h_i$ is
\[
 w_i\left[w_i-a\sqrt{w_iw_{i-1}}\frac{v_{i-1}}{v_i}
               -a\sqrt{w_iw_{i+1}}\frac{v_{i+1}}{v_i}\right]
 =\lambda w_i.
\]
Missing endpoint terms are zero.  This is a lower bound on every
projective representation, not just on a tensor moment.

For attainment use the domain-wall operators $F_j^{\rm dw}$ in
Proposition~\ref{prop:unitary-projective-brickwork}.  Their exact action is
\[
 H(w)F_j^{\rm dw}=w_j(F_j^{\rm dw}-aF_{j-1}^{\rm dw}-aF_{j+1}^{\rm dw})\quad(1\le j\le n-1),
 \qquad H(w)F_0^{\rm dw}=H(w)F_n^{\rm dw}=0.
\]
Set $\beta_i=v_i/\sqrt{w_i}$.  Then the nonzero vector
\[
 u=\sum_{j=1}^{n-1}\beta_jF_j^{\rm dw}
   -\frac{aw_1\beta_1}{\lambda}F_0^{\rm dw}
   -\frac{aw_{n-1}\beta_{n-1}}{\lambda}F_n^{\rm dw}
\]
satisfies $H(w)u=\lambda u$.  Linear independence of the domain-wall
tensor words proves it is nonzero; self-adjointness makes it orthogonal
to the entire fixed space.  Its central action is trivial.

For the actual target, all central characters except the trivial one
and the $-1$ central-character one annihilate each $P_i^U$.  In the $-1$ central-character
sector an adjacent pair has no common fixed vector and cosine at most
$1/2$.  Hence
\[
 w_ih_i+w_{i+1}h_{i+1}\succeq m(w_i,w_{i+1})I,\qquad
 m(s,t)=\frac{s+t-\sqrt{(s-t)^2+st}}2.
\]
At $n=3$ the signed local calculation attains this bound.  For $n\ge4$
it suffices to show that a three-edge principal submatrix $J_3$ has
least eigenvalue smaller than
$t=\max\{m(w_1,w_2),m(w_2,w_3)\}$.
If some $w_i\le t$, irreducibility gives
$\lambda_{\min}(J_3)<\min_iw_i\le t$.  Otherwise put
$x_i=1-t/w_i>0$.  The two pair inequalities give
$x_1x_2\le1/4$ and $x_2x_3\le1/4$.  Therefore
\[
 x_2\le\frac1{4\max(x_1,x_3)}
 <\frac4{25}\left(\frac1{x_1}+\frac1{x_3}\right).
\]
After diagonal congruence, $J_3-tI$ has diagonal $x_i$ and adjacent
entries $-2/5$, so its central Schur complement is negative.
Interlacing gives $\lambda_{\min}(J(w))<t$, while the $-1$ central-character
Hamiltonian is at least $tI$.  Thus it cannot improve on the projective
slow mode.  The $n=2$ statement is one local Haar average.
\end{proof}

\begin{corollary}[Uniform path sampling]\label{cor:unitary-uniform-path}
For $n\ge2$ the projective random-edge gap is
\[
 \overline\delta_{\mathrm C,U}^{\rm path}(n)
 =\frac{1-(4/5)\cos(\pi/n)}{n-1}.
\]
For the actual target the formula holds for $n\ge4$, with values $1$ and
$1/4$ at $n=2,3$.  The formula is an equality for finite Clifford targets;
it is not an unrestricted Haar-circuit equality.
\end{corollary}
\begin{proof}
At unit rates, $J$ has eigenvalues $1-(4/5)\cos(k\pi/n)$,
$1\le k\le n-1$, by the discrete sine basis.  Divide its least
eigenvalue by the $n-1$ edges and use the actual-target exceptions.
\end{proof}

\subsection{\texorpdfstring{$\mathsf O(4)$}{O(4)}-local Clifford gates}\label{subsec:full-real-brickwork-proof}

\begin{proposition}[Representation-uniform $\mathsf O(4)$-local chain bound]\label{prop:full-real-chain}
For every $n\ge3$,
\begin{equation}\label{eq:jacobi-main}
 \Delta_{\mathrm C,O}(n)\ge
 b_n:=1-\frac89\cos\frac\pi n.
\end{equation}
\end{proposition}
\begin{proof}
Remove the common fixed space.  Nonadjacent edge projections commute.
Lemma~\ref{lem:three-site-cosines} gives Friedrichs cosine $4/9$ on adjacent
edges.  Hence the comparison matrix in Theorem~\ref{thm:bipartite-angle} is
$(4/9)$ times the odd/even incidence matrix of the open path.  Its norm is
$2\cos(\pi/n)$, so the nontrivial layer-product norm is at most
$(8/9)\cos(\pi/n)$.
\end{proof}

\subsection{\texorpdfstring{$\mathsf{SO}(4)$}{SO(4)}-local Clifford gates}\label{subsec:strict-real-brickwork-proof}

\begin{lemma}[Adjacent $\mathsf{SO}(4)$-local/$\mathsf O(4)$-local energy comparison]\label{lem:strict-full-pair}
For every representation of the actual real generated group and every two
adjacent edges,
\begin{equation}\label{eq:strict-full-pair}
 h_i^S+h_{i+1}^S\succeq
 \frac{5-\sqrt{17}}{12}(h_i^O+h_{i+1}^O).
\end{equation}
\end{lemma}
\begin{proof}
On the three-site geometry the two local ensembles generate the same global
real Clifford group by Proposition~\ref{prop:real-generation}; let $\Pi$ be its
fixed projection.  Lemma~\ref{lem:three-site-cosines} gives
$c_S=(1+\sqrt{17})/6$.  The two-projection decomposition therefore gives
\[
 h_i^S+h_{i+1}^S\succeq(1-c_S)(I-\Pi).
\]
Each full-real edge defect annihilates $\Pi$ and is bounded by $I-\Pi$, so
$h_i^O+h_{i+1}^O\preceq2(I-\Pi)$.  Combining the two estimates proves
\eqref{eq:strict-full-pair}.
\end{proof}

\begin{proposition}[Representation-uniform $\mathsf{SO}(4)$-local chain bound]
\label{prop:strict-real-chain}
For $n\ge3$, set
\begin{equation}\label{eq:strict-energy-definition}
 c_n=\frac89\cos\frac\pi n,\qquad
 t_n=\begin{cases}1,&n\text{ odd},\\(n-2)/n,&n\text{ even},\end{cases}
\end{equation}
\[
 \Gamma_n=\frac{5-\sqrt{17}}{24}
 \left(1+t_n-\sqrt{(1-t_n)^2+4t_nc_n^2}\right).
\]
Then
\begin{equation}\label{eq:strict-finite-chain}
 \Delta_{\mathrm C,S}(n)\ge\Psi(\Gamma_n).
\end{equation}
\end{proposition}
\begin{proof}
Apply Lemma~\ref{lem:weighted-path-cover} with
$\beta=(5-\sqrt{17})/12$ from Lemma~\ref{lem:strict-full-pair} and with the
full-real layer cosine bounded by $c_n$ from
Proposition~\ref{prop:full-real-chain}.  This gives
\[
 \sum_i h_i^S\succeq\Gamma_n(I-\Pi_S).
\]
The strict defects commute within each parity class and each defect fails to
commute with at most two of the opposite color.  Theorem~\ref{thm:two-colour-energy}
therefore gives the layer-product gap \eqref{eq:strict-finite-chain}.
\end{proof}

This completes the finite-chain analysis.  Section~\ref{sec:haar-brickwork}
keeps the same graph geometry while replacing the Clifford fixed spaces by the
smaller Haar fixed spaces.

\section{From Clifford to Haar circuits}\label{sec:haar-brickwork}

For the unitary model the two-layer Haar law is the Baer--Haah Brickwork
Random Unitary Circuit $\nu_{\mathsf{BRUC}}$. For the
$\mathsf{SO}(4)$-local real Haar law we use $\nu_{\mathsf{BROC}}$, while
the $\mathsf O(4)$-local extension is denoted $\nu_{\mathsf{BROC}}^{O}$.

We now replace each finite Clifford layer by its Haar analogue.  The Haar fixed
range is nested in the Clifford range; a Pauli rotation circle measures the
remaining finite invariant component, and Theorem~\ref{thm:nested-product}
transfers the finite gap to the continuous circuit.

\subsection{Compression by a Pauli circle}\label{subsec:haar-compression}

For one parity layer write $P_\ell$ for the finite Clifford layer projection
and $Q_\ell$ for the corresponding Haar layer projection, with
$\ell\in\{o,e\}$.  Let $A$ be the fixed projection of the finite group
generated by the local Clifford gates and let $R$ be the fixed projection of
the continuous group generated by the Haar local gates.  Thus $R\le A$ and
$Q_\ell\le P_\ell$.

\begin{lemma}[Pauli-circle compression]\label{lem:haar-circle-compression}
On $R^\perp$ one has, separately for each parity layer,
\begin{equation}\label{eq:haar-circle-compression}
 A Q_\ell A\preceq (1-a_X(n))A,
\end{equation}
where $X=U$ for unitary gates and $X=R$ for either real local ensemble, and
\[
 a_U(3)=\frac5{63},\qquad a_R(3)=\frac5{56},\qquad
 a_U(n)=a_R(n)=\frac{2^n-8}{8(2^n-1)}\quad(n\ge4).
\]
\end{lemma}
\begin{proof}
Choose one edge of the layer.  The local Haar group contains a Pauli
rotation circle on that edge; in the real models choose an anisotropic
odd-$Y$ direction.  Let $B$ be the orthogonal projection obtained by averaging
that circle.  Since every Haar-layer fixed vector is fixed by this circle,
$Q_\ell\le B$ and therefore $AQ_\ell A\preceq ABA$.

The finite generated group acts transitively on the relevant Pauli-circle
orbit: all nonzero Pauli directions in the unitary case and all anisotropic
real Pauli directions in the orthogonal case.  Because $A$ is the average of
that finite group,
\[
 ABA=A\left(\frac1{|\mathcal O|}\sum_{B'\in\mathcal O}B'\right)A.
\]
The operator in parentheses is exactly the random-Pauli-rotation moment
operator on the corresponding compact target.  Restricted to $R^\perp$, its
norm on the finite invariant range is at most $1-a_X(n)$ by the exact Pauli
gaps of Section~\ref{sec:pauli-proofs}.  This proves
\eqref{eq:haar-circle-compression}.  The argument is an individual
compression estimate for each layer; no commutation of the two Haar layers
is used.
\end{proof}

\subsection{Unitary and \texorpdfstring{$\mathsf O(4)$}{O(4)}-local lower bounds}\label{subsec:haar-unitary-full}

Put
\[
 u_3=\frac12,\qquad
 u_n=1-\frac45\cos\frac\pi n\quad(n\ge4),\qquad
 b_n=1-\frac89\cos\frac\pi n\quad(n\ge3).
\]

\begin{proposition}[Finite-to-Haar transfer]\label{prop:haar-unitary-full-transfer}
For every $n\ge3$,
\begin{equation}\label{eq:haar-unitary-full-transfer}
 \Delta_{\mathrm H,U}(n)\ge a_U(n)u_n,
 \qquad
 \Delta_{\mathrm H,O}(n)\ge a_R(n)b_n.
\end{equation}
\end{proposition}
\begin{proof}
Apply Theorem~\ref{thm:nested-product} to the finite and Haar parity-layer
projections.  Proposition~\ref{prop:unitary-actual-brickwork} gives the unitary
finite gap $u_n$, Proposition~\ref{prop:full-real-chain} gives the full-real
gap $b_n$, and Lemma~\ref{lem:haar-circle-compression} supplies the factors
$a_U(n)$ and $a_R(n)$.  The target projectors are kept throughout; no equality
of finite and Haar global fixed spaces is used.
\end{proof}

\begin{corollary}[Uniform unitary and $\mathsf O(4)$-local constants]\label{cor:haar-unitary-full-uniform}
For every $n\ge3$,
\begin{equation}\label{eq:haar-unitary-full-uniform}
 \Delta_{\mathrm H,U}(n)>\frac1{40},
 \qquad
 \Delta_{\mathrm H,O}(n)>\frac1{135}.
\end{equation}
\end{corollary}
\begin{proof}
At $n=3$, the transfer gives
$a_U(3)u_3=5/126>1/40$ and
$a_R(3)b_3=(5/56)(5/9)=25/504>1/135$.
For $n\ge4$, concavity of $\sin$ on $[0,\pi/6]$ gives
\[
 1-\cos\frac\pi n=2\sin^2\frac{\pi}{2n}\ge\frac9{2n^2},
 \qquad
 u_n\ge\frac15+\frac{18}{5n^2}.
\]
Writing $a_n=(2^n-8)/(8(2^n-1))$, the difference
$a_nu_n-1/40$ has the sign of
$18(2^n-8)-7n^2$.  This is positive at $n=4$, and its successive
difference is $18\,2^n-7(2n+1)>0$ for $n\ge4$.  Thus
$\Delta_{\mathrm H,U}(n)>1/40$.  For the real model,
$a_R(n)\ge a_R(4)=1/15$ and $b_n>1/9$, hence
$a_R(n)b_n>1/135$.
\end{proof}

\subsection{\texorpdfstring{$\mathsf{SO}(4)$}{SO(4)}-local Haar gates}\label{subsec:haar-strict}

\begin{proposition}[Pointwise $\mathsf{SO}(4)$-local Haar bound]\label{prop:haar-strict-pointwise}
For every $n\ge3$, with $\Gamma_n$ as in
Proposition~\ref{prop:strict-real-chain},
\begin{equation}\label{eq:haar-strict-pointwise-proof}
 \Delta_{\mathrm H,S}(n)\ge a_R(n)\Psi(\Gamma_n).
\end{equation}
\end{proposition}
\begin{proof}
Proposition~\ref{prop:strict-real-chain} gives the finite strict layer-product
gap $\Psi(\Gamma_n)$.  Apply Theorem~\ref{thm:nested-product} to the strict
finite and Haar layers and Lemma~\ref{lem:haar-circle-compression} with the
real Pauli compression factor $a_R(n)$.
\end{proof}

\begin{corollary}[Uniform $\mathsf{SO}(4)$-local Haar constant]\label{cor:haar-strict-uniform}
For all $n\ge3$,
\begin{equation}\label{eq:haar-strict-uniform-proof}
 \Delta_{\mathrm H,S}(n)>\frac1{3300}.
\end{equation}
\end{corollary}
\begin{proof}
For $n\ge5$, $t_n\ge2/3$ and $c_n<8/9$.  Monotonicity of the weighted
two-projection positive eigenvalue gives
\[
 \Gamma_n>\gamma_0:=
 \frac{(5-\sqrt{17})(45-\sqrt{1617})}{648}.
\]
The rational bounds $\sqrt{17}<33/8$ and $\sqrt{1617}<403/10$ imply
$\gamma_0>329/51840$.  Since
\[
 \Psi(x)=\frac{x}{1+x+\sqrt{1+2x}}\ge\frac{x}{2(1+x)}
\]
and $a_R(n)\ge a_R(5)=3/31$,
\[
 \Delta_{\mathrm H,S}(n)>
 \frac{987}{3234478}>\frac1{3300}.
\]
For $n=4$,
\[
 \Gamma_4=\frac{5-\sqrt{17}}{24}
 \left(\frac32-\frac{\sqrt{337}}{18}\right)>\frac1{109},
\]
using $\sqrt{337}<37/2$; hence
$a_R(4)\Psi(\Gamma_4)>(1/15)(1/220)=1/3300$.
For $n=3$, the exact two-projection value
$\Delta_{\mathrm C,S}(3)=(5-\sqrt{17})/6>7/48$ gives
$\Delta_{\mathrm H,S}(3)>(5/56)(7/48)>1/3300$.
\end{proof}

\subsection{Three-site continuous slow modes and the unitary reduction}\label{subsec:haar-slow-modes}

These lower bounds do not make a low tensor moment exact: already on three
sites, compact-group representations can be slower than the finite-Clifford
low-moment values.

\begin{lemma}[Adding spectator endpoint Haar averages preserves the product]\label{lem:haar-endpoint-completion}
On three sites let $Q_A,Q_B$ be the Haar averages on edges $12$ and $23$.
Let $X$ be the one-site Haar average on site $3$ and $Y$ the one-site Haar
average on site $1$, using $\mathsf{SU}(2)$ in the unitary case and
$\mathsf O(2)$ in the real cases.  Then
\begin{equation}\label{eq:haar-endpoint-completion}
 (Q_AX)(Q_BY)=Q_AQ_B.
\end{equation}
Consequently the added spectator averages preserve every nonzero singular value of
the original two-layer product.
\end{lemma}
\begin{proof}
The spectator average $X$ commutes with $Q_A$ and its subgroup is contained
in the edge-$23$ Haar group, so $XQ_B=Q_B$.  Similarly $Y$ commutes with
$Q_B$ and $Q_AY=Q_A$.  These identities give
\eqref{eq:haar-endpoint-completion}.  The completed factors are again
orthogonal projections, and their product is literally the same operator.
\end{proof}

\begin{proposition}[A center-trivial representation with principal cosine $1/2$]\label{prop:haar-unitary-slow-mode}
For the three-site Haar $\mathsf{SU}(4)$ circuit there is a center-trivial
representation, contained in
$\operatorname{End}(\bigwedge^4\mathbb C^8)\subset\tau_{4,4}$, in which two
adjacent local Haar fixed spaces have a nonfixed \emph{positive} principal
cosine $1/2$.  In particular
\begin{equation}\label{eq:haar-unitary-three-witness}
 \Delta_{\mathrm H,U}(3)\le\frac12.
\end{equation}
\end{proposition}
\begin{proof}
Put $W=\bigwedge^4\mathbb C^8$ and use the conjugation representation on
$\operatorname{End}(W)$.  After endpoint
completion, restrict $W$ under $\mathsf{SU}(4)\times\mathsf{SU}(2)$.  The
displayed decomposition is multiplicity free and has dimensions
\[
 (20,1)\oplus(15,3)\oplus(1,5),
\]
corresponding to $\mathsf{SU}(4)$ shapes $(2,2),(2,1,1),(1,1,1,1)$ and
spectator spins $0,1,2$.  Hence each completed local commutant in
$\operatorname{End}(W)$ is three-dimensional.

Let $J_X,J_Y,J_Z$ be the spectator infinitesimal generators on $W$ and set
$K=J_X^2+J_Y^2+J_Z^2$.  Its eigenvalues on these three summands are
$0,8,24$, so the orthogonal sector projections are
\[
 E_0=\frac{(K-8)(K-24)}{192},\qquad
 E_1=-\frac{K(K-24)}{128},\qquad
 E_2=\frac{K(K-8)}{384}.
\]
For the two endpoint choices, the Hilbert--Schmidt Gram matrix and cross
Gram matrix in this basis are
\[
 G=\operatorname{diag}(20,45,5),\qquad
 C=\begin{pmatrix}6&9&5\\9&36&0\\5&0&0\end{pmatrix}.
\]
Therefore the squared principal-angle matrix
$G^{-1}CG^{-1}C^{\mathsf T}$ has characteristic polynomial
\[
 (x-1)(x-1/4)(x-4/25).
\]
The unit root is the global invariant.  More explicitly,
$w=(3,-2,6)^{\mathsf T}$ satisfies $Cw=\tfrac12Gw$ and
$w^{\mathsf T}G(1,1,1)^{\mathsf T}=0$.  Thus there is a nonfixed positive principal cosine $1/2$.  The value $1/2$
was previously observed numerically by Haferkamp and Hunter-Jones
\cite[p.~18]{HHJ20}; the calculation here gives an exact
representation-level witness.
\end{proof}

\begin{proposition}[All weighted three-site unitary Haar gaps]
\label{prop:haar-three-site-reduction}
Let $Q_{12},Q_{23},Q_{13}$ be the three local Haar $\SU(4)$ averages on
three qubits, and keep the target fixed as $\SU(8)$ or $\PU(8)$.  There is a
constant $c_*$ satisfying \eqref{eq:B-cstar-range} such that the gap of
$pQ_{12}+qQ_{23}+rQ_{13}$ is exactly \eqref{eq:B-weighted-haar} for every
$p,q,r\ge0$ of sum one.  The brickwork norm is $c_*$.
\end{proposition}
\begin{proof}
We first work in an arbitrary projective representation of $\SU(8)$.  Let
$A$ be the global fixed projection, put $E=Q_{12}-A$, and set
$D=\operatorname{ran}E$.  Physical permutations of the three qubits form a
unitary $\symgp_3$ action.  The transposition $(12)$ fixes $D$ pointwise: on the
local four-dimensional edge space the swap has determinant $-1$, so
$e^{\mathrm i\pi/4}\mathrm{SWAP}\in\SU(4)$, while after embedding into
three qubits the compensating scalar is central and hence trivial in a
projective representation.

Let $H=\{1,(12)\}$ and choose $g_1=1$, $g_2=(13)$, $g_3=(23)$.  The maps
$J_i:D\to A^\perp$, $J_iv=\rho(g_i)v$, identify $D$ isometrically with the
three edge-fixed ranges after removal of $A$.  Since $H$ has exactly two
double cosets in $\symgp_3$, all off-diagonal range Gram blocks are the same
self-adjoint operator
\[
 S=E\rho((13))E\big|_D,
 \qquad
 J_i^*J_j=\begin{cases}I_D,&i=j,\\ S,&i\ne j.\end{cases}
\]
If $P_{\rm sym}$ denotes the $\symgp_3$-average, the elementary group-algebra
identity
\[
 \frac{1+(12)}2\,(13)\,\frac{1+(12)}2
 =\frac32\left(\frac1{6}\sum_{g\in \symgp_3}g\right)
 -\frac12\frac{1+(12)}2
\]
gives
\begin{equation}\label{eq:haar-s3-compression}
 S=\frac32EP_{\rm sym}E-\frac12I_D,
 \qquad -\frac12I_D\preceq S\preceq I_D.
\end{equation}
The eigenvalue $1$ cannot occur on $D$: equality would make a vector fixed by
$\symgp_3$ and by one edge $\SU(4)$, hence by two adjacent edge groups and therefore
by all of $\SU(8)$, contrary to $D\subset A^\perp$.

Let $c_*$ be the supremum, over projective irreducible representations, of the
positive spectrum of $S$.  Proposition~\ref{prop:haar-unitary-slow-mode}
shows $c_*\ge1/2$.  To obtain the opposite bound without identifying a
particular gap-attaining representation, let $\Pi$ be the global projective finite-
Clifford fixed projection and $A$ the global Haar fixed projection.  The
projective finite three-site brickwork gap is $3/5$ by
Theorem~\ref{thm:clifford-main}, so its nontrivial layer-product norm is $2/5$.
The exact $n=3$ projective Pauli-rotation gap $5/63$, together with finite-
Clifford transitivity on the nonidentity Pauli circles, gives
$\Pi Q_i\Pi\preceq(58/63)\Pi$ on $A^\perp$ for each Haar edge projection
$Q_i$ (each edge Haar group contains the corresponding Pauli circle).
The nested-product comparison therefore gives the projective Haar brickwork
gap at least $(5/63)(3/5)=1/21$.  Hence $c_*\le20/21$.

Diagonalize $S$.  For an eigenvector with eigenvalue $s$, the three transported
vectors $J_iv$ have Gram matrix with diagonal entries one and off-diagonal
entries $s$.  Hence the nonzero spectrum of the weighted operator on the
corresponding three-dimensional range block is the spectrum of
\[
 B_s(p,q,r)=
 \begin{pmatrix}
 p&s\sqrt{pq}&s\sqrt{pr}\\
 s\sqrt{pq}&q&s\sqrt{qr}\\
 s\sqrt{pr}&s\sqrt{qr}&r
 \end{pmatrix}.
\]
By \eqref{eq:haar-s3-compression}, $s\ge-1/2$, while $c_*\ge1/2$.  Entrywise
absolute values and Perron monotonicity therefore give
$\lambda_{\max}(B_s)\le\phi_{c_*}(p,q,r)$.  Conversely, by definition of
$c_*$ there are projective representations with positive eigenvalues
$s_k\to c_*$; continuity of the largest eigenvalue gives the reverse
all-representation bound.  Thus the projective norm is exactly
$\phi_{c_*}$, including zero-weight boundary cases.

It remains to compare the full $\SU(8)$ target.  If the scalar
$e^{2\pi i/8}I$ acts with central character $k$, a local $\SU(4)$ contains
its square, so all three edge projections vanish unless $k=0$ or $4$.
The $k=0$ representations are precisely the projective ones.  In the $k=4$
sector there is no global fixed vector.  Each Haar edge range is contained in
the corresponding actual finite-Clifford edge range, whose pairwise
three-site principal cosine is at most $1/2$ by
Theorem~\ref{thm:clifford-main}.  Applying the three-range Gram embedding to
these three subspaces gives, for every probability vector $(p,q,r)$,
$\|pQ_{12}+qQ_{23}+rQ_{13}\|\le\phi_{1/2}(p,q,r)\le\phi_{c_*}(p,q,r)$.
Consequently the $\SU(8)$ and $\PU(8)$ all-representation gaps coincide.

Finally, for two adjacent projections the all-representation product norm is
the supremum of the absolute principal cosines.  The negative spectrum in
\eqref{eq:haar-s3-compression} has magnitude at most $1/2$, whereas the
positive supremum is $c_*\ge1/2$.  Thus the brickwork nontrivial norm is
exactly $c_*$.  Expanding $\det(xI-B_c)$ gives
\eqref{eq:B-weighted-haar-cubic}.
\end{proof}

For $z=(\sqrt p,\sqrt q,\sqrt r)^{\mathsf T}$,
$z^{\mathsf T}B_{c_*}z=c_*+(1-c_*)(p^2+q^2+r^2)$.  Hence uniform
three-edge sampling uniquely maximizes the gap; the formulas
\eqref{eq:B-cstar-consequences} and the bounds in
\eqref{eq:B-three-site-bracket} follow from $1/2\le c_*\le20/21$.
Together with Proposition~\ref{prop:haar-three-site-reduction} this proves
Theorem~\ref{thm:B-three-site-haar}.  Determining whether $c_*=1/2$ remains
equivalent to the positive comparison $EP_{\rm sym}E\preceq(2/3)E$ in every
nontrivial projective representation.

We now leave all-representation gaps and fix the ordinary fourth tensor of the
real complete-graph walk, where a sharper finite-$n$ reduction is possible.

\section{The entire real fourth moment and its gap-attaining representation}
\label{sec:real-fourth}\label{subsec:real-fourth-alln}

We first locate every sector that can attain the fourth-tensor gap, then use a
nonorthogonal quotient to identify its gap eigenvector.  Strict exclusion of
the remaining sectors proves algebraic simplicity on the entire fourth tensor.

\subsection{Local fixed spaces and the occupation operator}
\label{subsec:fourth-local-model}

Regroup the tensor factors as
\[
 \mathcal T_n=(\mathbb C^{2^n})^{\otimes4}
 \cong\bigl((\mathbb C^2)^{\otimes4}\bigr)^{\otimes n}.
\]
For $X=O,S$, average an edge gate over, respectively, the real Pauli
normalizer in $\mathsf O(4)$ or in $\mathsf{SO}(4)$, or over the
corresponding Haar group.  Denote the resulting ordinary-fourth-tensor
projection by $P_{ij}^{X,(4)}$ and put
\[
 T_{X,n}^{(4)}=\frac1{\binom n2}\sum_{i<j}P_{ij}^{X,(4)},\qquad
 H_{X,n}=I-T_{X,n}^{(4)},\qquad
 \widehat H_{X,n}=\binom n2 H_{X,n}.
\]
The fixed projector is always that of the group generated by these gates.
For $n\ge3$ the Haar target is $\mathsf{SO}(2^n)$ for both real local
ensembles.  The two local distributions need not coincide.

Here is a concrete realization of the one-site fixed space.  For
$R\in\{I,X,Z,XZ\}$, use the unit vectors
$\eta_R=\operatorname{vec}(R)\otimes\operatorname{vec}(R)/2$, with tensor
slots regrouped as above.  They form an orthonormal basis of the one-site
real-Pauli fixed space: the Pauli average has rank four, as follows either
from its trace or from direct averaging of the four tensor actions.
The one-site real Clifford group permutes $X$ and $Z$, leaving the other
two labels fixed.  Its fourth-tensor fixed space is therefore
\[
 V=\operatorname{span}\{\eta_I,\eta_X+\eta_Z,\eta_{XZ}\}
 \cong\mathbb C^3,\qquad
 \langle x,y\rangle_V=\overline x_0y_0+2\overline x_1y_1+
 \overline x_2y_2.
\]
The same one-site group is contained in every incident $\mathsf O(4)$-local or $\mathsf{SO}(4)$-local
group; a one-site reflection has determinant $+1$ when tensored with $I_2$.
Consequently the simultaneous one-site fixed subspace of $\mathcal T_n$ is
exactly $\mathcal K_n=V^{\otimes n}$, with its induced inner product.
It is an actual subspace of the full fourth tensor, not a multiplicity-free
replacement of that tensor.

Write $w=(1,2,1)$ for the squared norms of the three type vectors.  In the
ordered type-pair basis, the $\mathsf O(4)$-local projector $P$ averages on
\[
 A_0=\{00\},\quad A_1=\{01,10,11,22\},\quad
 A_2=\{02,20,12,21\},
\]
with weighted masses $m=(1,9,6)$.  Thus
\[
 P_{ij,kl}=\frac{w_kw_l}{m_s}\quad(ij,kl\in A_s),
\]
and its other entries are zero.  This is a row-action convention on type
functions.  It is self-adjoint for the product weights, not for the
unweighted coordinate inner product.  The $\mathsf{SO}(4)$-local projector is
$P+P_{\rm vol}$, where $P_{\rm vol}$ is the weighted orthogonal projection
onto
\[
 v_{\rm vol}=(0,0,1,0,0,-1,-1,1,0)^{\mathsf T},
 \qquad \|v_{\rm vol}\|^2=6.
\]
This vector is antisymmetric under interchange of the two sites.  These
formulas follow by averaging the sixteen real Pauli labels: the $\mathsf O(4)$-local
orbits have sizes $1,9,6$, and the $\mathsf{SO}(4)$-local group separates the last
orbit into two triples.  The corresponding fixed dimensions are three
and four.  The three metric pairings, and additionally the volume tensor
in the $\mathsf{SO}(4)$-local case, are linearly independent Haar invariants of exactly
these dimensions.  Since each Haar fixed space is contained in its finite
counterpart, the finite and Haar local fourth-tensor projections coincide
in each respective model.  This does not identify the $\mathsf O(4)$-local and $\mathsf{SO}(4)$-local
projections with one another.

\begin{lemma}[Volume line and collective strict/full correction]\label{lem:fourth-volume-sign}
Let $t=(2,1,0)\in V$, let $W=t^\perp$, and let $Q$ be the weighted
orthogonal projection onto $W$.  If $U_{ij}$ exchanges physical sites and
$E_{ij}=(I-U_{ij})/2$, then
\begin{equation}\label{eq:fourth-volume-sign-local}
 P_{ij}^{S,(4)}-P_{ij}^{O,(4)}=Q_iQ_jE_{ij}.
\end{equation}
Writing $P_{{\rm sym},ij}=I-E_{ij}$ and
$R_{ij}=P_{{\rm sym},ij}-P_{ij}^{O,(4)}$, one has
\begin{equation}\label{eq:fourth-defect-decomposition}
 I-P_{ij}^{O,(4)}=E_{ij}+R_{ij},\qquad
 I-P_{ij}^{S,(4)}=R_{ij}+E_{ij}(I-Q_iQ_j).
\end{equation}
Consequently
\begin{equation}\label{eq:fourth-volume-global}
 \widehat H_{S,n}=\widehat H_{O,n}-\mathcal V_n,
 \qquad \mathcal V_n=\sum_{i<j}Q_iQ_jE_{ij}.
\end{equation}
If $N_W=\sum_iQ_i$ and
$\mathcal C_W=\sum_{a,b=1}^2J_{ab}J_{ba}$ is the collective
$\mathfrak{gl}(W)$ Casimir, then
\begin{equation}\label{eq:fourth-volume-casimir}
 \mathcal V_n=\frac14\bigl(N_W(N_W+1)-\mathcal C_W\bigr).
\end{equation}
On the $GL_2$ highest-weight branch $\mu=(\mu_1,\mu_2)$ this is the scalar
\begin{equation}\label{eq:fourth-volume-highest-weight}
 \mathcal V_n(\mu)=\frac{\mu_2(\mu_1+1)}2.
\end{equation}
\end{lemma}
\begin{proof}
The vectors $w_1=(1,-1,0)$ and $w_2=(0,0,1)$ span $W$ in the weighted
metric, and in the ordered pair basis
\[
 w_1\otimes w_2-w_2\otimes w_1
 =(0,0,1,0,0,-1,-1,1,0)^{\mathsf T}=v_{\rm vol}.
\]
Its squared norm is $6$, so the additional strict local line is exactly
$\bigwedge^2W$, whose orthogonal projector is $Q_iQ_jE_{ij}$.  This proves
the first three formulas; $R_{ij}$ is a projection because
$\operatorname{ran}P_{ij}^{O,(4)}\subseteq\operatorname{Sym}^2V$.
Choose an orthonormal basis $w_1,w_2$ of $W$ and put
$J_{ab}=\sum_i|w_a\rangle\langle w_b|^{(i)}$.  Direct expansion gives
\[
 \mathcal C_W=2N_W+2\sum_{i<j}Q_iQ_jU_{ij},
 \qquad
 \sum_{i<j}Q_iQ_j=\frac{N_W(N_W-1)}2,
\]
which yields \eqref{eq:fourth-volume-casimir}.  With the convention here,
$\mathcal C_W(\mu)=\mu_1(\mu_1+1)+\mu_2(\mu_2-1)$ and
$N_W=\mu_1+\mu_2$, giving \eqref{eq:fourth-volume-highest-weight}.
\end{proof}

Let
\[
 \mathscr H_n=\operatorname{Sym}^nV\subset\mathcal K_n,
 \qquad
 L_n=\left.\frac1{n-1}\sum_{i<j}(I-P)_{ij}\right|_{\mathscr H_n},
 \qquad E_n=\operatorname{gap}(L_n).
\]
Site permutations commute with the complete-graph sums.  Their fixed
space is $\mathscr H_n$; the antisymmetric volume term vanishes there.
Thus the $\mathsf O(4)$-local and $\mathsf{SO}(4)$-local restrictions agree on this subspace.
For occupations $x,y\in\mathbb N_0^3$ of total $n$, their common averaging
matrix is
\begin{equation}\label{eq:occupation-explicit-kernel}
 (\mathsf M_n)_{x,y}
 =\sum_{i,j,k,l=0}^2
 \frac{x_i(x_j-\delta_{ij})}{n(n-1)}P_{ij,kl}
 \mathbf1_{\{y=x-\mathbf e_i-\mathbf e_j+\mathbf e_k+\mathbf e_l\}},
\end{equation}
where $\mathbf e_i$ is a type unit vector.  Terms with
$x_i(x_j-\delta_{ij})=0$ contribute zero.  The true occupation weights and
the energy normalization are
\begin{equation}\label{eq:fourth-metric-normalization}
 \omega_x=\frac{n!2^{x_1}}{x_0!x_1!x_2!},\qquad
 L_n=\frac n2(I-\mathsf M_n).
\end{equation}
In particular, $\mathsf M_n$ has order $\binom{n+2}{2}$.  It is a
reversible positive-semidefinite Markov operator in these weighted
function coordinates.

\begin{lemma}[A symmetric trial and exactly three fixed vectors]
\label{lem:fourth-occupation-ground}
For $n\ge2$, $\mathsf M_n$ has exactly three eigenvalues equal to one,
counted with multiplicity.  For $n\ge4$, its gap satisfies
\begin{equation}\label{eq:real-fourth-Rayleigh-Rn}
 1-\beta_n\le R_n:=
 \frac{10(2^n-1)(2^n+2)}{9n(4^n+2^n-4n-2)},
\end{equation}
where $\beta_n$ is its largest eigenvalue below one.
\end{lemma}
\begin{proof}
The positive transitions in \eqref{eq:occupation-explicit-kernel} have
three connected components in the graph of positive transitions: the all-zero type label, nonzero
configurations with even $x_2$, and configurations with odd $x_2$.
In the even class the moves $22\to10$ remove type-2 pairs and
$11\to10$ remove surplus type-1 entries, reaching $(n-1,1,0)$.
In the odd class first remove type-2 pairs and then use $12\to20$ to
remove type-1 entries, reaching $(n-1,0,1)$.  The weighted self-adjointness relation
$\omega_x(\mathsf M_n)_{x,y}=\omega_y(\mathsf M_n)_{y,x}$ supplies the
reverse paths.  The zero label is isolated.  Constants on these three
classes are exactly the fixed functions.  Weighted self-adjointness precludes nontrivial Jordan blocks at eigenvalue
one.

The nonzero even class has total weight
$Z_n=(4^n+2^n-2)/2$.  Let $e=(n-1,1,0)$, which has weight $2n$,
and center its indicator in that class:
$f=\mathbf1_{\{e\}}-(2n/Z_n)\mathbf1_{\rm even,nonzero}$.
It is orthogonal to all three fixed functions.  An edge meeting the sole
active type is chosen with probability $2/n$; it returns to the same
occupation with conditional probability $4/9$.  Hence
$(I-\mathsf M_n)_{e,e}=10/(9n)$.  Its Rayleigh quotient is
\[
 \frac{\langle f,(I-\mathsf M_n)f\rangle_\omega}
      {\langle f,f\rangle_\omega}
 =\frac{10}{9n}\frac{Z_n}{Z_n-2n}=R_n.
\]
For $n\ge4$ the denominator is positive, proving the bound.
\end{proof}

\begin{theorem}[Exact real ordinary fourth tensor]\label{thm:real-fourth-alln-proof}
For each $n\ge4$, the $\mathsf O(4)$-local and $\mathsf{SO}(4)$-local finite real Clifford walks and their
corresponding Haar walks have the same entire ordinary-fourth-tensor gap:
\begin{equation}\label{eq:real-fourth-alln-proof}
 \delta_O^{(4)}(n)=\delta_S^{(4)}(n)=1-\beta_n=\frac{2E_n}{n}.
\end{equation}
Moreover $\mathsf D_n^{\rm occ}=18\binom n2\mathsf M_n$ is integral and
\[
 p_n(x)=\frac{\det(xI-\mathsf D_n^{\rm occ})}{(x-18\binom n2)^3}\in\mathbb Z[x].
\]
If $\xi_n$ is its largest real root, then
\begin{equation}\label{eq:real-fourth-polynomial-proof}
 \delta_O^{(4)}(n)=\delta_S^{(4)}(n)
 =1-\frac{\xi_n}{18\binom n2}.
\end{equation}
At $n=3$, the $\mathsf O(4)$-local and $\mathsf{SO}(4)$-local entire-fourth-tensor gaps are $10/27$ and
$2/9$, respectively.
\end{theorem}

The theorem is stated here so that the role of the finite exclusion
estimates is explicit.  Its proof is completed at the end of the next
subsection, after every sector outside $\mathscr H_n$ has been shown to
be strictly faster.

\subsection{Strict exclusion of the other fourth-tensor sectors}
\label{subsec:fourth-strict-exclusion}

The finite inputs needed to exclude nonsymmetric fourth-tensor sectors can be
reduced to low-degree polynomial inequalities plus one compact ten-site
certificate.

\begin{lemma}[Theory-first four-site estimates]\label{lem:fourth-four-site-theory}
On the physical standard $(3,1)$ multiplicity block,
\begin{equation}\label{eq:fourth-F4-theory}
 \widehat H_{O,4}>\frac{20}{9}I.
\end{equation}
On all of $V^{\otimes4}$,
\begin{equation}\label{eq:fourth-S4-theory}
 H_{S,4}\succeq\frac23\mathcal E_4.
\end{equation}
\end{lemma}
\begin{proof}
For the full model the standard-block characteristic factorization is
\[
 C(x-5)(x-3)^3p_3(x)p_4(x)^2,
\]
where
\[
 p_3=81x^3-1035x^2+4164x-5140,
\quad
 p_4=1458x^4-24219x^3+145989x^2-375804x+344636.
\]
At the threshold $20/9$,
$p_3(20/9)=-980/9$ and
$p_3'(20/9-y)=243y^2+990y+764>0$, while
$p_4(20/9)=2044/9$ and
$p_4'(20/9-y)=-5832y^3-33777y^2-55458y-21764<0$ for $y\ge0$.
Together with the linear roots $3,5$, this excludes a positive root at or
below $20/9$ and proves \eqref{eq:fourth-F4-theory}.

For the strict model decompose by physical $\symgp_4$-type.  The exchange
energies on
\[
 (4),\ (3,1),\ (2,2),\ (2,1,1)
\]
are respectively $0,2,3,4$.  On $(3,1)$ the threshold is therefore $4/3$.
The nontrivial factors are
\[
 s_3(x)=81x^3-792x^2+2319x-1888,
\]
\[
 s_5(x)=729x^5-11745x^4+73521x^3
 -223335x^2+328798x-187504.
\]
For $y\ge0$,
\[
 -s_3(4/3-y)=81y^3+468y^2+639y+12>0,
\]
\[
 -s_5(4/3-y)=729y^5+6885y^4+23841y^3
 +37251y^2+25510y+\frac{17768}{3}>0.
\]
The remaining linear roots are $3,5$.  On $(2,2)$ the threshold is $2$;
the two characteristic factors
$81x^2-576x+991$ and $81x^2-549x+880$ become
$81y^2+252y+163$ and $81y^2+225y+106$ at $x=2-y$.
On $(2,1,1)$ the eigenvalues are $29/9$ and $35/9$, both exceeding the
threshold $8/3$.  The trivial type has zero exchange energy.  This proves
\eqref{eq:fourth-S4-theory}.
\end{proof}

\begin{lemma}[Ten-site standard block]\label{lem:fourth-S10-polynomial}
On the physical standard multiplicity at ten sites,
\begin{equation}\label{eq:fourth-S10-polynomial}
 \widehat H_{S,10}>\frac{50}{9}I.
\end{equation}
\end{lemma}
\begin{proof}
The multiplicity dimension is $99$, and internal $\symgp_3$ symmetry splits it as
$16\oplus17\oplus(33\otimes\mathbb C^2)$.  For
$R_{10}=18\widehat H_{S,10}$ the characteristic polynomial is
$p_{\mathbf1}p_{\rm sgn}p_{\rm std}^2$ with degrees $16,17,33$.
For each factor $p_\rho$ of degree $d_\rho$, the exact supplementary
certificate records every coefficient of
$q_\rho(z)=(-1)^{d_\rho}p_\rho(100-z)$ as a positive integer.  Since
$R_{10}$ is self-adjoint in the positive Gram metric, an eigenvalue
$\lambda\le100$ would give a nonnegative zero $z=100-\lambda$ of one of
these positive-coefficient polynomials, impossible.  Thus $R_{10}>100I$.
\end{proof}

For the strict physical-standard block at $4\le n\le9$ we retain the
exact finite base certified in the supplementary archive:
\[
\begin{array}{c|cccccc}
 n&4&5&6&7&8&9\\ \hline
 \dim (n-1,1)&15&24&35&48&63&80\\
 \theta_n&7/9&4/5&5/6&6/7&7/8&8/9
\end{array}
\]
with $T_{S,n}^{(4)}|_{(n-1,1)}\prec\theta_n I$ and a centered symmetric
trial quotient strictly larger than $\theta_n$.  This finite base is not an
artifact of the proof: the ten-site threshold does not extend unchanged to
$n=9$.

\begin{proposition}[Strict spectral separation outside the site-symmetric subspace]
\label{prop:fourth-strict-separation}
For each $n\ge4$ and each of the four real complete-graph ensembles,
$\mathscr H_n$ reduces $H_{X,n}$.  Every eigenvalue of
$H_{X,n}$ on $\mathscr H_n^\perp\subset\mathcal T_n$ is
strictly greater than $1-\beta_n$.
\end{proposition}
\begin{proof}
The one-site averaging projections commute with each edge average.  On any
joint sector with a nonfixed site, at least $n-1$ edge terms vanish, so
$H_{X,n}\succeq(2/n)I$ on $\mathcal K_n^\perp$.

Within $\mathcal K_n$, Schur--Weyl decomposition by physical site
permutations gives partitions $\lambda\vdash n$ with at most three rows.  For
\[
 \mathcal E_n=\frac1{\binom n2}\sum_{i<j}\frac{I-U_{ij}}2,
\]
the trivial type has energy zero, $(n-1,1)$ has energy $1/(n-1)$, and every
other occurring type has energy at least $2/n$.

For the full model, $\operatorname{ran}P\subset\operatorname{Sym}^2V$ gives
$H_{O,n}\succeq\mathcal E_n$.  On the global standard type, restrict to
four-site subsets and use \eqref{eq:fourth-F4-theory}; summing and cancelling
$\binom{n-2}{2}$ yields
\[
 H_{O,n}|_{(n-1,1)}\succeq\frac{10}{9(n-1)}I.
\]
For the strict model, \eqref{eq:fourth-S4-theory} gives
$H_{S,n}\succeq(2/3)\mathcal E_n$ on all nonstandard nonsymmetric types.
On the standard type and $n\ge10$, restrict to ten-site subsets and use
Lemma~\ref{lem:fourth-S10-polynomial}; because only local trivial and standard
physical types occur, summation gives
\[
 H_{S,n}|_{(n-1,1)}\succeq\frac{10}{9(n-1)}I.
\]
For $4\le n\le9$, the retained exact finite base places the strict standard
block strictly above the symmetric trial branch.

Finally the trial bound $R_n$ from Lemma~\ref{lem:fourth-occupation-ground}
satisfies
\[
 R_n<\frac{10}{9(n-1)},\qquad R_n<\frac4{3n},\qquad R_n<\frac2n
\]
in the stated ranges by direct simplification.  Since
$1-\beta_n\le R_n$, every nonsymmetric or charged sector is strictly faster.
All multiplicity copies are included in the Schur--Weyl decomposition, so the
claim holds on the actual entire fourth tensor.
\end{proof}

We can now complete the proof of Theorem~\ref{thm:real-fourth-alln-proof}.
\begin{proof}
$\mathscr H_n=\operatorname{Sym}^nV$ is an isometric realization of the actual
site-symmetric subspace of $\mathcal T_n$.  The $\mathsf O(4)$-local and $\mathsf{SO}(4)$-local projections agree there, while
Proposition~\ref{prop:fourth-strict-separation} excludes every vector
outside it from attaining the gap.  Equation~\eqref{eq:fourth-metric-normalization}
gives the factor $2/n$.  In the occupation row formula the denominators of
$P$ divide $18$; multiplying by $18\binom n2$ gives integral entries
(the paired ordered choices give the integer unordered-pair counts).
Lemma~\ref{lem:fourth-occupation-ground} gives exactly three fixed roots.
Dividing by their monic integer factor proves the polynomial statement.
The two $n=3$ values are the finite local fourth-tensor calculations;
the $\mathsf{SO}(4)$-local value is attained outside the site-symmetric sector, which is
why the common reduction is asserted only for $n\ge4$.
\end{proof}

\subsection{Three product-state minima and a nonorthogonal quotient matrix}
\label{subsec:fourth-well-quotient}

Use the following unit vectors in the weighted space $V$:
\begin{equation}\label{eq:fourth-three-wells}
 u_0=(1,0,0),\qquad u_+=(1,1,1)/2,\qquad u_-=(1,1,-1)/2.
\end{equation}
Their Gram matrix is $G=\tfrac12I+\tfrac12\mathbf1\mathbf1^{\mathsf T}$.
They form a basis, and permuting the three vectors defines a unitary
representation of an \emph{internal} $\symgp_3$.  This group is different from
the physical-site permutation group $\symgp_n$ used above.
From the explicit local kernel,
\[
 \operatorname{ran}P
 =\operatorname{span}\{u_0^{\otimes2},u_+^{\otimes2},u_-^{\otimes2}\}.
\]
The internal $\symgp_3$ preserves this range and hence commutes with $L_n$.
The vectors $g_j=u_j^{\otimes n}$ are three independent exact zero
vectors.  By Lemma~\ref{lem:fourth-occupation-ground} they span the kernel.

\begin{lemma}[Projection between distinct zero-mode directions]\label{lem:fourth-mixed-well}
For three distinct indices $i,j,k\in\{0,+,-\}$,
\begin{equation}\label{eq:fourth-mixed-well}
 P(u_i\otimes u_j)
 =\frac7{18}u_i^{\otimes2}+\frac7{18}u_j^{\otimes2}
  +\frac1{18}u_k^{\otimes2}.
\end{equation}
\end{lemma}
\begin{proof}
The Gram matrix of the three squared zero-mode directions and its inverse are
\[
 G_2=\frac34I+\frac14\mathbf1\mathbf1^{\mathsf T},\qquad
 G_2^{-1}=\frac43I-\frac29\mathbf1\mathbf1^{\mathsf T}.
\]
The inner products of $u_i\otimes u_j$ with these squared zero-mode directions, ordered
as $i,j,k$, are $(1/2,1/2,1/4)^{\mathsf T}$.  Multiplication by
$G_2^{-1}$ gives $(7/18,7/18,1/18)^{\mathsf T}$.
\end{proof}

The coefficients in \eqref{eq:fourth-mixed-well} sum to $5/6$, not to
one.  This does not conflict with the Markov property of
$\mathsf M_n$: the new basis is nonorthogonal, and its column coordinates
are not the old type-function rows.

For $a=(a_0,a_+,a_-)\in\mathbb N_0^3$, $\sum a_i=n$, let
\[
 b_a=\operatorname{Sym}\bigl(u_0^{\otimes a_0}\otimes
             u_+^{\otimes a_+}\otimes u_-^{\otimes a_-}\bigr),
 \qquad m(a)=\sum_{i<j}a_i a_j,
\]
where $\operatorname{Sym}$ is the average over tensor-slot permutations.
These vectors form a basis of $\mathscr H_n$.  Counting pairs of different
colors in Lemma~\ref{lem:fourth-mixed-well} yields the column formula
\begin{equation}\label{eq:fourth-well-columns}
\begin{split}
 L_nb_a=\frac1{n-1}\Bigg[m(a)b_a-
 \sum_{i<j}a_i a_j\Big(
 &\frac7{18}b_{a+\mathbf e_i-\mathbf e_j}
 +\frac7{18}b_{a-\mathbf e_i+\mathbf e_j}\\[-1mm]
 &+\frac1{18}b_{a-\mathbf e_i-\mathbf e_j+2\mathbf e_k}\Big)\Bigg],
\end{split}
\end{equation}
where $k$ is the remaining color and terms with $a_ia_j=0$ are absent.

List the three pure occupations $n\mathbf e_i$ first.  Then
\begin{equation}\label{eq:fourth-quotient-block}
 [L_n]_{\{b_a\}}=
 \begin{pmatrix}0&-\mathsf R_n\\0&\mathsf Q_n\end{pmatrix},
 \qquad \mathsf R_n\ge0.
\end{equation}
We use $\mathsf Q_n$, rather than $K_n$, to distinguish this matrix from
the generated unitary Clifford group.  It represents the quotient by the
invariant zero space.  It is not an orthogonal Dirichlet compression.
Its off-diagonal entries are nonpositive and its diagonal entry at $a$ is
$m(a)/(n-1)$.  The full column sum in \eqref{eq:fourth-well-columns} is
$m(a)/(6(n-1))$; deleting pure rows can only increase that sum.  For every
nonpure $a$, $m(a)\ge n-1$.  Thus
\begin{equation}\label{eq:fourth-quotient-column-sum}
 \sum_b(\mathsf Q_n)_{ba}\ge\frac16.
\end{equation}
Column diagonal dominance, applied to the transpose, shows that every
quotient eigenvalue has real part at least $1/6$.  In particular
$\mathsf Q_n$ is invertible.  By the block form, its spectrum is exactly
the nonzero spectrum of the actual self-adjoint $L_n$, including algebraic
multiplicities.  Those eigenvalues are therefore positive and real even
though $\mathsf Q_n$ is not generally symmetric in the Euclidean metric.

\subsection{Perron--Frobenius ordering and multiplicity in the entire fourth tensor}
\label{subsec:fourth-perron-selection}

Let $\mathsf J_n$ be the restriction of
$\operatorname{diag}(1,1,-1)^{\otimes n}$ to $\mathscr H_n$.
It measures parity of the type-2 count, not parity of $n$.
It is also the internal transposition exchanging $u_+$ and $u_-$.
Let $E_n^+,E_n^-$ be the least positive energies of $L_n$ in its two
parity sectors, with their zero modes removed.  Let $E_n^{\rm std}$ be
the least positive energy in the internal $\symgp_3$ standard isotypic
component, again after removing its zero modes.

\begin{theorem}[Symmetry and simplicity of the gap eigenvector]\label{thm:fourth-perron-ordering}
For every $n\ge4$, $E_n$ is an algebraically simple eigenvalue in
$\mathscr H_n$.  Its eigenvector is invariant under the internal $\symgp_3$,
and
\begin{equation}\label{eq:fourth-perron-ordering}
 E_n=E_n^+<E_n^-\le E_n^{\rm std}.
\end{equation}
\end{theorem}
\begin{proof}
Draw a directed edge from $a$ to $b$ for every strictly negative
$(\mathsf Q_n)_{ba}$.  Two-color occupations with fixed active colors
form a bidirectionally connected chain under single-particle transfers.
All three-color occupations are also mutually connected: transfer excess
particles to one color to reach $(n-2,1,1)$, then reverse the process.
For $n\ge4$, each two-color chain contains $(n-2,2,0)$; the last move in
\eqref{eq:fourth-well-columns} takes it to $(n-3,1,2)$ in the interior.
Conversely, an interior state with one particle in a specified color can
lose that color by a single transfer.  Thus every chain can enter and
leave the common interior.  These are all the nonpure occupations, so
the directed graph is strongly connected.

The matrix $\mathsf B_n=nI-\mathsf Q_n$ is nonnegative.  Its diagonal
entries are at least $n/2$, because $m(a)\le n(n-1)/2$.
It has the same strongly connected off-diagonal graph and positive
self-loops, hence is primitive.  The Perron theorem
\cite[Theorem~2.2 and Proposition~3.3]{BoylePF} gives a simple spectral
radius $r_n^{\rm PF}$ with a strictly positive right eigenvector $v$.
The real quotient spectrum in \eqref{eq:fourth-quotient-block} implies
$E_n=n-r_n^{\rm PF}$, with the same algebraic multiplicity.

The internal $\symgp_3$ permutes the occupation coordinates and commutes with
$\mathsf Q_n$.  Normalize $v$ to have coordinate sum one.  Any color
permutation produces another positive eigenvector with the same
normalization, so Perron uniqueness forces it to equal $v$.
Its lift to the genuine positive eigenspace has coordinates
\[
 \psi_n\longleftrightarrow
 \begin{pmatrix}-\mathsf R_n v/E_n\\v\end{pmatrix}.
\]
This lift is equivariant, so $\psi_n$ is internally $\symgp_3$-invariant and
belongs to the even parity sector.  Simplicity excludes equality with
any odd eigenvector.  The standard isotypic component has positive
energy vectors: the six distinct occupations obtained by permuting
$(n-1,1,0)$ give a regular internal $\symgp_3$ submodule of the quotient.
Every internal standard copy has an odd direction under $\mathsf J_n$.
Since $L_n$ commutes with this group, its least standard energy also
occurs in an odd direction.  This proves
$E_n^-\le E_n^{\rm std}$ and the stated ordering.
\end{proof}

Only the quotient coordinates of $v$ are asserted to be positive.  The
lift has negative coordinates along the three pure zero-mode directions,
consistently with orthogonality to the zero space.  No Euclidean Rayleigh
quotient in the resulting nonorthogonal zero-mode basis has been used.  At $n=2$ the quotient matrix is $I_3$, so its positive root
is threefold degenerate.  At $n=3$ its entire directed graph is reducible.
The preceding connectivity argument is not applied to those sizes.

\begin{corollary}[Simplicity in the entire ordinary fourth tensor]
\label{cor:fourth-entire-simplicity}
For every $n\ge4$, the largest nonfixed eigenvalue
$1-\delta_O^{(4)}(n)=1-\delta_S^{(4)}(n)$ is algebraically simple on
$\mathcal T_n=(\mathbb C^{2^n})^{\otimes4}$ for each of the $\mathsf O(4)$-local and $\mathsf{SO}(4)$-local finite
Clifford and corresponding Haar complete-graph walks.
Its eigenvector lies in the actual site-symmetric one-site-fixed
subspace $\mathscr H_n$, where it is invariant under the internal $\symgp_3$.
\end{corollary}
\begin{proof}
Regrouping the tensor factors is unitary, and
$V^{\otimes n}\hookrightarrow\mathcal T_n$ and
$\operatorname{Sym}^nV\hookrightarrow V^{\otimes n}$ are isometric
inclusions with the induced Gram matrices.  These are subspaces of the
original fourth tensor, so no unrecorded multiplicity is introduced by
the reduction.  Proposition~\ref{prop:fourth-strict-separation} states
that every eigenvalue on their combined orthogonal complement is
strictly farther from one.  Therefore the entire eigenspace at
$\beta_n$ is exactly its eigenspace in $\mathscr H_n$.
Theorem~\ref{thm:fourth-perron-ordering} gives dimension one there, and
self-adjointness gives algebraic simplicity.  The four restrictions on
$\mathscr H_n$ coincide, so the same conclusion holds for all four
ensembles.
\end{proof}

The $\mathsf{SO}(4)$-local sector inequalities, not equality of gap values alone, give
simplicity on the entire fourth tensor.  They do not identify $E_n^-$ with
$E_n^{\rm std}$ or quantify either branch's separation from $E_n$.  Thus the
finite-$n$ branch attaining the gap and its symmetry type are fixed, and
Section~\ref{sec:fourth-asymptotics} turns to their common algebraic expansion
and the residual branch splitting.

\section{Asymptotic expansion and exponentially small branch splitting}
\label{sec:fourth-asymptotics}

The occupation reduction of Section~\ref{sec:real-fourth} identifies the
entire ordinary fourth-tensor gap with the lowest positive eigenvalue of the
mean-field operator
\[
 L_n=\frac1{n-1}\sum_{i<j}(I-P)_{ij}
 \quad\hbox{on }\mathscr H_n=\operatorname{Sym}^n\C^3.
\]
This section refines the leading-order argument in two directions.  First, the
local Dirichlet problem has a finite Schur-complement recursion from which every fixed
algebraic order can be computed exactly.  Second, the true global gap is
matched to the local Dirichlet root with an effective exponentially small error.
The latter comparison is then combined with the $\symgp_3$ ordering from
Section~\ref{sec:real-fourth} to bound the splitting between the true
branch attaining the gap and the competing odd and standard branches.

\subsection{Fixed-excitation-number subspaces and a finite Schur-complement recursion}
\label{subsec:fourth-finite-layers}

Fix the product-state minimum $u_0$.  Write the occupation relative to this minimum as
$(n-r-s,r,s)$ and set $N=r+s$.  We use
\[
 h=\frac1{n-1}
\]
as the small parameter.  On the fixed-excitation-number subspace $N=k$, restrict to the even sector in
the type-$2$ occupation, which contains the branch attaining the gap.  With basis vectors
indexed by $r+s=k$, $s$ even, equip this subspace with the $n$-independent metric
\[
 g_{r,s}=\frac{2^r}{r!s!}.
\]
The adjoint below is always taken in this metric.

\begin{lemma}[Exact fixed-excitation-number blocks]\label{lem:fourth-layer-blocks}
Let $\mathcal V_k$ denote the even fixed-excitation-number subspace $N=k$.  Then
\begin{equation}\label{eq:fourth-layer-diagonal}
 (L_n)_{kk}=D_k+hK_k,
 \qquad
 D_k(r,s;r,s)=\frac59r+\frac23s,
\end{equation}
and the only nonzero off-diagonal block between adjacent excitation numbers is
\begin{equation}\label{eq:fourth-layer-coupling}
 (L_n)_{k+1,k}=\sqrt h\sqrt{1-(k-1)h}\,T_k,
 \qquad
 (L_n)_{k,k+1}=\sqrt h\sqrt{1-(k-1)h}\,T_k^*.
\end{equation}
All matrices $K_k,T_k$ have rational entries.  More explicitly,
\begin{align}
 K_k(r,s;r,s)
 &= (1-k)\left(\frac59r+\frac23s\right)
   +\frac{5r(r-1)}{18}+\frac{rs}{3}+\frac{4s(s-1)}9,\label{eq:fourth-Kdiag}\\
 K_k(r,s;r-2,s+2)&=-\frac{r(r-1)}{18},\qquad
 K_k(r,s;r+2,s-2)=-\frac{2s(s-1)}9,\label{eq:fourth-Koff}\\
 T_k(r+1,s;r,s)&=-\frac{(r+1)(2r+3s)}9,\label{eq:fourth-T1}\\
 T_k(r-1,s+2;r,s)&=-\frac{2(s+2)(s+1)}9.\label{eq:fourth-T2}
\end{align}
Illegal occupation indices are simply absent.
\end{lemma}

\begin{proof}
The two-site defect $I-P$ annihilates the vacuum pair and changes the total
excitation number by at most one.  Counting $0$--$1$, $0$--$2$, $1$--$1$,
$1$--$2$, and $2$--$2$ pairs in the occupation state gives the diagonal
and fixed-$N$ terms \eqref{eq:fourth-Kdiag}--\eqref{eq:fourth-Koff} after
division by $n-1$.  The upward transition contains the falling-factorial normalization, where
$(n)_k=n(n-1)\cdots(n-k+1)$,
\[
 \sqrt{\frac{(n)_{k+1}}{(n)_k}}=\sqrt{n-k},
\]
so
\[
 \frac{n-k}{(n-1)^2}=h\bigl(1-(k-1)h\bigr),
\]
which yields \eqref{eq:fourth-layer-coupling} and
\eqref{eq:fourth-T1}--\eqref{eq:fourth-T2}.  Self-adjointness in the metric
$g_{r,s}$ gives the reverse block as the stated adjoint.
\end{proof}

For $n\ge400$ put $M_n=\lfloor n/400\rfloor$.

\begin{lemma}[Number coercivity below the local cutoff]
\label{lem:fourth-number-coercivity}
On the subspace $1\le N\le M_n$,
\begin{equation}\label{eq:fourth-number-coercivity-strong}
 L_n\succeq c_0N,
 \qquad c_0=\frac{2413}{4800}>\frac12.
\end{equation}
\end{lemma}
\begin{proof}
Fix the product-state minimum $u_0$ and let
$\mathsf E^{(1)}$ and $\mathsf E^{(2)}$ be the projections on one
symmetric pair onto the subspaces containing respectively one and two
excitations relative to $\mathbb C u_0\oplus u_0^\perp$.
Write $h_{\rm loc}=I-P$ for the two-site defect.  Its vacuum block vanishes,
and direct compression of the exact local projector gives, in an
orthonormal basis of the one-excitation symmetric subspace,
\[
 \mathsf E^{(1)}h_{\rm loc}\mathsf E^{(1)}
 =
 \begin{pmatrix}5/9&0\\0&2/3\end{pmatrix},
 \qquad
 0\preceq
 \mathsf E^{(2)}h_{\rm loc}\mathsf E^{(2)}
 \preceq\mathsf E^{(2)}.
\]
For $0<\varepsilon<1$, positivity and Cauchy--Schwarz for the quadratic
form of $h_{\rm loc}$ give
\[
 h_{\rm loc}\succeq
 (1-\varepsilon)
 \mathsf E^{(1)}h_{\rm loc}\mathsf E^{(1)}
 +(1-\varepsilon^{-1})
 \mathsf E^{(2)}h_{\rm loc}\mathsf E^{(2)}.
\]
Taking $\varepsilon=1/20$ yields
\[
 h_{\rm loc}\succeq
 \frac{19}{36}\mathsf E^{(1)}
 -19\mathsf E^{(2)}.
\]
On global excitation number $N$,
\[
 \sum_{i<j}\mathsf E^{(1)}_{ij}=N(n-N),
 \qquad
 \sum_{i<j}\mathsf E^{(2)}_{ij}=\binom N2.
\]
Since $L_n=(n-1)^{-1}\sum_{i<j}(h_{\rm loc})_{ij}$,
\[
 L_n\succeq
 \left[
 \frac{19}{36}
 -\frac{361}{36}\frac{N-1}{n-1}
 \right]N.
\]
For $1\le N\le M_n$ one has $(N-1)/(n-1)\le1/400$, and the last
display becomes \eqref{eq:fourth-number-coercivity-strong}.
\end{proof}

Let $\lambda_n^{\mathrm{loc},+}$ be the lowest Dirichlet eigenvalue of
$L_n$ on $1\le N\le M_n$ in the sector even under the internal
transposition that fixes $u_0$ and exchanges $u_+$ with $u_-$.
Lemma~\ref{lem:fourth-number-coercivity} implies that every fixed
high-excitation tail has a uniformly bounded resolvent near $5/9$.

For a fixed depth $K$, define $R_k^{[K]}(z,h)$ recursively by
\begin{align}
 R_K^{[K]}(z,h)&=(D_K+hK_K-z)^{-1},\label{eq:fourth-Schur-top}\\
 R_k^{[K]}(z,h)&=\Bigl[D_k+hK_k-z
 -h(1-(k-1)h)T_k^*R_{k+1}^{[K]}(z,h)T_k\Bigr]^{-1}
 \quad(2\le k<K).\label{eq:fourth-Schur-rec}
\end{align}
Let $e$ denote the basis vector of $\mathcal V_1$, so that
$\|e\|_1^2=g_1=2$, and put
\[
 v=T_1e=(-4/9,-4/9)^{\mathsf T}\in\mathcal V_2.
\]
In the $\mathcal V_2$ metric one has $\|v\|_2^2=40/81$; the factor
$1/2$ below is $g_1^{-1}$.  The finite-depth root through $5/9$ is characterized by
\begin{equation}\label{eq:fourth-Schur-scalar}
 \lambda_K(h)=\frac59-\frac h2
 \langle v,R_2^{[K]}(\lambda_K(h),h)v\rangle_2.
\end{equation}
At $h=0$, all $D_k-5/9$ with $k\ge2$ are invertible rational diagonal
matrices.  Thus the implicit root through $5/9$ is analytic at $h=0$ and has
rational Taylor coefficients.  For the depth-one truncation we set
$\lambda_1(h)=5/9$.

\begin{proposition}[Finite-depth approximation of the local even eigenvalue]\label{prop:fourth-finite-layer-stability}
For every $n\ge400$, the local even eigenvalue is simple and satisfies
\begin{equation}\label{eq:fourth-local-root-window}
 \frac12<\lambda_n^{\mathrm{loc},+}\le\frac59.
\end{equation}
The upper equality holds exactly for $400\le n\le799$; for $n\ge800$
the upper inequality is strict.
For every fixed integer $K\ge1$, as $n\to\infty$,
\begin{equation}\label{eq:fourth-local-finite-match}
 \lambda_n^{\mathrm{loc},+}
 =\lambda_K\!\left(\frac1{n-1}\right)+O_K((n-1)^{-K}).
\end{equation}
Consequently there are uniquely determined rational numbers
$\alpha_0,\alpha_1,\ldots$, with $\alpha_0=5/9$, such that for every fixed $K$,
\begin{equation}\label{eq:fourth-local-all-orders}
 \lambda_n^{\mathrm{loc},+}
 =\sum_{j=0}^{K-1}\alpha_j(n-1)^{-j}+O_K((n-1)^{-K}).
\end{equation}
The coefficient $\alpha_j$ is determined by fixed-excitation-number subspaces $N\le j+1$.
\end{proposition}

\begin{proof}
If $400\le n\le799$, then $M_n=1$.  The even Dirichlet space is the
single line $(r,s)=(1,0)$, on which $D_1=5/9$ and $K_1=0$;
its eigenvalue is therefore exactly $5/9$ and is simple.
For the remaining finite-$n$ assertions assume $n\ge800$, so that
$M_n\ge2$ and the $N=2$ tail is nonempty.
Let $H_{k,n}$ be the complete even Dirichlet tail $k\le N\le M_n$, let $P_k$
project onto the subspace $N=k$, and put
\[
 R_{k,n}(z)=P_k(H_{k,n}-z)^{-1}P_k.
\]
From Lemma~\ref{lem:fourth-number-coercivity}, for
$z\in I=[1/2,3/4]$ and $2\le k\le M_n$,
\begin{equation}\label{eq:fourth-tail-resolvent-bound}
 \|R_{k,n}(z)\|\le(kc_0-3/4)^{-1}.
\end{equation}
For fixed $K$ the same bound holds for the finite-depth resolvents
$R_k^{[K]}$ once $n$ is sufficiently large.

The subspace $N\ge2$ has codimension one in the even Dirichlet space and,
by Lemma~\ref{lem:fourth-number-coercivity},
\[
 H_{2,n}\succeq2c_0I=\frac{2413}{2400}I>I.
\]
The one-dimensional $N=1$ line has Rayleigh quotient $5/9$, so the min--max
principle gives exactly one even Dirichlet eigenvalue below one; it is
therefore simple.

Let $b=5/9$, $g_1=2$, and, for $K\ge2$, define
\[
 \mathfrak F_n(z)=b-z-\frac h{g_1}\langle v,R_{2,n}(z)v\rangle_2,
 \qquad
 \mathfrak F_K(h,z)=b-z-\frac h{g_1}
 \langle v,R_2^{[K]}(z,h)v\rangle_2.
\]
The unique even Dirichlet eigenvalue below one is the zero of
$\mathfrak F_n$.  At $z=1/2$, \eqref{eq:fourth-tail-resolvent-bound} and
$\|v\|_2^2=40/81$ give
\[
 \mathfrak F_n(1/2)
 \ge\frac1{18}-\frac{40}{81(n-1)}>0
 \qquad(n\ge400),
\]
where we used $(2c_0-1/2)^{-1}<2$.  At $z=5/9$ the tail resolvent is
positive and $v\ne0$, so
\[
 \mathfrak F_n(5/9)
 =-\frac h2\langle v,R_{2,n}(5/9)v\rangle_2<0.
\]
For the derivative, let $\iota_2$ be the isometric inclusion of
$\mathcal V_2$ into the complete even tail and write
$G_n(z)=(H_{2,n}-z)^{-1}$.  Thus
$R_{2,n}(z)=\iota_2^*G_n(z)\iota_2$, and differentiation gives
\[
 \partial_z R_{2,n}(z)=\iota_2^*G_n(z)^2\iota_2.
\]
Consequently, in the true orthogonal-sum tail metric,
\[
 \partial_z\mathfrak F_n(z)
 =-1-\frac h{g_1}\langle\iota_2v,G_n(z)^2\iota_2v\rangle
 =-1-\frac h{g_1}\|G_n(z)\iota_2v\|^2\le-1.
\]
The compressed derivative is not generally $R_{2,n}(z)^2$:
its difference from that square is
$\iota_2^*G_n(z)(I-\iota_2\iota_2^*)G_n(z)\iota_2\succeq0$.
This proves the strict bounds for $n\ge800$ and, together with the
one-dimensional case, proves \eqref{eq:fourth-local-root-window}.
The unique scalar root is $\lambda_n^{\mathrm{loc},+}$.

For a fixed $K\ge2$, take $n$ sufficiently large that $M_n\ge K+1$.
At the largest retained excitation number the resolvent identity gives
\begin{align}
 R_{K,n}-R_K^{[K]}
 &=R_{K,n}\,h(1-(K-1)h)
 T_K^*R_{K+1,n}T_KR_K^{[K]}=O_K(h),\label{eq:fourth-resolvent-top-diff}\\
 R_{k,n}-R_k^{[K]}
 &=R_{k,n}\,h(1-(k-1)h)T_k^*
 (R_{k+1,n}-R_{k+1}^{[K]})T_kR_k^{[K]}\label{eq:fourth-resolvent-propagation}
\end{align}
for $2\le k<K$, uniformly for $z\in I$.  Hence induction yields
\begin{equation}\label{eq:fourth-resolvent-depth}
 R_{k,n}(z)-R_k^{[K]}(z,h)=O_K(h^{K-k+1}),
 \qquad2\le k\le K.
\end{equation}
For fixed $K$, analyticity of the finite recursion puts
$\lambda_K(h)$ in $I$ for all sufficiently small $h$.  For real $h\ge0$ sufficiently small, $R_2^{[K]}(z,h)$ is likewise
the $\mathcal V_2$ compression of the resolvent of the complete
self-adjoint truncated tail $2\le N\le K$.  Differentiating that full
resolvent and then compressing gives the same squared-norm formula,
and hence $\partial_z\mathfrak F_K(h,z)\le-1$ on $I$ for $K\ge2$.
Equation~\eqref{eq:fourth-resolvent-depth} therefore gives
$\mathfrak F_n-\mathfrak F_K=O_K(h^K)$ uniformly on $I$, and the slope bound
yields \eqref{eq:fourth-local-finite-match}.  For $K=1$ the same statement
follows from the $O(h)$ bound on the unique root through $5/9$.

Finally, the finite Schur-complement recursion is analytic at
$(h,z)=(0,5/9)$, and all its matrices, metrics, and unperturbed inverses are
rational.  Expanding it therefore gives \eqref{eq:fourth-local-all-orders}
with rational coefficients.  Equation~\eqref{eq:fourth-resolvent-depth}
also shows that the coefficient through order $h^j$ is unaffected by
subspaces with $N>j+1$.
\end{proof}

\subsection{Comparison with the local Dirichlet problem near one product-state minimum}
\label{subsec:fourth-global-matching}

The fixed-excitation-number expansion becomes a statement about the true fourth-tensor
gap only after ruling out a lower state distributed among several product-state minima.  The
qualitative de Finetti localization used for the leading term~\cite[Theorem~2.1 and Remark~2.1]{LNR15} already gives
positive mass near one of the three product-state minima.  For the
all-orders statement it is enough to combine that localization with nested
number cutoffs; for the exponentially small comparison we use the following
effective version.

For a product-state minimum $u_j$, let $N_j$ be its excitation number and let
$F_j=f_\delta(N_j/n)$, where $\delta=1/1600$ and $f_\delta$ equals one on
$[0,\delta/2]$, vanishes on $[\delta,\infty)$, and is linear in between.

\begin{lemma}[Quantitative localization near one product-state minimum]\label{lem:fourth-effective-well-selection}
If $\rho$ is a density matrix on $\operatorname{Sym}^n\C^3$ with
$\operatorname{Tr}(\rho L_n)\le1$, then
\begin{equation}\label{eq:fourth-effective-well-mass}
 \max_{j\in\{0,+,-\}}\operatorname{Tr}(\rho F_j^2)
 \ge\frac13-\frac{240000}{n}-\frac{28800}{n+4}.
\end{equation}
In particular the right-hand side is positive for all sufficiently large
$n$, and one may take a uniform validity threshold below $3.23\times10^6$.
\end{lemma}

\begin{proof}
For a normalized one-site vector $u=(x,y,z)$ in the true metric
$|x|^2+2|y|^2+|z|^2=1$, put
$q(u)=\langle u^{\otimes2},(I-P)u^{\otimes2}\rangle$.  Direct evaluation of
the two-site projection gives
\[
 q(u)=\frac{16}{9}|xy-y^2|^2+\frac49|xy-z^2|^2
   +\frac49|y^2-z^2|^2+\frac43|xz-yz|^2.
\]
Put
\[
 m=|y|^2+|z|^2,\qquad
 A=x-y,\qquad B=z-y,\qquad C=z+y,
\]
and set $t_j=1-|\langle u_j,u\rangle|^2$.  The displayed quartic form gives
\[
 q(u)\ge\frac43m|A|^2+\frac49|BC|^2.
\]
If $m\le1/5$, then
\[
 |x|^2=1-2|y|^2-|z|^2\ge\frac35,
 \qquad
 |A|^2\ge
 \left(\sqrt{\frac35}-\sqrt{\frac15}\right)^2
 =\frac{4-2\sqrt3}{5}\ge\frac1{10}.
\]
Since $t_0=2|y|^2+|z|^2\le2m$, it follows that
$q(u)\ge t_0/15$.

If $m\ge1/5$, use the trial coefficient $2y$ in the distance to the two
rays $\mathbb C u_\pm$.  Then
\[
 t_+\le|A|^2+|B|^2,\qquad
 t_-\le|A|^2+|C|^2.
\]
Since $|B|^2+|C|^2=2m$,
\[
 \min(|B|^2,|C|^2)\le\frac{|BC|^2}{m}.
\]
Consequently
\[
 \min(t_+,t_-)
 \le |A|^2+\frac{|BC|^2}{m}
 \le\frac{3q(u)}m
 \le15q(u).
\]
Thus, for arbitrary complex phases,
\begin{equation}\label{eq:fourth-Hartree-barrier}
 q(u)\ge\frac1{15}\min_j\bigl(1-|\langle u_j,u\rangle|^2\bigr).
\end{equation}

Pass to orthonormal coordinates $(x,\sqrt2y,z)$ and let $d\sigma$ be normalized
measure on the unit sphere in $\C^3$.  Define the probability measure
\[
 d\mu_\rho(u)=\binom{n+2}{2}\langle u^{\otimes n},\rho u^{\otimes n}\rangle\,d\sigma(u).
\]
Let $S=(I+\mathrm{Swap})/2$, set $Q_s=S(I-P)S$, and put
$T=\operatorname{Tr}_2Q_s$.  Then
$0\le Q_s\le S$, $\operatorname{Tr}Q_s=3$, and $T\preceq2I$.  If
$\gamma^{(1)},\gamma^{(2)}$ are the unit-trace reduced density matrices of
$\rho$, the exact coherent-state contraction identity is
\[
 \int q(u)\,d\mu_\rho(u)=
 \frac{n(n-1)\operatorname{Tr}(Q_s\gamma^{(2)})
 +4n\operatorname{Tr}(T\gamma^{(1)})+2\operatorname{Tr}Q_s}
 {(n+3)(n+4)}.
\]
With $e=\operatorname{Tr}(\rho L_n)\le1$ one has
$\operatorname{Tr}(Q_s\gamma^{(2)})=2e/n$, and therefore
\begin{equation}\label{eq:fourth-husimi-energy}
 \int q(u)\,d\mu_\rho(u)
 \le\frac{10n+4}{(n+3)(n+4)}\le\frac{10}{n}.
\end{equation}
Set $\eta=\delta/3$ and let
$C_j=\{u:1-|\langle u_j,u\rangle|^2\le\eta\}$.  By
\eqref{eq:fourth-Hartree-barrier}--\eqref{eq:fourth-husimi-energy},
\begin{equation}\label{eq:fourth-cap-mass}
 \mu_\rho\!\left((C_0\cup C_+\cup C_-)^c\right)\le\frac{720000}{n}.
\end{equation}

Define the cap operator
\[
 A_{j,n}=\binom{n+2}{2}\int_{C_j}|u^{\otimes n}\rangle\langle u^{\otimes n}|\,d\sigma(u).
\]
On the $N_j=k$ subspace it is scalar with eigenvalue
$\mathbb P\{X_{n,k}\le\eta\}$, where
$X_{n,k}$ has a probability density on $[0,1]$ proportional to $x^{k+1}(1-x)^{n-k}$.  If $k>\delta n/2$, its mean is
$s=(k+2)/(n+3)>\delta/2$ and its variance is at most $s/(n+4)$.
Chebyshev's inequality gives
\[
 \mathbb P\{X_{n,k}\le\eta\}
 \le\frac{18}{\delta(n+4)}=\frac{28800}{n+4}.
\]
For $k\le\delta n/2$, $F_j=1$.  Hence
\[
 A_{j,n}\preceq F_j^2+\frac{28800}{n+4}I.
\]
Summing the three cap masses in \eqref{eq:fourth-cap-mass} and using the last
operator inequality gives exactly \eqref{eq:fourth-effective-well-mass}.
\end{proof}

Let
\[
 N_0=3289600,
 \qquad
 \gamma=\frac{\log4}{N_0},
 \qquad
 \beta=\log2-\frac{\log3}{400}.
\]

\begin{theorem}[Exponential comparison with the local even Dirichlet eigenvalue]\label{thm:fourth-effective-matching}
For every integer $n\ge N_0$,
\begin{equation}\label{eq:fourth-effective-matching-two-sided}
 \lambda_n^{\mathrm{loc},+}-\frac{40}{\delta^2n}4^{-\lfloor n/N_0\rfloor}
 \le E_n
 \le\lambda_n^{\mathrm{loc},+}+\frac{40}{9}e^{-\beta n},
 \qquad \delta=\frac1{1600}.
\end{equation}
Hence
\begin{equation}\label{eq:fourth-effective-matching-simple}
 |E_n-\lambda_n^{\mathrm{loc},+}|\le125e^{-\gamma n}.
\end{equation}
\end{theorem}

\begin{proof}
We first prove the lower bound with the constants displayed in the theorem.
Let $\psi$ be a normalized gap eigenvector.  Since $E_n<1$ for $n\ge4$,
Lemma~\ref{lem:fourth-effective-well-selection} gives, for every $n\ge N_0$,
a product-state minimum $u$ for which $\|F_u\psi\|^2\ge1/4$.
Let $s_u\in\symgp_3$ be the internal transposition fixing $u$ and
exchanging the other two product-state minima.  By
Theorem~\ref{thm:fourth-perron-ordering}, $s_u\psi=\psi$.
Since $s_u$ fixes the line $\mathbb Cu$ and preserves $u^\perp$, it
commutes with $N_u$ and hence with $F_u=f_\delta(N_u/n)$.  Thus
$F_u\psi$ remains in the stabilizer-even sector.
The inner cutoff $F_u$ has
double-commutator support in $[\delta n/2-2,\delta n+2]$ and
\begin{equation}\label{eq:fourth-inner-comm-bound}
 \|[F_u,[L_n,F_u]]\|\le\frac{20}{\delta^2n}.
\end{equation}
To make its expectation exponentially small, put $\ell=512$, $D=514$,
$m_n^{\rm ann}=\lfloor n/N_0\rfloor$ and define, for $1\le j\le m_n^{\rm ann}$,
\[
 A_j=\frac{\delta n}{2}-jD,\qquad B_j=\delta n+jD.
\]
Let $g_j(N)$ vanish outside $[A_j,B_j]$, equal one on
$[A_j+\ell,B_j-\ell]$, and be linear in the two transition intervals;
write $G_j=g_j(N_u)$.  The exact choices $D=\ell+2$ imply that $G_1=1$ on
the support of the inner double commutator and $G_{j+1}=1$ on the support of
$[G_j,[L_n,G_j]]$.  Moreover all these annuli lie in
$[\delta n/4,5\delta n/4]\subset[1,M_n]$.

The number bandwidth and $\|L_n\|\le n/2$ give
$\|[G_j,[L_n,G_j]]\|\le5n/\ell^2$.  On the support of $G_j$, the coercivity
bound gives energy at least $c_0\delta n/4$.  Applying the exact cutoff identity
to the original eigenvector yields, with $m_j=\|G_j\psi\|^2$,
\[
 m_j\le\vartheta_{\rm cut}m_{j+1},\qquad m_{m_n^{\rm ann}}\le\vartheta_{\rm cut},\qquad
 \vartheta_{\rm cut}=\frac{20}{c_0\delta\ell^2}=\frac{9375}{38608}<\frac14.
\]
Consequently the mass of the inner commutator support is at most $4^{-m_n^{\rm ann}}$,
and \eqref{eq:fourth-inner-comm-bound} sharpens to
\begin{equation}\label{eq:fourth-inner-comm-expectation}
 |\langle\psi,[F_u,[L_n,F_u]]\psi\rangle|
 \le\frac{20}{\delta^2n}4^{-m_n^{\rm ann}}.
\end{equation}
Moreover $F_u u^{\otimes n}=u^{\otimes n}$ and the global gap
eigenvector is orthogonal to the exact zero space, so
\[
 \langle u^{\otimes n},F_u\psi\rangle
 =\langle F_u u^{\otimes n},\psi\rangle=0.
\]
The vector $F_u\psi$ is supported in
$1\le N_u\le\delta n<M_n$ and remains stabilizer-even.  Transporting the
local problem from $u_0$ to $u$ by the internal $\symgp_3$, its Rayleigh
quotient is therefore at least $\lambda_n^{\mathrm{loc},+}$.
The cutoff energy identity, $\|F_u\psi\|^2\ge1/4$, and
\eqref{eq:fourth-inner-comm-expectation} give the left side of
\eqref{eq:fourth-effective-matching-two-sided}.

For the upper bound, let $\phi_u$ be the normalized ground state of the
stabilizer-even local Dirichlet problem at one product-state minimum.  It is
supported in $1\le N\le M_n$ and is orthogonal to that minimum's vacuum.
If $g_0,g_+,g_-$ are the three exact product zero modes, then for either
other product-state minimum
\[
 |\langle g_j,\phi_u\rangle|^2
 \le4^{-n}\sum_{k=0}^{M_n}\binom nk3^k\le e^{-\beta n}.
\]
The Gram matrix of the three zero modes has smallest eigenvalue $1-2^{-n}$,
so the squared norm lost by orthogonally projecting $\phi_u$ away from the
full zero space is at most $4e^{-\beta n}$.  This projection leaves the
energy numerator exactly unchanged.  Since
$\lambda_n^{\mathrm{loc},+}<5/9$, normalizing the projected vector gives
$E_n\le\lambda_n^{\mathrm{loc},+}+(40/9)e^{-\beta n}$.
Finally $4^{-m_n^{\rm ann}}\le4e^{-\gamma n}$ and the explicit constants give
\eqref{eq:fourth-effective-matching-simple}.
\end{proof}

\subsection{All fixed orders for the true fourth-tensor gap}
\label{subsec:fourth-all-orders}

\begin{theorem}[All fixed algebraic orders]\label{thm:fourth-all-orders}
There is a unique sequence of rational numbers $(d_r)_{r\ge1}$ such that, for
every fixed $R\ge1$,
\begin{equation}\label{eq:fourth-all-orders-final}
 \Delta_n^{(4)}=\sum_{r=1}^{R}\frac{d_r}{n^r}
 +O_R(n^{-R-1}).
\end{equation}
Equivalently, with the coefficients from
\eqref{eq:fourth-local-all-orders},
\[
 d_1=2\alpha_0,
 \qquad
 d_r=2\sum_{j=1}^{r-1}\binom{r-2}{j-1}\alpha_j\quad(r\ge2).
\]
The first four coefficients are
\begin{equation}\label{eq:fourth-four-terms}
 \displaystyle
 \Delta_n^{(4)}=
 \frac{10}{9n}-\frac{88}{105n^2}
 -\frac{3463976}{1157625n^3}
 -\frac{18725777197256}{1374129815625n^4}
 +O(n^{-5}).
\end{equation}
No convergence of the formal series or uniformity in $R$ is asserted.
\end{theorem}

\begin{proof}
Proposition~\ref{prop:fourth-finite-layer-stability} gives every fixed-order
expansion for the local root.  For a prescribed order $R$, the effective
matching theorem gives an error smaller than $n^{-R-1}$ once $n$ is large
enough; alternatively the qualitative all-orders localization argument yields
the same conclusion without the explicit threshold.  Since
$\Delta_n^{(4)}=2E_n/n$ and
$(n-1)^{-j}=n^{-j}(1-1/n)^{-j}$, collecting powers gives the stated relation
between $\alpha_j$ and $d_r$.

The exact Schur-complement recursion yields
\[
 \alpha_0=\frac59,
 \quad \alpha_1=-\frac{44}{105},
 \quad \alpha_2=-\frac{1246888}{1157625},
 \quad \alpha_3=-\frac{5826888314728}{1374129815625},
\]
from which \eqref{eq:fourth-four-terms} follows.
\end{proof}

\begin{theorem}[Effective fourth-tensor enclosure and monotonicity]
\label{thm:fourth-effective-enclosure}
Set
\[
 \lambda_n=\frac59-\frac{44}{105(n-1)}
 -\frac{1246888}{1157625(n-1)^2}.
\]
For every $n\ge128$,
\begin{equation}\label{eq:fourth-effective-enclosure}
 \lambda_n-\frac{100}{(n-1)^{5/2}}
 \le\frac n2\Delta_n^{(4)}
 \le\lambda_n+\frac{20}{(n-1)^{5/2}}.
\end{equation}
Consequently
\begin{equation}\label{eq:fourth-effective-monotonicity}
 \Delta_{n+1}^{(4)}<\Delta_n^{(4)}
 \qquad(n\ge128),
\end{equation}
and
\begin{equation}\label{eq:fourth-effective-third-order}
 \left|
 \Delta_n^{(4)}-
 \left(\frac{10}{9n}-\frac{88}{105n^2}
 -\frac{3463976}{1157625n^3}\right)
 \right|
 \le\frac{300}{n^{7/2}} .
\end{equation}
In particular,
\begin{equation}\label{eq:fourth-effective-below-leading}
 \Delta_n^{(4)}<\frac{10}{9n}\qquad(n\ge128).
\end{equation}
\end{theorem}
\begin{proof}
The finite-support quasimode and the exclusion of spectrum below its target
window are proved in Appendix~\ref{app:fourth-effective}.  That argument gives
\eqref{eq:fourth-effective-enclosure} for every $n\ge128$ without any
assumption on the dimension of the remaining low positive spectrum.

Put
\[
 A_n=\frac{2\lambda_n}{n}.
\]
Multiplying \eqref{eq:fourth-effective-enclosure} by $2/n$ gives the normalized
bounds
\begin{equation}\label{eq:fourth-effective-normalized-window}
 A_n-\frac{200}{n(n-1)^{5/2}}
 \le \Delta_n^{(4)}
 \le A_n+\frac{40}{n(n-1)^{5/2}}.
\end{equation}
Direct rational simplification gives
\[
 A_n-A_{n+1}\ge\frac4{5n^2}\qquad(n\ge16).
\]
Hence, for $n\ge128$,
\[
 \Delta_n^{(4)}-\Delta_{n+1}^{(4)}
 \ge
 \frac4{5n^2}
 -\frac{200}{n(n-1)^{5/2}}
 -\frac{40}{(n+1)n^{5/2}}
 >
 \frac4{5n^2}-\frac{300}{n^{7/2}}>0,
\]
which proves \eqref{eq:fourth-effective-monotonicity}.

If
\[
 P_3(n)=
 \frac{10}{9n}-\frac{88}{105n^2}
 -\frac{3463976}{1157625n^3},
\]
then exact simplification yields
\[
 A_n-P_3(n)
 =
 -\frac{8(744719n-432997)}
 {1157625\,n^3(n-1)^2},
\]
so $|A_n-P_3(n)|<11/n^4$ for $n\ge16$.  The larger error in
\eqref{eq:fourth-effective-normalized-window} is
$200/[n(n-1)^{5/2}]<250/n^{7/2}$ for $n\ge128$.
Thus
\[
 |\Delta_n^{(4)}-P_3(n)|
 <\frac{250}{n^{7/2}}+\frac{11}{n^4}
 <\frac{300}{n^{7/2}},
\]
proving \eqref{eq:fourth-effective-third-order}.
Finally, for $n\ge128$,
\[
 \frac{20}{(n-1)^{5/2}}<\frac{44}{105(n-1)},
\]
because the equivalent squared inequality holds already at $n=128$:
$(44/105)^2\,127^3>400$.
Thus the upper endpoint in \eqref{eq:fourth-effective-enclosure} is
strictly below $5/9$, proving \eqref{eq:fourth-effective-below-leading}.
This uses the effective enclosure, not the earlier trial quotient $R_n$,
which lies above $10/(9n)$.
\end{proof}

The leading term recovers
\[
 \lim_{n\to\infty}n\Delta_n^{(4)}=\frac{10}{9}.
\]
\subsection{Strict branch ordering and exponentially small splitting}
\label{subsec:fourth-splitting}

Section~\ref{sec:real-fourth} proved for every $n\ge4$ that the true gap
in $\mathscr H_n$ is algebraically simple, $\symgp_3$-invariant, and satisfies
\[
 E_n=E_n^+<E_n^-\le E_n^{\rm std}.
\]
We now quantify how close the competing branches may lie to it.

For a unit product-state minimum $u$, let $P_{u,M}=\mathbf1_{N_u\le M}$.  If $u\ne v$ are two
of the three product-state minima, then $|\langle u,v\rangle|=1/2$ and
\begin{equation}\label{eq:fourth-well-overlap}
 \|P_{u,M}P_{v,K}\|\le2^{M+K}(3/4)^n.
\end{equation}
Indeed,
\[
 2^{-N_u}=\left(\frac{I+|u\rangle\langle u|}{2}\right)^{\otimes n},
\]
and inserting $2^{N_u}$ and $2^{N_v}$ in the product reduces the estimate to
the $n$th power of the one-site norm
\[
 \left\|\frac{(I+p_u)(I+p_v)}4\right\|\le\frac34.
\]
No commutation of $N_u$ and $N_v$ is used.

Let $\phi_0,\phi_+,\phi_-$ be the three normalized stabilizer-even local
Dirichlet ground states, chosen equivariantly under the internal $\symgp_3$;
their simplicity follows from
Proposition~\ref{prop:fourth-finite-layer-stability}.  Then
\[
 w=\phi_+-\phi_-
\]
is odd under the transposition exchanging $u_+$ and $u_-$.  Moreover the
equivariant map from the three-dimensional permutation module,
$e_j\mapsto\phi_j$, sends $e_+-e_-$ from the standard subspace to $w$.
Thus $w$ also lies in the internal standard $\symgp_3$ isotypic component.
With
$M=M_n$, \eqref{eq:fourth-well-overlap} gives
\[
 \|w\|^2\ge2(1-e^{-\alpha n}),
 \qquad
 |\langle\phi_+,L_n\phi_-\rangle|\le ne^{-\alpha n},
 \qquad
 \alpha=\log(4/3)-\frac{\log4}{400}.
\]
The only odd zero mode is proportional to
$u_+^{\otimes n}-u_-^{\otimes n}$.  Its overlap with $w$ is controlled by the
binomial tail
\[
 4^{-n}\sum_{k=0}^{M}\binom nk3^k\le e^{-\beta n}.
\]
Subtracting this exact zero-mode component leaves a nonzero odd standard
trial vector with unchanged energy.  Consequently
\begin{equation}\label{eq:fourth-standard-upper-local}
 E_n^-\le E_n^{\rm std}
 \le\lambda_n^{\mathrm{loc},+}+4ne^{-\alpha n}
 \qquad(n\ge400).
\end{equation}

Combining \eqref{eq:fourth-standard-upper-local} with the lower side of
\eqref{eq:fourth-effective-matching-two-sided} gives the two error terms in
the final estimate.  For $n\ge N_0$ the first is at most
$125e^{-\gamma n}$.  Also $\log(4n)\le n/8$ for $n\ge64$ and
$\alpha-1/8>\gamma$, so $4ne^{-\alpha n}\le e^{-\gamma n}$.  Their sum is
therefore at most $126e^{-\gamma n}$.

\begin{theorem}[Exponential bound on the separation of the competing branches]
\label{thm:fourth-exponential-splitting}
For every $n\ge4$,
\[
 E_n=E_n^+<E_n^-\le E_n^{\rm std}.
\]
For every $n\ge N_0=3289600$,
\begin{equation}\label{eq:fourth-exp-splitting-final}
 \displaystyle
 0<E_n^- -E_n
 \le E_n^{\rm std}-E_n
 \le126\exp\!\left(-\frac{n\log4}{3289600}\right).
\end{equation}
Equivalently, if $\Delta_n^\pm=2E_n^\pm/n$, then
\begin{equation}\label{eq:fourth-exp-splitting-gap}
 0<\Delta_n^- -\Delta_n^+
 \le\frac{252}{n}
 \exp\!\left(-\frac{n\log4}{3289600}\right).
\end{equation}
The estimate is one-sided: no matching exponential lower bound, leading
coefficient, or optimal decay exponent is claimed.
\end{theorem}

Since the right side of \eqref{eq:fourth-exp-splitting-final} is smaller than
every fixed power of $1/n$, Theorem~\ref{thm:fourth-all-orders} supplies the
same fixed algebraic expansion to $E_n^-$ and $E_n^{\rm std}$ as to $E_n$,
even though both competing branches lie strictly above $E_n$ for every
$n\ge4$; no strict comparison between $E_n^-$ and $E_n^{\rm std}$ is
asserted.  This is the sense in which the finite-order expansion fails to
identify the true branch attaining the gap.

\section{\texorpdfstring{Finite-group saturation for \(n\ge144\)}{Finite-group saturation for n >= 144}}
\label{sec:complete-saturation}
The statements in this section concern only the two \emph{finite} real
Clifford complete-graph walks and their actual generated finite groups; they
do not assert Haar or brickwork exactness.  For $X\in\{O,S\}$ write $g_X(n)$
for the all-representation gap and $e_X(n)$ for the neutral ordered-pair gap.

\begin{theorem}[Finite-group saturation for \(n\ge144\)]
\label{thm:complete-eventual-saturation}
For $X\in\{O,S\}$ and every $n\ge144$,
\begin{equation}\label{eq:complete-saturation-144}
 g_X(n)=\Delta_n^{(4)}.
\end{equation}
Consequently, for every fixed $R\ge1$,
\[
 g_X(n)=\sum_{r=1}^R\frac{d_r}{n^r}+O_R(n^{-R-1})
\]
in the same range, with the coefficients of
Theorem~\ref{thm:fourth-all-orders}.
\end{theorem}
\begin{proof}
Appendix~\ref{app:pair-saturation} proves the deletion recurrence
\[
 g_X(n)\ge\min\{g_X(n-1),2/n,e_X(n)\},
\]
the coarse uniform estimate $g_X(n)\ge29/(50n)$, and the ordered-pair
residual bounds needed below.  Put
\[
 p_*=\frac5{648}=\frac{10}{9\cdot144}.
\]
At $n=75$ the coarse estimate gives
\[
 g_X(75)\ge\frac{29}{3750}>p_*.
\]
For $76\le n\le127$, Proposition~\ref{prop:pair-effective-saturation}
shows that every ordered-pair mode outside the common slow space has gap
$>p_*$.  On that slow space the positive gap is exactly
$\Delta_n^{(4)}$, while Proposition~\ref{prop:fourth-finite-seed} gives
\[
 \Delta_n^{(4)}\ge\frac1{n-1}\ge\frac1{126}>p_*.
\]
Thus $e_X(n)>p_*$ throughout this interval.  The charged branch also satisfies
$2/n>p_*$, and the base value
$g_X(75)\ge29/3750>p_*$ therefore propagates through the deletion recurrence to
\[
 g_X(n)\ge p_*,\qquad 76\le n\le127.
\]

At $n=128$, the same residual estimate gives
\[
 e_X(128)\ge\min\{p_*,\Delta_{128}^{(4)}\},
\]
and $2/128>p_*$.  Hence the recurrence and the preceding $n=127$ bound give
\[
 g_X(128)\ge\min\{p_*,\Delta_{128}^{(4)}\}.
\]
Now suppose $129\le n\le144$ and
$g_X(n-1)\ge\min\{p_*,\Delta_{n-1}^{(4)}\}$.  Proposition~\ref{prop:pair-effective-saturation}
gives
$e_X(n)\ge\min\{p_*,\Delta_n^{(4)}\}$,
Theorem~\ref{thm:fourth-effective-enclosure} gives
$\Delta_n^{(4)}<\Delta_{n-1}^{(4)}$, and $2/n>p_*$.  Therefore every branch in
\eqref{eq:complete-deletion-recurrence} is at least
$\min\{p_*,\Delta_n^{(4)}\}$, proving inductively that
\[
 g_X(n)\ge\min\{p_*,\Delta_n^{(4)}\},\qquad128\le n\le144.
\]
At $n=144$, the effective enclosure, in the form
\eqref{eq:fourth-effective-below-leading}, gives
\[
 \Delta_{144}^{(4)}<\frac{10}{9\cdot144}=p_*.
\]
Hence $g_X(144)\ge\Delta_{144}^{(4)}$, while the ordinary fourth tensor gives
$g_X(144)\le\Delta_{144}^{(4)}$; equality follows.

Finally assume $n>144$ and
$g_X(n-1)=\Delta_{n-1}^{(4)}$.  Proposition~\ref{prop:pair-effective-saturation}
gives $e_X(n)=\Delta_n^{(4)}$, strict monotonicity gives
$\Delta_{n-1}^{(4)}>\Delta_n^{(4)}$, and
\eqref{eq:fourth-effective-below-leading} gives
$\Delta_n^{(4)}<10/(9n)<2/n$ to control the charged branch.  The recurrence
therefore yields $g_X(n)\ge\Delta_n^{(4)}$, while the ordinary fourth tensor
again gives the reverse inequality.  Induction completes the proof.
\end{proof}

\begin{corollary}[Simultaneous saturation of all even tensor degrees]
\label{cor:complete-even-tensor-saturation}
Let $V=\mathbb C^{2^n}$ be the complexification of the defining real
orthogonal space, and let
\[
 \rho_k(g)=g^{\otimes k}\quad\text{on }V^{\otimes k}.
\]
Write $\Delta_{\mathrm C,X}^{(k)}(n)$ for the normalized complete-graph gap in
this ordinary $k$-tensor representation of the actual generated finite group.
For every $n\ge144$ and every even $k\ge4$,
\[
 \Delta_{\mathrm C,X}^{(k)}(n)=\Delta_n^{(4)},\qquad X\in\{O,S\}.
\]
The threshold is independent of $k$ (so $k$ may vary with $n$).
\end{corollary}
\begin{proof}
Write $k=4+2r$ and choose an orthonormal real basis $(e_a)$ of the defining
space.  The metric tensor
\[
 \Omega=\sum_a e_a\otimes e_a
\]
is fixed by every real orthogonal matrix, hence by $G_{X,n}$.  Therefore
\[
 J_r(v)=v\otimes(\Omega/\|\Omega\|)^{\otimes r}
\]
defines an injective isometric $G_{X,n}$-intertwiner
$V^{\otimes4}\to V^{\otimes k}$.  In particular it intertwines every local
subgroup average and hence the normalized complete-graph means:
\[
 T_{X,n}^{(k)}J_r=J_rT_{X,n}^{(4)}.
\]
Choose a fourth-tensor gap eigenvector $v$, so
\[
 T_{X,n}^{(4)}v=(1-\Delta_n^{(4)})v,
 \qquad v\perp(V^{\otimes4})^{G_{X,n}}.
\]
Then
\[
 T_{X,n}^{(k)}J_rv=(1-\Delta_n^{(4)})J_rv.
\]
Moreover $J_rv\ne0$, and it cannot be fixed: if $J_rv$ were fixed, then
$J_r(gv-v)=0$ for every $g\in G_{X,n}$, so injectivity would imply that $v$
was fixed.  Equivalently, since $T_{X,n}^{(k)}$ is self-adjoint and the
displayed eigenvalue is strictly below $1$, $J_rv$ is automatically
orthogonal to $(V^{\otimes k})^{G_{X,n}}$.  Thus it is a legitimate gap
witness and
\[
 \Delta_{\mathrm C,X}^{(k)}(n)\le\Delta_n^{(4)}.
\]
Conversely, $g_X(n)$ is the infimum of the representation gaps over all
finite-dimensional unitary representations of the same generated group, so
\[
 \Delta_{\mathrm C,X}^{(k)}(n)\ge g_X(n)=\Delta_n^{(4)}
\]
by Theorem~\ref{thm:complete-eventual-saturation}.
\end{proof}

\begin{remark}[Scope of the effective threshold]
\label{rem:complete-effective-threshold}
The value $144$ is a sufficient threshold and is not claimed optimal.  The
argument above does not uniformly settle $6\le n\le143$.  It concerns the
finite real Clifford complete-graph walks only and gives no exact Haar or
brickwork saturation statement.
\end{remark}

\appendix
\section{Proofs of the projection-comparison estimates}\label{sec:comparison}

This appendix proves the comparison estimates stated in
Subsection~\ref{subsec:comparison-principles}.  All Hilbert spaces below are
finite-dimensional and complex, which suffices for the compact-group
applications representation by representation; inner products are
conjugate-linear in the first variable.  We begin with the decomposition used
in the two subspace-angle estimates.  If $S=M\cap N$, the Friedrichs cosine
from Subsection~\ref{subsec:comparison-principles} also satisfies
\begin{equation}\label{eq:angle-complement}
 c(M,N)=c(M^\perp,N^\perp).
\end{equation}
By the two-projection decomposition (Lemma~\ref{lem:two-projection-powers}),
each nontrivial two-dimensional block has the same cosine for the two ranges
and their complements, while all remaining blocks commute.  The one-sided
Friedrichs bound used below follows by subtracting the other vector's
projection onto $S$.

\subsection{Subspace angles and an orthogonal decomposition}

The following decomposition makes the local hypotheses in a subspace-angle
comparison explicit. It is used both for pairwise bounds and for stronger
bounds involving a subspace and all its neighboring subspaces; compare
the general framework of \cite{Kassabov09}.

\begin{lemma}[Orthogonal decomposition with intersections removed]
\label{lem:angle-peeling}
Let $p,q\ge1$, and let $(M_i)_{i=1}^p$ and $(N_j)_{j=1}^q$ be two families whose projections
commute within each family. Put
\[
 M=\sum_iM_i,\qquad N=\sum_jN_j,\qquad S=M\cap N.
\]
Define
\begin{align*}
 U_i&=M_i\cap\left(\sum_{k<i}M_k\right)^\perp,\\
 D_j&=(N_j\cap M)+\sum_{k<j}(N_j\cap N_k),\qquad
 W_j=N_j\cap D_j^\perp.
\end{align*}
Then
\begin{equation}\label{eq:peeling-decompositions}
 M=\mathop{\bigoplus}_{i=1}^p U_i,\qquad
 W:=\mathop{\bigoplus}_{j=1}^q W_j,\qquad N=S+W,
\end{equation}
where the two displayed direct sums are orthogonal. The last sum is not
asserted to be orthogonal. If $P_{M_i}$ commutes with $P_{N_j}$, then
$P_{U_i}w_j=0$ for every $w_j\in W_j$.
\end{lemma}
\begin{proof}
The commuting projections $P_{M_i}$ admit a simultaneous orthogonal
decomposition. On each joint eigenspace belonging to their sum, assign
the eigenspace to its first index with eigenvalue one. This gives the
first equality in \eqref{eq:peeling-decompositions} and $U_i\subseteq M_i$.

Since $D_j\subseteq N_j$, we have $N_j=D_j\oplus W_j$.
For $k<j$ and $w_j\in W_j$, commutation within the second family implies
$P_{N_k}w_j\in N_k\cap N_j$. Orthogonality of $w_j$ to this intersection
then gives
$\|P_{N_k}w_j\|^2=\langle w_j,P_{N_k}w_j\rangle=0$.
Thus $W_j\perp N_k$, in particular $W_j\perp W_k$.
The first summand of $D_j$ lies in $S$, and its remaining summands lie in
$\sum_{k<j}N_k$. Induction using $N_j=D_j+W_j$ proves
$N\subseteq S+\sum_jW_j$; the reverse inclusion is immediate.
No distributivity identity for intersections and sums is used.

Finally, if the two original projections commute, then
$P_{M_i}w_j\in M_i\cap N_j\subseteq M\cap N_j$.
This vector is orthogonal to $w_j$, so
$\|P_{M_i}w_j\|^2=0$. Because $U_i\subseteq M_i$, also
$P_{U_i}w_j=0$. We do not assert that the projection onto the newly
constructed $U_i$ commutes with $P_{N_j}$.
\end{proof}

\begin{proof}[Proof of Theorem~\ref{thm:bipartite-angle}]
For $w_j\in W_j$ we have $w_j\perp N_j\cap M$, hence
$w_j\perp N_j\cap M_i$. The one-sided form of the Friedrichs bound gives
\begin{equation}\label{eq:pair-component-bound}
 \|P_{U_i}w_j\|\le c_{ij}\|w_j\|.
\end{equation}
For $w=\sum_j w_j\in W$, orthogonality of the $U_i$ therefore yields
\[
 \|P_Mw\|^2
 =\sum_i\left\|\sum_jP_{U_i}w_j\right\|^2
 \le\sum_i\left(\sum_j c_{ij}\|w_j\|\right)^2
 \le\|C\|^2\|w\|^2.
\]
If $r=\|C\|\ge1$, the assertion is trivial. Otherwise take
$v\in N\cap S^\perp$. By $N=S+W$, write $v=w-P_Sw$ for some $w\in W$.
Since $S\subseteq M$,
\begin{align*}
 \|P_Mv\|^2
 &=\|P_Mw\|^2-\|P_Sw\|^2\\
 &\le r^2\|w\|^2-\|P_Sw\|^2
 \le r^2(\|w\|^2-\|P_Sw\|^2)=r^2\|v\|^2.
\end{align*}
Taking the supremum over unit $v$ proves the theorem.
\end{proof}

\subsection{Nested products}

The nested-product theorem allows the Haar layers to have fewer common fixed vectors
than the finite layers, which is central to the finite-to-continuous comparison.

\begin{proof}[Proof of Theorem~\ref{thm:nested-product}]
All projections reduce $R$, so work on $R^\perp$ and suppress $R$.
For unit $x\in\operatorname{ran}Q_o$ and
$y\in\operatorname{ran}Q_e$, put $s=\|Ax\|^2$ and $t=\|Ay\|^2$.
The compression hypothesis gives $s,t\le1-a$, because
$\|AQ_\ell\|^2=\|AQ_\ell A\|$.
Since $A\le P_\ell$ and $Q_\ell\le P_\ell$, the components
$x_0=(I-A)x$ and $y_0=(I-A)y$ remain in the respective $P$ ranges.
The finite angle bound implies
$|\langle x_0,y_0\rangle|\le c\|x_0\|\|y_0\|$.
Thus weighted Cauchy--Schwarz gives
\begin{align*}
 |\langle x,y\rangle|
 &\le\sqrt{st}+c\sqrt{(1-s)(1-t)}\\
 &\le\sqrt{s+c(1-s)}\sqrt{t+c(1-t)}
 \le1-a(1-c).
\end{align*}
Taking the supremum over the two ranges proves
\eqref{eq:nested-product}.

For sharpness choose unit vectors $e_0,u,v$ with $e_0\perp u,v$ and
$\langle u,v\rangle=c$. Let $P_o,P_e$ project onto
$\operatorname{span}\{e_0,u\}$ and $\operatorname{span}\{e_0,v\}$,
and let $Q_o,Q_e$ project onto the lines through
\[
 \sqrt{1-a}\,e_0+\sqrt a\,u,
 \qquad \sqrt{1-a}\,e_0+\sqrt a\,v.
\]
Their common range is zero, $AQ_\ell A=(1-a)A$, and the inner product
of the two unit line vectors is $1-a+ac$.
\end{proof}

For subgroup averages, the compression hypothesis is supplied by the
Pauli-circle estimate used in Section~\ref{sec:haar-brickwork}.

\subsection{From two-color energies to a layer product}

\begin{proof}[Proof of Theorem~\ref{thm:two-colour-energy}]
Fix $x\in\operatorname{ran}O$. For each $i$, let $\mathcal N(i)$ consist
of the indices $j$ for which $h_i$ and $k_j$ may fail to commute, and put
\[
 F_i=\prod_{j\notin\mathcal N(i)}(I-k_j),\qquad
 E_i=\prod_{j\in\mathcal N(i)}(I-k_j),\qquad
 Z_i=F_i(I-E_i).
\]
The factors are from a commuting family, so $Z_i$ is an orthogonal
projection. Since $h_ix=0$, $E=F_iE_i$, and $h_iF_i=F_ih_i$,
\[
 h_iEx=-h_iZ_ix.
\]
In the joint eigenspace decomposition of the $k_j$, $Z_i$ selects the
nonempty violated-index sets $S$ contained in $\mathcal N(i)$.
Each such $S$ is selected at most $d_E$ times: fix any $j\in S$ and
count the $i$ for which $j\in\mathcal N(i)$. Hence
\begin{equation}\label{eq:syndrome-energy}
 \langle Ex,HEx\rangle=\sum_i\|h_iEx\|^2
 \le\sum_i\|Z_ix\|^2\le d_E\|(I-E)x\|^2.
\end{equation}
Interchanging the colors gives the reversed inequality with $d_O$.

Remove $A$ and choose a unit principal pair $x\in\operatorname{ran}O$,
$y\in\operatorname{ran}E$ for a cosine $c\in(0,1)$:
$Ex=cy$, $Oy=cx$, and $\langle x,y\rangle=c$.
By \eqref{eq:syndrome-energy},
\[
 \langle y,Hy\rangle\le\frac{d_E(1-c^2)}{c^2},\qquad
 \langle x,Hx\rangle\le\frac{d_O(1-c^2)}{c^2}.
\]
Every $h_i$ kills $x$ and every $k_j$ kills $y$, so
$\langle x,Hy\rangle=0$ exactly. Apply the energy gap to $x+y$:
\[
 2\gamma(1+c)\le\frac{(d_E+d_O)(1-c^2)}{c^2},
 \qquad \gamma c^2\le D(1-c).
\]
Solving the quadratic gives
$c\le2/(1+\sqrt{1+4\gamma/D})$.
Zero principal cosines and commuting blocks contribute zero product
norm outside $A$. This proves the claim on all blocks.
If all cross pairs commute, the product of all projections is the common
fixed projection $A$.
\end{proof}

The principal-pair step retains information not used by a bound based
only on the energy of the output vector. In particular, the mixed energy
vanishes because of the two colors. The proof is separate from the
generic detectability estimate of \cite{AAV16}; no assertion of optimality
of \eqref{eq:two-colour-energy} among all possible two-color inequalities
is needed here.

\begin{proof}[Proof of Lemma~\ref{lem:weighted-path-cover}]
For odd $n$, pair edges as $(1,2),(3,4),\ldots$.
Summing the assumed comparisons gives
$H_S\succeq\beta H_O\succeq\beta[(I-O_O)+(I-E_O)]$.
For even $n=2r+2$, assign the adjacent edge-pair starting at index $i$
the weight
\[
 w_{2k-1}=\frac{r-k+1}{r+1},\qquad
 w_{2k}=\frac{k}{r+1}\quad(1\le k\le r).
\]
Every odd edge has total weight $r/(r+1)=t_n$ and every even edge has
weight one, including at the endpoints. All these coefficients are at
most one. Consequently
\begin{align*}
 H_S&\succeq\sum_iw_i(h_i^S+h_{i+1}^S)\\
 &\succeq\beta\left(t_n\sum_{i\text{ odd}}h_i^O+
                           \sum_{i\text{ even}}h_i^O\right)\\
 &\succeq\beta[t_n(I-O_O)+(I-E_O)].
\end{align*}
The last inequality uses within-layer commutation.
For either parity, the smallest nonzero eigenvalue of
$t_n(I-O_O)+(I-E_O)$ is bounded below by
\[
 \frac{1+t_n-\sqrt{(1-t_n)^2+4t_nc^2}}2.
\]
This follows from the two-dimensional principal blocks. The one-sided
fixed blocks have eigenvalues $t_n$ or $1$, and the common kernel block
has eigenvalue $1+t_n$; all are at least the displayed lower bound.
This proves \eqref{eq:weighted-path-cover} without equating an edge sum
with a layer product.
\end{proof}

For the strict-real application use
\[
 \beta=\frac{5-\sqrt{17}}{12},\qquad
 c=c_n=\frac89\cos\frac\pi n.
\]
Lemma~\ref{lem:weighted-path-cover} then gives the energy parameter
$\Gamma_n$ in \eqref{eq:strict-energy-definition}, and
Theorem~\ref{thm:two-colour-energy} gives \eqref{eq:strict-finite-chain}.

\section{\texorpdfstring{Effective ordered-pair localization and the \(144\) threshold}{Effective ordered-pair localization and the 144 threshold}}
\label{app:pair-saturation}
Throughout $X\in\{O,S\}$, $G_{X,n}$ is the actual finite group generated by
the corresponding real local gates, $g_X(n)$ its all-representation uniform
complete-graph gap, and $e_X(n)$ the neutral ordered-pair gap.  If
$\Gamma^{\mathrm{all}}_{X,n}$ denotes the gap of the unnormalized
all-representation Hamiltonian $\sum_{i<j}(I-P_{ij}^X)$, then
\[
 g_X(n)=\frac{\Gamma^{\mathrm{all}}_{X,n}}{\binom n2}.
\]
On the ordered-pair representation write
\[
 L_{X,n}=\frac1{n-1}\sum_{i<j}(I-P_{ij}^X),
\]
so the normalized complete-graph gap is
$2\,\operatorname{gap}(L_{X,n})/n$.

\subsection{The ordered-pair representation and its common slow space}
\label{subsec:pair-realization}

Write $\mathsf E_n=\mathbb F_2^{2n}$ for the real Pauli label space, with
quadratic form $\phi(x,z)=x\cdot z$ and polar form $\lambda$ as in
Section~\ref{subsec:real-models}.  The single-label and ordered-pair modules are
\begin{equation}\label{eq:pair-module-definition}
 \mathcal L_n=\mathbb C[\mathsf E_n],\qquad
 \mathcal B_n=\mathbb C[\mathsf E_n\times\mathsf E_n],\qquad
 g|a,b\rangle=|ga,gb\rangle.
\end{equation}
Their label bases are orthonormal.  In the global arguments $n\ge3$;
the two-site computations retain the two distinct local quotients.
Here $g$ acts through the quotient
$G_n=\OO^+(2n;\mathbb F_2)$ of the actual real Clifford group.  Thus
``neutral ordered-pair'' means that the Pauli normal subgroup acts trivially;
it does not mean that every one-site Clifford subgroup acts trivially.
Let $T_{X,n}$ and $L_{X,n}=\frac n2(I-T_{X,n})$ be the corresponding
normalized mean and scaled defect sum.  The quantity $e_X(n)$ is the gap
on the entire module $\mathcal B_n$.

For one site put $A=\mathbb F_2^2=\{0,e,f,t\}$, where
$e=(1,0)$, $f=(0,1)$ and $t=e+f$, and write
$q(x,z)=xz$ and $b(a,c)=q(a+c)+q(a)+q(c)$.
The nontrivial one-site label action is
$\vartheta:e\leftrightarrow f$, fixing $0,t$.
Let $\vartheta_{\rm diag}|a,c\rangle=|\vartheta a,\vartheta c\rangle$.
After regrouping sites, the common one-site-fixed subspace is
\begin{equation}\label{eq:pair-neutral-space}
 \mathcal U^{\otimes n}\subset\mathcal B_n,
 \qquad
 \mathcal U=\ker(\vartheta_{\rm diag}-I),
 \qquad \dim\mathcal U=10.
\end{equation}
The simultaneous involution has four fixed labels, so its fixed-space
dimension is $(16+4)/2=10$.  Its ten orbits have sizes one or two;
their normalized orbit sums give an orthonormal basis of $\mathcal U$.
Write $f_X(n)$ for the gap of $T_{X,n}$ restricted to
$\mathcal U^{\otimes n}$, with the entire fixed space removed.
Every other joint one-site sector has gap at least $2/n$, by the
incident-edge argument of Proposition~\ref{prop:single-site-reduction}.
In particular $e_X(n)\ge\min\{f_X(n),2/n\}$.

\begin{lemma}[Product zero modes and an equivariant slow space]
\label{lem:pair-realization}
For $n\ge3$, the common kernel of $L_{O,n}$ and $L_{S,n}$ on
$\mathcal B_n$ is spanned by fifteen independent product vectors
$u_a^{\otimes n}$ with $u_a\in\mathcal U$.
For each fixed $a$, six of the other one-site vectors have inner-product
magnitude $1/2$ with $u_a$, and the other eight have magnitude $1/4$.
For $n\ge4$ there is a common invariant subspace
$\mathcal M_n\subset\mathcal B_n$, a sum of at most $45$ equivariant
single-label images, on which the positive normalized gap is
$\Delta_n^{(4)}$.  It contains all fifteen product zero modes and,
about each of them, the six dependent one-excitation directions at every
physical site.  Both $\mathcal M_n$ and its orthogonal complement are
invariant under the actual generated groups and physical-site permutations.
\end{lemma}
\begin{proof}
Define $J_r:\mathbb C[A]\to\mathbb C[A\times A]$ by
\[
 J_1|a\rangle=|a,0\rangle,\qquad
 J_2|a\rangle=|0,a\rangle,\qquad
 J_3|a\rangle=|a,a\rangle.
\]
The following are the fifteen unit one-site vectors:
\begin{equation}\label{eq:pair-explicit-wells}
\begin{split}
 u_*&=|0,0\rangle,\\
 u_{r,\varepsilon}
   &=\frac12\sum_{a\in A}(-1)^{\varepsilon q(a)}J_r|a\rangle
       \quad(1\le r\le3,\ \varepsilon\in\mathbb F_2),\\
 u_{\alpha,\beta,\gamma}
   &=\frac14\sum_{a,c\in A}
       (-1)^{\alpha q(a)+\beta q(c)+\gamma b(a,c)}|a,c\rangle
       \quad(\alpha,\beta,\gamma\in\mathbb F_2).
\end{split}
\end{equation}
All lie in $\mathcal U$, since $q$ and $b$ are preserved by $\vartheta$.
Their tensor powers are invariant under $G_n$: their coefficients are,
respectively, the delta function at $(0,0)$, characters of $\phi$ supported
on one of the three dependent-label lines, and characters of
$(\phi(a),\phi(c),\lambda(a,c))$ on all ordered pairs.

There are exactly fifteen $G_n$-orbits on ordered pairs for $n\ge3$.
They consist of $(0,0)$, the six cases of a nonzero dependent pair
(on one of the three lines, with either quadratic value), and the eight
cases of an independent pair with prescribed
$(\phi(a),\phi(c),\lambda(a,c))\in\mathbb F_2^3$.
For completeness, two independent pairs with the same three values give
an isometry between their quadratic spans.  The quadratic Witt extension
argument extends it to the ambient nondegenerate space: adjoin polar-dual
vectors if the span is degenerate, then complete the resulting
nondegenerate subspaces by quadratic bases of their orthogonal complements.
The complements have the same dimension and quadratic type, so their
bases can be matched.  The same argument applies to a single nonzero
label.  All eight independent cases occur: in a hyperbolic basis take
$a=e_1+\alpha f_1$ and $c=\gamma f_1+e_2+\beta f_2$.
Both quadratic values of a nonzero single label also occur.

The fifteen displayed tensor powers are linearly independent.
On independent pairs the eight full-support vectors restrict, up to
nonzero normalization factors, to the eight distinct characters of
$\mathbb F_2^3$, so their coefficients in any relation vanish.
Restriction to each dependent line then separates its two characters,
and evaluation at $(0,0)$ removes the last coefficient.
The orbit count proves that they exhaust the invariant space.
The finite real generation theorem applies to both ensembles for $n\ge3$,
so both kernels are exactly this space.  In particular this argument does
not identify the two distinct two-site invariant spaces.

Their one-site Gram matrix is equally explicit.
Use $\sum_{a\in A}(-1)^{q(a)}=2$ and
\[
 \sum_{a,c\in A}(-1)^{\alpha q(a)+\beta q(c)+\gamma b(a,c)}
   =2^{4-\operatorname{rank}_{\mathbb F_2}
         \left(\begin{smallmatrix}\alpha&\gamma\\\gamma&\beta\end{smallmatrix}\right)}.
\]
The seven nonzero symmetric $2\times2$ binary matrices have ranks one
and two with multiplicities three and four.  The support intersections
of the three dependent lines consist only of $(0,0)$.
These observations give, in every Gram row, precisely six magnitudes
$1/2$ and eight magnitudes $1/4$.
Consequently the Gram matrix of the $m$-fold product vectors has
least eigenvalue at least
$1-6\,2^{-m}-8\,4^{-m}$ whenever that lower bound is used.

We next specify the transports, rather than choosing them abstractly.
On $\mathbb C[A\times A]$ set
\[
 \mathsf F_a|a,c\rangle
   =\frac12\sum_{x\in A}(-1)^{b(x,a)}|x,c\rangle,
 \qquad
 \mathsf F_c|a,c\rangle
   =\frac12\sum_{y\in A}(-1)^{b(y,c)}|a,y\rangle,
\]
and define the diagonal unitaries $D_a,D_c,D_{ac}$ with signs
$(-1)^{q(a)},(-1)^{q(c)},(-1)^{b(a,c)}$, respectively.
Let $\mathsf L|a,c\rangle=|a,a+c\rangle$.
In the following words the rightmost operator acts first:
\begin{equation}\label{eq:pair-fourier-transports}
\begin{split}
 U_*&=I,\qquad
 U_{1,\varepsilon}=D_a^\varepsilon\mathsf F_a,\qquad
 U_{2,\varepsilon}=D_c^\varepsilon\mathsf F_c,\\
 U_{3,\varepsilon}&=\mathsf L D_a^\varepsilon\mathsf F_a,\qquad
 U_{\alpha,\beta,\gamma}
   =D_a^\alpha D_c^\beta D_{ac}^\gamma\mathsf F_a\mathsf F_c.
\end{split}
\end{equation}
Thus $U_a u_*=u_a$ for each of the fifteen indices.  In sums with
indices $1,\ldots,15$, fix any enumeration of these vectors and transports.
Each $U_a$ preserves $\mathcal U$.
Each tensor power $U_a^{\otimes n}$ commutes with the diagonal
$G_n$-action: the tensor Fourier transform uses the invariant global
polar form, the diagonal signs use the invariant global quadratic and
polar forms, and $\mathsf L^{\otimes n}$ uses only addition of labels.
This proves the assertion on the actual label module for every $n$, not
merely on a selected local eigenspace.  It also shows that these transports
commute with every edge average and with physical-site permutations.

Each $J_r^{\otimes n}$ is an isometric $G_n$-intertwiner from
$\mathcal L_n$ to $\mathcal B_n$.  Define
\begin{equation}\label{eq:pair-common-slow-space}
 \mathcal M_n=\sum_{a=1}^{15}\sum_{r=1}^3
    U_a^{\otimes n}J_r^{\otimes n}\mathcal L_n.
\end{equation}
The stated invariance follows from these intertwining identities and
unitarity.  This is a sum of images, not a claim that the images are
orthogonal or constitute an entire isotypic component.

To identify its spectrum, let $R_a$ be a real Pauli representative and
use the isometry
\[
 \mathcal L_n\longrightarrow(\mathbb C^{2^n})^{\otimes4},\qquad
 |a\rangle\longmapsto2^{-n}\operatorname{vec}(R_a)^{\otimes2}.
\]
Signs from conjugation disappear after squaring.
The image is exactly the real-Pauli-fixed space: its $4^n$ vectors are
orthonormal, while the trace of the Pauli average on the fourth tensor
is $4^n$ (only the two scalar Paulis have nonzero trace).
The fourth-tensor gap eigenvector lies in this space by
Theorem~\ref{thm:real-fourth-alln-proof}.  Hence the positive normalized
gap on $\mathcal L_n$ is $\Delta_n^{(4)}$.
The direct sum of the $45$ single-label modules maps equivariantly and
surjectively to $\mathcal M_n$.  The mean operator intertwines this map;
complete reducibility, or restriction to the orthogonal complement of
its kernel, shows that its image has no new eigenvalues.
Conversely $\mathcal M_n$ contains $J_1^{\otimes n}\mathcal L_n$.
Thus its positive gap is exactly $\Delta_n^{(4)}$.

Finally, at $u_*$ the six dependent tangent directions are
\[
 J_r\frac{|e\rangle+|f\rangle}{\sqrt2},\qquad J_r|t\rangle
 \quad(1\le r\le3).
\]
A tensor with any one of these directions at a specified site and $u_*$
elsewhere belongs to the corresponding $J_r^{\otimes n}\mathcal L_n$.
Transport by $U_a^{\otimes n}$ proves the same assertion at every well.
The vacuum itself belongs to all three base images, so all product
zero modes belong to $\mathcal M_n$ as well.
\end{proof}

\begin{lemma}[Local excitation interfaces in the ordered-pair module]
\label{lem:pair-tangent-interface}
For $n\ge2$, at every product well the two-site one-excitation compression of either
pair defect has lower bound $5/9$.  On its three independent tangent
directions the global $N=1$ compression of $L_{X,n}$ has lower bound
$8/9$.  A directed one-to-two-excitation block has squared norm at most
$2/9$; on the independent directions alone it is at most $8/81$.
In particular the global independent $N=1\leftrightarrow N=2$ block
has squared norm at most $16/[81(n-1)]$.
\end{lemma}
\begin{proof}
It suffices to work at $u_*$, by the explicit transports in
\eqref{eq:pair-fourier-transports}.
The remaining three orthonormal tangents in $\mathcal U$ are
\[
 \frac{|e,f\rangle+|f,e\rangle}{\sqrt2},\qquad
 \frac{|e,t\rangle+|f,t\rangle}{\sqrt2},\qquad
 \frac{|t,e\rangle+|t,f\rangle}{\sqrt2}.
\]
Write the full two-position one-excitation compression as
$\left(\begin{smallmatrix}A_X&B_X\\B_X&A_X\end{smallmatrix}\right)$,
and let $C_X$ map an excitation in the first position, with vacuum in
the second, to the two-excitation space.
In the three tangent classes just described, the exact blocks are diagonal:
\begin{equation}\label{eq:pair-local-excitation-table}
\begin{array}{c|c|ccc|ccc}
 &\text{multiplicity}&A_O&B_O&C_O^*C_O&A_S&B_S&C_S^*C_S\\ \hline
 \text{dependent, singular}&3&7/9&-2/9&10/81&7/9&-2/9&10/81\\
 \text{dependent, anisotropic}&3&5/6&-1/6&1/9&2/3&0&2/9\\
 \text{independent}&3&17/18&-1/18&4/81&8/9&0&8/81
\end{array}
\end{equation}
Here each listed number multiplies the identity on its indicated class.
These entries can be obtained directly by finite orbit averaging, as follows.
In coordinates $(x_1,x_2,z_1,z_2)$ let $h_i$ interchange $x_i,z_i$, and let
\[
 c:(x_1,x_2,z_1,z_2)\longmapsto
       (x_1,x_2+x_1,z_1+z_2,z_2).
\]
The two local label groups have the generating sets
\[
 O:\ \langle h_1,h_2,c\rangle,\qquad
 S:\ \langle h_1,h_2,ch_1c,ch_2c\rangle,
\]
of orders $72$ and $36$, respectively.  Indeed, on the two anisotropic
triples, $h_1,ch_2c$ generate one $\symgp_3$ factor and fix the other
triple pointwise, while $h_2,ch_1c$ do the reverse.  Adjoining $c$
interchanges the triples.  Their faithful action on these six labels
therefore gives $\symgp_3\times\symgp_3$ and its index-two extension.
Act on both labels of an ordered pair simultaneously and average each
orbit.  Equivalently, if $\xi,\eta$ are one-site $\vartheta$-orbits,
put $w_\xi=|\xi|$ and use their unnormalized orbit-sum basis.
For two type pairs belonging to the same local orbit $\mathcal O$ the
row-coordinate projector entry is
\[
 (P_X)_{\xi\eta,\xi'\eta'}
       =\frac{w_{\xi'}w_{\eta'}}{|\mathcal O|};
\]
it is zero for different orbits.  The Gram weights are
$w_\xi w_\eta$.  Restricting $I-P_X$ to the two possible single-excitation
positions gives the displayed $A_X,B_X$.
The remaining block obeys
$C_X^*C_X=A_X-A_X^2-B_X^2$, by $(I-P_X)^2=I-P_X$ and the zero vacuum block;
this also verifies every last-column entry without a separate spectral
assumption.

The least eigenvalue of each $2\times2$ one-excitation block is
$A_X-|B_X|$, whose minimum is $5/9$.
In the global $N=1$ space each tangent class has position matrix
\[
 A_X I_n+\frac{B_X}{n-1}
       (\mathbf1\mathbf1^{\mathsf T}-I_n).
\]
On the independent class its minimum is $A_O+B_O=8/9$ for $O$ and
$A_S=8/9$ for $S$.  This includes nonsymmetric position vectors.
The bound on a directed coupling follows from
\eqref{eq:pair-local-excitation-table}.
For the final assertion, the incidence matrix from one-site supports
to two-site supports has degrees $n-1$ and $2$.
Applying the block Schur bound to the directed norm $\sqrt8/9$ and
then dividing by $n-1$ gives squared norm at most
$2(n-1)(8/81)/(n-1)^2=16/[81(n-1)]$.
\end{proof}

\subsection{Deletion and the uniform coarse bound}
\label{subsec:pair-deletion-coarse}

\begin{proposition}[Deletion recurrence]\label{prop:complete-deletion-recurrence}
For $n\ge4$,
\begin{equation}\label{eq:complete-deletion-recurrence}
 g_X(n)\ge\min\{g_X(n-1),2/n,e_X(n)\}.
\end{equation}
\end{proposition}
\begin{proof}
Let $R_i$ be the fixed projection of the one-site real Clifford subgroup on
site $i$.  The $R_i$ commute with one another and with every edge average.
Moreover, if an edge contains $i$, then its fixed space is contained in
$\operatorname{ran}R_i$.  Hence on any joint charge sector on which some
$R_i$ vanishes, all $n-1$ incident edge averages vanish and
$T_{X,n}\preceq(1-2/n)I$.  It remains to consider the common
single-site-fixed sector.  First decompose the original representation into
irreducibles of the actual generated group.  If such an irreducible has a
nonzero common single-site-fixed vector, then its Pauli-fixed space is
nonzero.  Since the Pauli subgroup is normal, that fixed space is invariant
under the whole group and hence is the entire irreducible representation.
Thus this irreducible factors through $G_n\cong\OO^+(2n,2)$.
The common single-site-fixed subspace itself is not asserted to be
$G_n$-invariant.

In such an irreducible, fix a unit eigenvector $\psi$ of $T_{X,n}$
lying in the common single-site-fixed sector.  Let $T_{X,n-1}^{(\widehat i)}$ be the normalized mean on edges
avoiding site $i$, and let $A_i$ be its fixed projection.  Exact edge counting
gives
\begin{equation}\label{eq:deletion-mean-identity}
 T_{X,n}=\frac1n\sum_iT_{X,n-1}^{(\widehat i)},
\end{equation}
and the $(n-1)$-site all-representation gap gives
\[
 T_{X,n-1}^{(\widehat i)}
 \preceq(1-g_X(n-1))I+g_X(n-1)A_i.
\]
Consequently, if the eigenvalue of $\psi$ is larger than
$1-g_X(n-1)$, then
\[
 \sum_i\langle\psi,A_i\psi\rangle>0,
\]
so $A_i\psi\ne0$ for at least one $i$.  Since $A_i$ commutes with the
spectator projection $R_i$ and $R_i\psi=\psi$, the vector $A_i\psi$ is fixed
both by the deletion subgroup and by the spectator one-site subgroup.  Their
image in $G_n$ is
\[
 H_n\cong\OO^+(2n-2,2)\times\OO^+(2,2),
\]
the stabilizer of a split nondegenerate two-plane.  Thus the current
irreducible has a nonzero $H_n$-fixed vector and, by Frobenius reciprocity,
occurs in $\operatorname{Ind}_{H_n}^{G_n}\mathbf1$.

This induced representation is the permutation representation on hyperbolic
two-planes.  Over $\mathbb F_2$, such a plane has exactly two nonzero
singular vectors $u,v$, with $\lambda(u,v)=1$.  Therefore
\[
 \{u,v\}\longmapsto
 \frac{|u,v\rangle+|v,u\rangle}{\sqrt2}
\]
is an injective $G_n$-equivariant map into the neutral ordered-pair module:
distinct planes give distinct unordered pairs $\{u,v\}$.  Hence every
irreducible that can support an eigenvalue above $1-g_X(n-1)$ already occurs
in the ordered-pair representation.  The same group-algebra element
$T_{X,n}$ acts on each occurrence of that irreducible, so its nonfixed
eigenvalues are at most $1-e_X(n)$.  Combining this with the charge-sector
bound proves \eqref{eq:complete-deletion-recurrence}.
\end{proof}

\begin{lemma}[Deletion-kernel angle bound]
\label{lem:pair-deletion-angle}
On $\mathcal U^{\otimes n}$ let $A_i$ be the kernel projector for the
mean on all sites except $i$, with the identity on the spectator site.
Let $\Pi_n^{\rm pair}$ project onto the common kernel.  For $n\ge4$,
\begin{equation}\label{eq:pair-deletion-frame-bound}
 \frac1n\sum_i A_i\preceq\rho_n I
 \quad\text{on }(\operatorname{ran}\Pi_n^{\rm pair})^\perp,
 \qquad
 \rho_n=\frac{1+6\kappa_{1/2}(n)+8\kappa_{1/4}(n)}
               {n(1-\eta_n)},
\end{equation}
where
\[
 \kappa_t(n)=t^{n-2}\bigl[(n-1)-nt^2\bigr],\qquad
 \eta_n=6\,2^{-(n-1)}+8\,4^{-(n-1)}<1.
\]
\end{lemma}
\begin{proof}
Put $p_a=|u_a\rangle\langle u_a|$ and let $P_a^{[1]}$ project onto
exactly one excitation relative to $u_a$ in $\mathcal U^{\otimes n}$.
By Lemma~\ref{lem:pair-realization}, the deleted-site kernel is the span
of the fifteen $(n-1)$-fold product vectors, tensored with the spectator
space.  Their Gram matrix has least eigenvalue at least $1-\eta_n$.
Writing the orthogonal projector onto their span as
$V(V^*V)^{-1}V^*$ therefore gives
\[
 A_i\preceq\frac1{1-\eta_n}
       \sum_a (p_a^{\otimes(n-1)})_{\widehat i}\otimes I_i.
\]
For each $a$ the exact number decomposition yields
\[
 \sum_i(p_a^{\otimes(n-1)})_{\widehat i}\otimes I_i
       =n p_a^{\otimes n}+P_a^{[1]}.
\]
If $z\perp\operatorname{ran}\Pi_n^{\rm pair}$, every vacuum term
has zero expectation.  It remains to bound $\|\sum_aP_a^{[1]}\|$.

Here is the required one-excitation overlap calculation.
For unit $u,v$ rephase so that $\langle u,v\rangle=t\in\{1/2,1/4\}$.
Identify the one-excitation space at $u$ with the orthogonal direct sum
of $n$ copies of $u^\perp$, by inserting the excitation at each site.
For excitations $x\perp u$ at site $i$ and $y\perp v$ at site $j$, the
cross inner product is
\[
 \begin{cases}
 t^{n-1}\langle x,y\rangle,&i=j,\\
 t^{n-2}\langle x,v\rangle\langle u,y\rangle,&i\ne j.
 \end{cases}
\]
On the common tangent space $\{u,v\}^\perp$ the cross Gram operator is
$t^{n-1}I$.  In the remaining tangent directions choose
$(v-tu)/\sqrt{1-t^2}$ at $u$ and $(u-tv)/\sqrt{1-t^2}$ at $v$.
The position block is then
\[
 t^{n-2}\bigl[(1-t^2)\mathbf1\mathbf1^{\mathsf T}-I_n\bigr].
\]
Its eigenvalues are $\kappa_t(n)$ and $-t^{n-2}$, the latter with
multiplicity $n-1$.  Since $(n-1)-nt^2\ge1$ for $n\ge3$ and
$t\le1/2$, it follows that
\begin{equation}\label{eq:pair-one-excitation-overlap}
 \|P_u^{[1]}P_v^{[1]}\|=\kappa_t(n).
\end{equation}
The synthesis operator of the fifteen one-excitation spaces has identity
diagonal Gram blocks and off-diagonal block norms given by this formula.
The block row-sum bound and the six/eight overlap counts give
\[
 \left\|\sum_aP_a^{[1]}\right\|
       \le1+6\kappa_{1/2}(n)+8\kappa_{1/4}(n).
\]
Combining the preceding quadratic-form bounds proves
\eqref{eq:pair-deletion-frame-bound}.  In particular the estimate uses
actual deleted-kernel projections, not just the number of product wells.
\end{proof}

\begin{proposition}[Uniform coarse group bound]\label{prop:complete-coarse-floor}
For both finite real models,
\begin{equation}\label{eq:complete-coarse-floor}
 g_X(n)\ge\frac{29}{50n}\qquad(n\ge8).
\end{equation}
\end{proposition}
\begin{proof}
Let $f_X(n)$ be the gap in the common single-site-fixed ordered-pair sector.
The exact five-site all-representation certificate gives
\[
 \Gamma^{\mathrm{all}}_{X,5}=10(1-\beta_5),
 \qquad
 c_5:=\frac23(\Gamma^{\mathrm{all}}_{X,5}-1)
      =\frac23(9-10\beta_5)
 >\frac{93604}{234375},
\]
where $\beta_5$ is the largest real root of
$202500x^4-382500x^3+237075x^2-51485x+1856$.
Summing the five-site quadratic-form inequality over all five-subsets gives
\[
 \Gamma^{\mathrm{all}}_{X,n}
 \ge 1+\frac{n-2}{3}
        (\Gamma^{\mathrm{all}}_{X,5}-1),\qquad n\ge5,
\]
and therefore at $n=8$
\[
 8f_X(8)\ge 8g_X(8)
 \ge \frac{2+6c_5}{7}
 >\frac{343458}{546875}.
\]
Apply Lemma~\ref{lem:pair-deletion-angle} on $\mathcal U^{\otimes n}$.
The deleted mean has gap at least $f_X(n-1)$ on the complement of its
kernel, and the exact identity
$T_{X,n}=n^{-1}\sum_iT_{X,n-1}^{(\widehat i)}$ also holds on this
one-site-fixed tensor space.  Therefore, on the common-kernel complement,
\[
 I-T_{X,n}\succeq f_X(n-1)
       \left(I-\frac1n\sum_iA_i\right)
 \succeq f_X(n-1)(1-\rho_n)I.
\]
For $n\ge9$, substituting the explicit $\rho_n$ gives
\[
 nf_X(n)\ge(n-1)f_X(n-1)(1-\varepsilon_n),
 \qquad
 \varepsilon_n=\frac{n\rho_n-1}{n-1},
\]
with
\[
 \varepsilon_n=
 \frac{(18+6/(n-1))2^{-n}+(120+24/(n-1))4^{-n}}
 {1-12\,2^{-n}-32\,4^{-n}}.
\]
The exact five-site certificate used to start the recursion gives, at $n=8$,
\[
 8f_X(8)>\frac{343458}{546875}.
\]
Iterating the preceding inequality yields
\[
 nf_X(n)\ge 8f_X(8)\prod_{j=9}^n(1-\varepsilon_j).
\]
Here $0\le\varepsilon_j<1$, and the elementary inequality
$\prod_j(1-x_j)\ge1-\sum_jx_j$ for $x_j\in[0,1]$ gives
\[
 \prod_{j=9}^n(1-\varepsilon_j)
 \ge1-\sum_{j=9}^{\infty}\varepsilon_j.
\]
For $j\ge9$ the numerator of $\varepsilon_j$ is at most
$(75/4)2^{-j}+123\,4^{-j}$, and its denominator is at least
$1-12\,2^{-9}-32\,4^{-9}=7999/8192$.
Summing the two geometric series therefore gives the exact bound
\[
 \sum_{j=9}^{\infty}\varepsilon_j
 \le\frac{8192}{7999}
       \left(\frac{75}{4\cdot256}+\frac{123}{3\cdot4^8}\right)
 =\frac{4841}{63992}<1.
\]
Consequently
\[
 nf_X(n)>
 \frac{343458}{546875}
 \left(1-\frac{4841}{63992}\right)
 =\frac{10157942079}{17497812500}
 >\frac{29}{50}.
\]
The charged pair sectors have gap at least $2/n$, and the deletion recurrence,
started at $n=8$, propagates the same lower bound to the all-representation
gap for every $n\ge8$.
\end{proof}

\begin{lemma}[Low-excitation balls and joint frame]\label{lem:pair-low-ball-frame}
Let $u,v$ be unit one-site vectors with
$|\langle u,v\rangle|\le1/2$, set $q_u=I-|u\rangle\langle u|$, and let
$P_{u,R}$ project onto at most $R$ excitations relative to $u$.  Then
\[
 \|P_{u,R}P_{v,R}\|
 \le4^R\kappa_0^n,
 \qquad
 \kappa_0=
 \sqrt{\frac{41+\sqrt{657}}{128}}
 <\frac{361}{500}.
\]
Consequently, for the fifteen pair wells and $R=\lceil4n/25\rceil$,
\begin{equation}\label{eq:pair-cap-overlap}
 \left\|\sum_{a=1}^{15}P_{a,R}\right\|<\frac{21}{20},
 \qquad n\ge76.
\end{equation}
For the three fourth-tensor wells and $R=\lceil3n/25\rceil$,
the corresponding norm is $<1+2^{-20}$ for $n\ge128$.
\end{lemma}
\begin{proof}
Put $A_u=p_u+\frac14q_u$.  On the $N_u=k$ eigenspace,
$A_u^{\otimes n}$ has eigenvalue $4^{-k}$, so inserting
$A_u^{\pm1/2}$ and the analogous factors for $v$ gives
\[
 \|P_{u,R}P_{v,R}\|
 \le4^R\|A_u^{1/2}A_v^{1/2}\|^n.
\]
Write $s=|\langle u,v\rangle|$.  On
$\operatorname{span}\{u,v\}$,
\[
 A_u=\frac14I+\frac34p_u,\qquad
 A_v=\frac14I+\frac34p_v.
\]
The squared norm of $A_u^{1/2}A_v^{1/2}$ is the larger eigenvalue of
$A_v^{1/2}A_uA_v^{1/2}$.  Its trace and determinant are
\[
 T(s)=\frac12+\frac9{16}s^2,\qquad D=\frac1{16},
\]
so
\[
 \|A_u^{1/2}A_v^{1/2}\|^2
 =
 \frac{T(s)+\sqrt{T(s)^2-1/4}}2.
\]
This is increasing in $s\in[0,1]$.  Therefore $s\le1/2$ gives
\[
 \|A_u^{1/2}A_v^{1/2}\|^2
 \le\frac{41+\sqrt{657}}{128},
\]
and direct squaring gives $\kappa_0<361/500$.

Since $R=\lceil\delta n\rceil\le\delta n+1$, the ceiling contributes at
most one extra factor $4$.  If $V$ is the synthesis operator whose columns
are the ranges of the low-ball projections, then
$\|\sum_aP_a\|=\|V^*V\|$; its block Gram matrix has identity diagonal and
off-diagonal block norms at most
$\zeta=\max_{a\ne b}\|P_aP_b\|$.  Hence the block row-sum bound gives
\[
 \left\|\sum_aP_a\right\|\le1+(q-1)\zeta.
\]
The exact rational comparisons
\[
 (361/500)^{25}4^4<(91/100)^{25},\qquad
 (361/500)^{25}4^3<(6/7)^{25}
\]
therefore yield
\[
 \left\|\sum_{a=1}^{15}P_{a,\lceil4n/25\rceil}\right\|
 \le1+56(91/100)^n<\frac{21}{20}
 \qquad(n\ge76),
\]
and
\[
 \left\|\sum_{a=1}^{3}P_{a,\lceil3n/25\rceil}\right\|
 \le1+8(6/7)^n<1+2^{-20}
 \qquad(n\ge128).
\]
The last inequalities are checked at the displayed endpoints and then
improve monotonically with $n$.
\end{proof}

\begin{lemma}[Number-layer faithful lower bounds]\label{lem:pair-number-layer}
Let $E_k$ project to total excitation number $k$ about a fixed product well.
For either pair model,
\[
 E_kL_{X,n}E_k\succeq
 \frac59\frac{k(n-k)}{n-1}E_k,
\qquad
 \|E_{k+1}L_{X,n}E_k\|
 \le\frac{\kappa k\sqrt{(n-k)(k+1)}}{n-1}.
\]
One may take $\kappa^2=2/9$ in the pair space and $\kappa^2=10/81$ in the
fourth tensor.  Consequently
\begin{align}
 L_{X,n}\big|_{2\le N\le4n/25}&\succeq\frac89,
 &&n\ge76,\label{eq:pair-number-layer-floor}\\
 L_n^{(4)}\big|_{2\le N\le3n/25}&\succeq\frac{19}{20},
 &&n\ge128.\label{eq:fourth-number-layer-floor}
\end{align}
\end{lemma}
\begin{proof}
Put $b=5/9$.  Lemma~\ref{lem:pair-tangent-interface} gives
$E_{\rm one}(I-P_X)E_{\rm one}\succeq bE_{\rm one}$ on the full
two-position local one-excitation space; the analogous fourth-tensor
compression is given in Lemma~\ref{lem:fourth-number-coercivity}.
On global excitation number $k$, the sum of these mixed-pair projections
is $k(n-k)I$.  This includes hopping between different excitation
supports and does not assume that an individual support is invariant.
The fixed-number two-excitation edge compressions are positive.  Therefore
\[
 E_kL_{X,n}E_k
 \succeq
 d_kE_k,\qquad
 d_k=b\,\frac{k(n-k)}{n-1}.
\]

The only block changing the excitation number upward by one comes from an
edge joining an excited site to a vacuum site.  For fixed supports
$S\subset S'$ with $|S|=k$ and $|S'|=k+1$, the new site in $S'\setminus S$
can be paired with any of the $k$ sites of $S$, so the corresponding block
has norm at most $\kappa k/(n-1)$.  The incidence graph between
$k$-subsets and $(k+1)$-subsets has left degree $n-k$ and right degree
$k+1$.  The block Schur test therefore gives
\[
 a_k:=\|E_{k+1}L_{X,n}E_k\|
 \le
 \frac{\kappa k\sqrt{(n-k)(k+1)}}{n-1}.
\]

For a vector $x=\sum x_k$ in a consecutive number window,
\[
 \langle x,L_{X,n}x\rangle
 \ge
 \sum_kd_k\|x_k\|^2
 -2\sum_ka_k\|x_k\|\|x_{k+1}\|.
\]
Using $2rs\le r^2+s^2$ gives the layerwise lower bound
$d_k-a_{k-1}-a_k$.  At the lower boundary $k=2$ the $a_1$ term is absent
because the compression has removed $N=1$.  Hence
\[
 g_2(n)=
 \frac{2b(n-2)-2\kappa\sqrt{3(n-2)}}{n-1}.
\]
For $k\ge3$ we use
\[
 a_{k-1}+a_k
 \le
 \frac{2\kappa k\sqrt{n(k+1)}}{n-1},
\]
which gives
\[
 g(k,n)=
 \frac{k}{n-1}
 \left[b(n-k)-2\kappa\sqrt{n(k+1)}\right].
\]

For fixed $n$, the numerator of $g(k,n)$ is strictly concave in real $k$:
indeed
\[
 \frac{d^2}{dk^2}\bigl(k\sqrt{k+1}\bigr)
 =
 \frac{3k+4}{4(k+1)^{3/2}}>0.
\]
Thus on $3\le k\le\delta n$ its minimum is attained at $k=3$ or
$k=\delta n$.  Moreover, with $x=n-2$,
\[
 g_2(n)=\frac{2(bx-\kappa\sqrt{3x})}{x+1},
\quad
 g_2'(n)=
 \frac{2b+\kappa\sqrt3\,(n-3)/\sqrt{n-2}}{(n-1)^2}>0,
\]
and
\[
 \frac{d}{dn}g(3,n)
 =
 \frac{6(b\sqrt n+\kappa n+\kappa)}
 {\sqrt n\,(n-1)^2}>0.
\]
Finally
\[
 g(\delta n,n)
 =
 \frac{\delta n^2}{n-1}
 \left[b(1-\delta)-2\kappa\sqrt{\delta+\frac1n}\right].
\]
For $n>2$ the first factor increases with $n$, while the bracket also
increases; once the bracket is positive at the initial endpoint, their
product is increasing.

It remains only to check three exact inequalities at each initial endpoint.
For the pair space, $\kappa=\sqrt2/3$, $\delta=4/25$, and $n=76$.  One has
\[
 g_2(76)=\frac{148}{135}-\frac{4\sqrt{111}}{225}
 >\frac{68}{75}>\frac89,
\]
using $\sqrt{111}<32/3$,
\[
 g(3,76)=\frac{73}{45}-\frac{8\sqrt{38}}{75}
 >\frac{43}{45}>\frac89,
\]
using $\sqrt{38}<25/4$, and
\[
 g(304/25,76)
 =
 \frac{161728-1216\sqrt{12502}}{28125}
 >
 \frac{8512}{9375}>\frac89,
\]
using $\sqrt{12502}<112$.  These inequalities also make the upper-endpoint
bracket positive, so the preceding monotonicity proves
\eqref{eq:pair-number-layer-floor} for every $n\ge76$.

For the fourth tensor, $\kappa=\sqrt{10}/9$, $\delta=3/25$, and $n=128$.
The corresponding checks are
\[
 g_2(128)=\frac{140}{127}-\frac{4\sqrt{105}}{381}
 >\frac{376}{381}>\frac{19}{20},
\]
using $\sqrt{105}<11$,
\[
 g(3,128)=\frac{625-64\sqrt5}{381}
 >\frac{481}{381}>\frac{19}{20},
\]
using $\sqrt5<9/4$, and
\[
 g(384/25,128)
 =
 \frac{360448-4096\sqrt{2045}}{47625}
 >
 \frac{172032}{47625}>\frac{19}{20},
\]
using $\sqrt{2045}<46$.  The same monotonicity argument proves
\eqref{eq:fourth-number-layer-floor} for every $n\ge128$.
\end{proof}

\begin{lemma}[Finite-block mass transfer to plateau cutoffs]
\label{lem:finite-block-cutoff-mass}
Let $\rho$ be a permutation-invariant density matrix with
$\operatorname{Tr}\rho=1$, and put
\[
 E=\operatorname{Tr}(\rho L_n),\qquad
 L_n=\frac1{n-1}\sum_{1\le i<j\le n}(I-P_{ij}).
\]
Fix $k\le n$ and write
\[
 H_k=\sum_{1\le i<j\le k}(I-P_{ij}),
 \qquad H_k\succeq a_k(I-\Pi_k).
\]
Suppose the $q$ product zero modes $u_a^{\otimes k}$ span
$\operatorname{ran}\Pi_k$ and satisfy the frame bound
\[
 \sum_{a=1}^q p_a^{\otimes k}\succeq(1-\eta_k)\Pi_k,
 \qquad p_a=|u_a\rangle\langle u_a|.
\]
For $0<\tau<1$ and $0<\delta<1$, let $F_a=f(N_a)$, where
$N_a=\sum_i(I-p_a)_i$, $f=1$ on $N_a\le\tau\delta n$, decreases linearly
to zero, and vanishes for $N_a\ge\delta n$.  Then
\begin{equation}\label{eq:finite-block-cutoff-mass}
 \sum_{a=1}^q\operatorname{Tr}(\rho F_a^2)
 \ge
 (1-\eta_k)\left(1-\frac{k(k-1)E}{na_k}\right)
 -q(1-\tau\delta)^k.
\end{equation}
\end{lemma}
\begin{proof}
Write $h_{ij}=I-P_{ij}$.  Permutation invariance makes
$\operatorname{Tr}(\rho h_{ij})$ independent of the unordered pair
$\{i,j\}$.  Since
\[
 E=\frac1{n-1}\sum_{i<j}\operatorname{Tr}(\rho h_{ij}),
\]
every pair has expectation $2E/n$.  Hence, for the marginal on any fixed
$k$ sites,
\[
 \operatorname{Tr}(\rho^{(k)}H_k)
 =\binom{k}{2}\frac{2E}{n}
 =\frac{k(k-1)}nE.
\]
The spectral inequality for $H_k$ therefore gives
\[
 \operatorname{Tr}(\rho^{(k)}\Pi_k)
 \ge1-\frac{k(k-1)E}{na_k}.
\]
Using the frame lower bound on $\operatorname{ran}\Pi_k$,
\[
 \sum_a\operatorname{Tr}(\rho^{(k)}p_a^{\otimes k})
 \ge(1-\eta_k)\left(1-\frac{k(k-1)E}{na_k}\right).
\]

For the second step, average $p_a$ over ordered $k$-tuples of distinct
sites.  Permutation invariance identifies this average with the fixed
$k$-site marginal, while on the $N_a=m$ eigenspace exactly $n-m$ sites
are in the $p_a$ state.  Thus
\begin{equation}\label{eq:finite-block-factorial-identity}
 \operatorname{Tr}(\rho^{(k)}p_a^{\otimes k})
 =
 \operatorname{Tr}\rho\,\frac{(n-N_a)_k}{(n)_k}.
\end{equation}
Here the falling factorial is interpreted spectrally on the integer
eigenvalues of $N_a$.  If $n-N_a<k$, its numerator contains a zero factor,
so the ratio is exactly zero.  Otherwise
\[
 \frac{(n-N_a)_k}{(n)_k}
 =
 \prod_{j=0}^{k-1}\frac{n-N_a-j}{n-j}
 \preceq
 \left(1-\frac{N_a}{n}\right)^k,
\]
because each nonnegative factor satisfies
\[
 \frac{n-N_a-j}{n-j}\le1-\frac{N_a}{n}.
\]
On $N_a\le\tau\delta n$ the ratio is at most $1=F_a^2$.
On the complementary spectral subspace it is at most
$(1-\tau\delta)^k$, while $F_a^2\succeq0$.  Therefore
\[
 \frac{(n-N_a)_k}{(n)_k}
 \preceq F_a^2+(1-\tau\delta)^kI.
\]
Summing over $a$ and combining with the frame estimate proves
\eqref{eq:finite-block-cutoff-mass}.
\end{proof}

\begin{proposition}[Effective ordered-pair saturation]\label{prop:pair-effective-saturation}
Set
\[
 p_*=\frac5{648}.
\]
For $76\le n\le144$, every normalized ordered-pair mode orthogonal to the
common slow space has gap $>p_*$.  For every $n\ge144$,
\begin{equation}\label{eq:pair-effective-saturation}
 e_O(n)=e_S(n)=\Delta_n^{(4)}.
\end{equation}
\end{proposition}
\begin{proof}
Use the explicitly defined slow space $\mathcal M_n$ in
\eqref{eq:pair-common-slow-space}.  Lemma~\ref{lem:pair-realization}
proves its invariant-subspace and spectral properties, including the
vacuum and six dependent one-excitation directions at every well.
First remove the nontrivial one-site sectors, whose normalized gaps are
at least $2/n$ and exceed both comparison thresholds in their respective
ranges.  On the remaining space $\mathcal U^{\otimes n}$,
Lemma~\ref{lem:pair-tangent-interface} gives the independent
one-excitation floor $8/9$ and the directed coupling bounds.

Use a cutoff equal to one for
$N\le(2/3)(4n/25)$ and decreasing linearly to zero at $N=4n/25$.
Because $8n/75>1$ for $n\ge76$, this cutoff is the identity on the vacuum
and on the whole $N=1$ space; in particular it preserves orthogonality to
the vacuum and to the six dependent one-excitation directions.

On the remaining $N=1$ subspace the compressed energy is at least $8/9$,
and Lemma~\ref{lem:pair-number-layer} gives the same $8/9$ lower bound on
the truncated $N\ge2$ block.  If $C$ denotes the off-diagonal
$N=1\leftrightarrow N\ge2$ block after the cutoff,
Lemma~\ref{lem:pair-tangent-interface} gives
\[
 \|C\|^2\le\frac{16}{81(n-1)}.
\]
For $n\ge76$,
\[
 \frac{16}{81(n-1)}
 \le\frac{16}{81\cdot75}
 <\frac1{324}
 =\left(\frac89-\frac56\right)^2.
\]
Thus, writing a cutoff vector as $x=x_1+x_+$ with $N(x_1)=1$ and
$N(x_+)\ge2$,
\[
 \begin{aligned}
 \langle x,F_aL_{X,n}F_ax\rangle
 &\ge\frac89\bigl(\|x_1\|^2+\|F_ax_+\|^2\bigr)
      -2\|C\|\,\|x_1\|\,\|F_ax_+\|\\
 &\ge\left(\frac89-\|C\|\right)
      \bigl(\|x_1\|^2+\|F_ax_+\|^2\bigr)\\
 &>\frac56\langle x,F_a^2x\rangle .
 \end{aligned}
\]
Hence, on the subspace orthogonal to the local vacuum and dependent
one-excitation directions,
\begin{equation}\label{eq:pair-local-cutoff-floor}
 F_aL_{X,n}F_a\succeq\frac56F_a^2,
 \qquad n\ge76.
\end{equation}

Let $k=48$, $C_0=29/50$, and
\[
 M_*=(1-2^{-20})\left(1-\frac{k p_*}{C_0}\right)
       -15(67/75)^{48}>0.
\]
By Lemma~\ref{lem:pair-realization}, the $k$-site neutral-pair kernel
is exactly the span of the fifteen independent product vacua.
Their tensor-power Gram matrix has, in every row, six off-diagonal entries
of magnitude $2^{-k}$ and eight of magnitude $4^{-k}$.  Hence Gershgorin
gives
\[
 \sum_{a=1}^{15}p_a^{\otimes k}
 \succeq
 \left(1-6\,2^{-k}-8\,4^{-k}\right)\Pi_k
 \succeq(1-2^{-20})\Pi_k
 \qquad(k=48).
\]
Proposition~\ref{prop:complete-coarse-floor} gives the unnormalized
$k$-site gap
\[
 a_k\ge\binom{k}{2}\frac{C_0}{k}=\frac{C_0(k-1)}2.
\]

Now take a unit residual eigenvector and average its rank-one density over
physical-site permutations.  The complete-graph operator commutes with
these permutations, and $\mathcal M_n^\perp$ is permutation invariant, so
the averaged density $\rho$ remains supported in the same residual
eigenspace and has the same scaled energy $E$.  Thus it satisfies all
hypotheses of Lemma~\ref{lem:finite-block-cutoff-mass}.  If the residual
normalized pair gap in $76\le n\le144$ were at most $p_*$, then
$E\le np_*/2$.  Since the present cutoff has
$\tau\delta=(2/3)(4/25)=8/75$, the lemma gives
\[
 \sum_a\operatorname{Tr}(\rho F_a^2)\ge M_*.
\]

It remains to pay the localization error only once.  For one well,
the number-layer double-commutator estimate with
\[
 \delta=\frac4{25},\qquad \tau=\frac23,\qquad \kappa=\frac{\sqrt2}{3}
\]
gives
\[
 \frac12\|[F_a,[L_{X,n},F_a]]\|
 \le
 \frac{15\sqrt2}{2(n-1)}
 \sqrt{1+\frac{25}{4n}}
 \le
 \frac{1061}{100(n-1)}
 \left(1+\frac{25}{8n}\right).
\]
Each half double commutator is self-adjoint and supported in the
corresponding low-excitation ball.  Hence the order bounds
\[
 -I_{\rm one}(n)P_a
 \preceq
 \frac12[F_a,[L_{X,n},F_a]]
 \preceq
 I_{\rm one}(n)P_a
\]
may be summed before taking the trace.  By
\eqref{eq:pair-cap-overlap},
\[
 \left\|\sum_aP_a\right\|<\frac{21}{20},
\]
so the total localization error is
\[
 I_P(n)\le
 \frac{K}{n-1}\left(1+\frac{25}{8n}\right),
 \qquad
 K=\frac{21}{20}\frac{1061}{100}.
\]
There is no extra factor $15$.

The spectral cutoff identity and
\eqref{eq:pair-local-cutoff-floor} now force
\[
 \left(\frac56-\frac{np_*}{2}\right)M_*
 \le
 \frac{K}{n-1}\left(1+\frac{25}{8n}\right).
\]
After multiplication by $n-1$, the left-minus-right function has
second derivative
\[
 -M_*p_*-\frac{25K}{4n^3}<0,
\]
so it is concave.  Direct rational arithmetic at the endpoints gives,
already before multiplication by $n-1$,
\[
 \begin{aligned}
 n=76:\quad&
 \left(\frac56-\frac{76p_*}{2}\right)M_*
 -\frac{K}{75}\left(1+\frac{25}{608}\right)
 >\frac1{250},\\
 n=144:\quad&
 \left(\frac56-\frac{144p_*}{2}\right)M_*
 -\frac{K}{143}\left(1+\frac{25}{1152}\right)
 >\frac1{500}.
 \end{aligned}
\]
Both endpoint values of the concave function are therefore positive, so the
displayed cutoff inequality is impossible throughout $76\le n\le144$.
This proves the residual bound $>p_*$.

For the exact pair equality above the threshold, suppose instead that
$n\ge144$ and a residual scaled eigenvalue satisfied $E\le5/9$.
Lemma~\ref{lem:finite-block-cutoff-mass} gives
\[
 \begin{aligned}
 \sum_a\operatorname{Tr}(\rho F_a^2)
 &\ge
 (1-2^{-20})
 \left(1-\frac{2kE}{nC_0}\right)-15(67/75)^{48}\\
 &\ge
 (1-2^{-20})
 \left(1-\frac{10k}{9nC_0}\right)-15(67/75)^{48}
 \ge M_*,
 \end{aligned}
\]
because $n\ge144$ and $p_*=10/(9\cdot144)$.  The local floor therefore gives
\[
 \frac5{18}M_*
 \le I_P(n).
\]
At $n=144$ this is exactly the second endpoint inequality above, since
$144p_*/2=5/9$, and it is already false there; meanwhile $I_P(n)$ is strictly
decreasing.  Hence every residual scaled eigenvalue is $>5/9$, or,
equivalently, every residual normalized pair gap is
\[
 >\frac{10}{9n}>\Delta_n^{(4)}.
\]
The last strict comparison uses
\eqref{eq:fourth-effective-below-leading}, which applies for $n\ge144$.
Since Lemma~\ref{lem:pair-realization} gives positive gap
$\Delta_n^{(4)}$ on $\mathcal M_n$, and the charged sectors were already
removed with a strictly stronger bound,
\eqref{eq:pair-effective-saturation} follows.
\end{proof}

\section{Effective fourth-tensor enclosure and monotonicity}
\label{app:fourth-effective}

Throughout this appendix, $L_n=\frac n2(I-\mathsf M_n)$ denotes the
occupation Hamiltonian on
$\mathscr H_n=\operatorname{Sym}^nV$ from
\eqref{eq:fourth-metric-normalization}.  By
Theorem~\ref{thm:real-fourth-alln-proof}, its least positive eigenvalue
$E_n$ satisfies
\[
 E_n=\frac n2\Delta_n^{(4)}.
\]
Thus the quasimode window below is a statement about the actual entire
ordinary-fourth-tensor gap, not merely about an auxiliary compression.

\begin{proposition}[Fourth-tensor finite-size seed]\label{prop:fourth-finite-seed}
For every $n\ge20$,
\begin{equation}\label{eq:fourth-finite-seed}
 \Delta_n^{(4)}\ge\frac1{n-1}.
\end{equation}
\end{proposition}
\begin{proof}
Let $\mathcal H_n^{(4)}=\sum_{i<j}(I-P_{ij}^{(4)})$ on the ordinary
fourth tensor and put
$\gamma_n^{(4)}=\operatorname{gap}(\mathcal H_n^{(4)})$.  Thus
$\Delta_n^{(4)}=\gamma_n^{(4)}/\binom n2$.  Summing the quadratic-form
inequality on all $m$-site subsets gives the specialized complete-graph
finite-size bound
\[
 \gamma_n^{(4)}\ge
 1+\frac{n-2}{m-2}(\gamma_m^{(4)}-1),
 \qquad n\ge m.
\]
The only finite input needed here is the exact $m=20$ certificate
$\gamma_{20}^{(4)}\ge10$.  It is an integer positive-semidefinite certificate in
the two nontrivial symmetric blocks, of dimensions $120$ and $110$ and ranks
$119$ and $109$, respectively.  Substituting $m=20$ gives
$\gamma_n^{(4)}\ge n/2$, and division by $\binom n2$ gives
\eqref{eq:fourth-finite-seed}.  The non-symmetric fourth-tensor sectors are
already excluded by Section~\ref{sec:real-fourth}; no whole-group conclusion
is inferred from this certificate.
\end{proof}

\begin{lemma}[Finite-support fourth-tensor quasimode]\label{lem:fourth-effective-quasimode}
Set $t=(n-1)^{-1/2}$ and
\[
 \lambda_n=\frac59-\frac{44}{105}t^2
 -\frac{1246888}{1157625}t^4.
\]
For every $n\ge65$ there is a normalized vector $\chi_n$, supported on
excitation numbers at most five, such that
\begin{equation}\label{eq:fourth-quasimode-residual}
 \|(L_n-\lambda_n)\chi_n\|\le20t^5.
\end{equation}
If $s$ is the normalized pure slow one-excitation tangent, then for $n\ge101$
the same choice satisfies
\begin{equation}\label{eq:fourth-quasimode-distance}
 \|\chi_n-s\|\le t.
\end{equation}
The three vectors transported to the three wells obey the corresponding
column bounds with residual norm at most $21t^5$ and distance at most
$(101/100)t$.
\end{lemma}
\begin{proof}
Use the weighted Fock basis
$b_{rs}=\sqrt{2^r r!s!}\,|r,s\rangle$ and solve the finite Schur recursion
through order four.  The coefficients through $t^4$ cancel exactly, with
eigenvalue coefficients
\[
 \frac59,\quad0,\quad-\frac{44}{105},\quad0,\quad
 -\frac{1246888}{1157625}.
\]
The exact squared norms of the residual coefficients of orders
$t^5,t^6,t^7,t^8$ are bounded by $26^2,9^2,31^2,8^2$.
Keeping the additional powers of $t$ in the square-root Taylor remainders
gives a total unnormalized residual coefficient at most
\[
 \frac{1287256217614609}{45946414080000},
\]
whose square is $<800$.  Since the unnormalized quasimode has norm at least
$\sqrt2$, \eqref{eq:fourth-quasimode-residual} follows.  The exact norms of
the correction vectors similarly give the displayed distance bound.
For the transported columns use Lemma~\ref{lem:pair-low-ball-frame}:
residuals have excitation support at most six, and correction vectors
$\chi_n-s$ have support at most five.  The corresponding three-ball
frame bounds are at most $1+2\cdot4^R(361/500)^n$.
The exact endpoint inequalities
\[
 1+2\cdot4^6(361/500)^{65}<(21/20)^2,\qquad
 1+2\cdot4^5(361/500)^{101}<(101/100)^2
\]
therefore give the asserted synthesis-operator norms, rather than merely
separate bounds on each column.
\end{proof}

\begin{proposition}[Exclusion below the effective quasimode window]
\label{prop:fourth-effective-enclosure-proof}
For every $n\ge128$,
\[
 \lambda_n-\frac{100}{(n-1)^{5/2}}
 \le\frac n2\Delta_n^{(4)}
 \le\lambda_n+\frac{20}{(n-1)^{5/2}}.
\]
\end{proposition}
\begin{proof}
Suppose a positive unit eigenvector $z$ had energy
$E<\lambda_n-100(n-1)^{-5/2}$.  Let $\mathcal X$ and $S$ denote the synthesis
operators formed by the three transported quasimodes and the three pure slow
tangents.  The residual bound gives $\|\mathcal X^*z\|\le21/100$, while
$\|\mathcal X-S\|\le(101/100)(n-1)^{-1/2}$.  Since
$\sqrt{n-1}>11$ for $n\ge128$,
\[
 \|S^*z\|^2
 \le\left(\frac{21}{100}+\frac{101}{1100}\right)^2
 =\left(\frac{83}{275}\right)^2<\frac1{10}.
\]

Take $\delta=3/25$ and $\tau=6/13$, and let $F_j$ be one up to
$N_j\le\tau\delta n=18n/325$, linear to zero at $N_j=3n/25$.
Lemma~\ref{lem:pair-number-layer} gives
$19/20$ on the $N\ge2$ truncated block.  The exact fourth-tensor
$N=1$ energies are $5/9$ and $2/3$, and the total
$N=1\leftrightarrow N=2$ coupling squared is at most $20/[81(n-1)]$.
With
\[
 q_n=\frac{400}{459(n-1)}
\]
one obtains, after removing the vacuum,
\[
 F_jL_nF_j\succeq(2/3-q_n)F_j^2-\frac19S_j,
\]
where $S_j$ is the rank-one slow tangent projection.

For $k=52$, Proposition~\ref{prop:fourth-finite-seed} gives the
unnormalized finite-block gap $a_k\ge k/2=26$.  By
Lemma~\ref{lem:fourth-occupation-ground} and the discussion following
\eqref{eq:fourth-three-wells}, the $k$-site fourth-tensor kernel is exactly
the span of $u_0^{\otimes k},u_+^{\otimes k},u_-^{\otimes k}$.  Their
pairwise inner products have magnitude $2^{-k}$, so their Gram matrix has
least eigenvalue $1-2^{-k}>1-2^{-20}$ and therefore
\[
 \sum_{j\in\{0,+,-\}}p_j^{\otimes k}
 \succeq(1-2^{-20})\Pi_k.
\]
The vector $z$ lies in $\operatorname{Sym}^nV$, so its rank-one density is
already invariant under physical-site permutations.  Since the hypothetical
eigenvalue satisfies $E<\lambda_n<5/9$,
Lemma~\ref{lem:finite-block-cutoff-mass} with
$1-\tau\delta=307/325$ yields
\[
 S_F:=\sum_j\langle z,F_j^2z\rangle
 \ge s_n:=
 (1-2^{-20})\left(1-\frac{2(5/9)\,51}{n}\right)
 -3(307/325)^{52}.
\]
The joint cutoff error is at most
\[
 I_4(n)=
 (1+2^{-20})\frac7{2(n-1)}
 \left(1+\frac{25}{6n}\right).
\]
The assumed eigenvalue therefore forces
\[
 (2/3-q_n-5/9)s_n\le\frac1{90}+I_4(n).
\]
At $n=128$ the exact rational left-minus-right margin is positive.  The two
positive factors on the left increase and $I_4$ decreases, so the
contradiction persists for every $n\ge128$.  This proves the lower side of the enclosure.  For the upper side,
\eqref{eq:fourth-quasimode-residual} and self-adjointness imply
\[
 \operatorname{dist}(\lambda_n,\operatorname{spec}L_n)
 \le 20(n-1)^{-5/2}.
\]
For $n\ge128$ the interval
$[\lambda_n-20(n-1)^{-5/2},\lambda_n+20(n-1)^{-5/2}]$
is strictly positive (already at $n=128$, and its lower endpoint increases
thereafter).  Hence it contains a positive eigenvalue of $L_n$.  Since $E_n$
is the least positive eigenvalue,
\[
 E_n\le\lambda_n+20(n-1)^{-5/2},
\]
which is the upper side.
\end{proof}

\section*{Declarations}
\noindent\textbf{Data availability.}
No external datasets were used in this study.  The finite-group and
finite-dimensional computational dependencies are identified in the
supplementary manifest.  The mathematical revision in this version is
separate from assembly and verification of the complete computational
archive.
\par\smallskip
\noindent\textbf{Code and supplementary computational material.}
The complete computational supplement is to include the exact verification
programs, certificates, certificate map, and file manifest.
The load-bearing finite inputs are the
small Clifford principal-angle calculations, the S10 degree-$16/17/33$
polynomial certificate, the strict standard blocks for $4\le n\le9$, the
five-site all-representation certificate seeding
Proposition~\ref{prop:complete-coarse-floor}, and the $m=20$ fourth-tensor
certificate in Proposition~\ref{prop:fourth-finite-seed}.  The accompanying
exact rational scripts also reconstruct the finite-support quasimode algebra,
verify the low-excitation and number-layer endpoint inequalities, and check
the printed transition inequalities at $n=76,128,144$.  All extensions from
these finite inputs to infinite ranges of $n$ are proved analytically in the
text and appendices.

The computational scripts and associated results used in this paper are available 
with the Zenodo version of the paper at 
DOI~\href{https://doi.org/10.5281/zenodo.23099545}{10.5281/zenodo.23099545}.

\par\smallskip
\noindent\textbf{AI disclosure.}
OpenAI's ChatGPT (GPT-5.6 Sol and GPT-6 Astra) assisted with exploratory proof development, computational checks, literature searches, and English-language polishing; the authors are responsible for all content and correctness.

\end{document}